\documentclass[11pt,fleqn,a4paper]{article} 

\usepackage{amsmath,amssymb,amsthm,amsfonts} 
\usepackage[mathscr]{eucal}
\usepackage{graphicx}
\usepackage{caption}
\usepackage{subcaption}

\usepackage{hyperref}
\hypersetup{colorlinks, linkcolor=blue, citecolor=blue, urlcolor=blue}

\allowdisplaybreaks

\makeatletter
\long\def\@makecaption#1#2{%
  \vskip\abovecaptionskip\footnotesize
  \sbox\@tempboxa{#1. #2}%
  \ifdim \wd\@tempboxa >\hsize
    #1. #2\par
  \else
    \global \@minipagefalse
    \hb@xt@\hsize{\hfil\box\@tempboxa\hfil}%
  \fi
  \vskip\belowcaptionskip}
\makeatother

\newcommand{\p}{\partial}

\newcommand{\rank}{\mathop{\rm rank}\nolimits}
\newcommand{\im}{\mathop{\rm im}\nolimits}
\newcommand{\const}{{\rm const}}

\newlength{\mylength}
\newcommand{\solution}{\hspace*{-\mylength}\bullet\quad}
\newcommand{\solutionRedEq}{\hspace*{-\mylength}\circ\quad}

\newtheorem{theorem}{Theorem}
\newtheorem{lemma}[theorem]{Lemma}
\newtheorem{corollary}[theorem]{Corollary}

{\theoremstyle{definition}
	
	\newtheorem{remark}[theorem]{Remark}
    \newtheorem*{notation}{Notation}
}

\newcommand{\todo}[1][\null]{\ensuremath{\clubsuit}}

\newcommand{\noprint}[1]{}
\newcommand{\lsemioplus}{\mathbin{\mbox{$\lefteqn{\hspace{.77ex}\rule{.4pt}{1.2ex}}{\in}$}}}

\usepackage{xcolor}

\begin{document}
\par\noindent {\LARGE\bf
Symmetry analysis and exact solutions\\
of multi-layer quasi-geostrophic problem
\par}

\vspace{5.5mm}\par\noindent{\large
Serhii D. Koval$^{\dag\ddag}$, Alex Bihlo$^\dag$ and Roman O. Popovych$^{\S\ddag}$
}

\vspace{5.5mm}\par\noindent{\it\small
$^\dag$Department of Mathematics and Statistics, Memorial University of Newfoundland,\\
$\phantom{^\dag}$\,St.\ John's (NL) A1C 5S7, Canada
\par}

\vspace{2mm}\par\noindent{\it\small
$^\S$\,Mathematical Institute, Silesian University in Opava, Na Rybn\'\i{}\v{c}ku 1, 746 01 Opava, Czech Republic
\par}

\vspace{2mm}\par\noindent{\it\small
$^\ddag$Institute of Mathematics of NAS of Ukraine, 3 Tereshchenkivska Str., 01024 Kyiv, Ukraine
\par}

\vspace{5mm}\par\noindent
E-mails:
skoval@mun.ca, abihlo@mun.ca, rop@imath.kiev.ua
	
\vspace{6mm}\par\noindent\hspace*{10mm}\parbox{140mm}{\small
We carry out an extended symmetry analysis of the multi-layer quasi-geostrophic problem.
This model is given by a system of an arbitrary number of coupled barotropic vorticity equations.
Conservation laws and a Hamiltonian structure for the general case of the model
are correctly described for the first time.
Using original methods, we compute the maximal Lie invariance algebra
and the complete point-symmetry pseudogroup of the model.
After classifying one- and two-dimensional subalgebras of the Lie invariance algebra,
we exhaustively study codimension-one, -two and -three Lie reductions.
Notably, among invariant submodels of the original nonlinear model,
we obtain uncoupled systems of well-known linear equations, including
the Helmholtz, modified Helmholtz, Laplace, Klein--Gordon, Whittaker, Bessel
and linearized Benjamin--Bona--Mahony equations.
Integration of these systems significantly depends on spectral properties
of the model's vertical coupling matrix, which we also revisit in detail.
As a result, we construct wide families of exact solutions,
including rediscovered representations of stationary and travelling baroclinic Rossby waves,
coherent baroclinic eddies, hetons and localized dipolar vortices.
We illustrate the physical relevance of obtained solutions
using real-world geophysical data for a three-layer ocean model.
}\par\vspace{4mm}
	
\noprint{


76-XX Fluid mechanics
  76Mxx Basic methods in fluid mechanics
    76M60 Symmetry analysis, Lie group and Lie algebra methods applied to problems in fluid mechanics

17-XX  Nonassociative rings and algebras
  17Bxx Lie algebras and Lie superalgebras
    17B81 Applications of Lie (super)algebras to physics, etc.

35-XX Partial differential equations
  35Axx General topics
    35A30 Geometric theory, characteristics, transformations [See also 58J70, 58J72]
  35Bxx Qualitative properties of solutions
    35B06 Symmetries, invariants, etc.
  35Cxx Representations of solutions
    35C05 Solutions in closed form
    35C06 Self-similar solutions
    35C07 Traveling wave solutions

15-XX  Linear and multilinear algebra; matrix theory
  15Bxx  Special matrices
    15B05  Toeplitz, Cauchy, and related matrices

37-XX Dynamical systems and ergodic theory
  37Kxx  Dynamical system aspects of infinite-dimensional Hamiltonian and Lagrangian systems
 	37K06  General theory of infinite-dimensional Hamiltonian and Lagrangian systems, Hamiltonian and Lagrangian structures, symmetries, conservation laws
}

\tableofcontents

\section{Introduction}

In ocean or atmosphere modeling,
the amount of literature devoted to numerical investigation of dynamics of fluid motion
greatly outweighs that focused on finding explicit exact closed-form solutions.
This trend is certainly not surprising,
given that the underlying systems of partial differential equations
are nonlinear and strongly coupled.
However, even within numerical analysis,
exact solutions remain important as they provide a reference point
for developing and verifying finite difference schemes or other numerical techniques.
Among various mathematical approaches,
Lie group analysis stands as one of the most
powerful frameworks for constructing exact solutions
of systems of nonlinear partial differential equations.

Based on analysis of horizontal and vertical characteristic length scales,
the complete set of governing hydro-thermodynamical equations for meteorology and oceanology
can be reduced to a hierarchy of simplified models.
Each model from this hierarchy accounts only for selected phenomena in dynamics of atmosphere or ocean.

The most elementary model is the famous Lorenz system,
which is a system of three ordinary differential equations
that captures nonlinear chaos but sacrifices spatial resolution~\cite{lore1963a}.
The next, more complicated model is the barotropic vorticity equation,
which simplifies the fluid to a single constant-density layer,
while still allowing for spatial dynamics.
Results on Lie symmetries and exact solutions of the barotropic vorticity equation
were collected in~\cite[Section~9.6]{CRChandbook1995V2};
this equation was called the equation of geopotential forecast therein.
Its point symmetry group was computed in~\cite{bihl2012a},
see also~\cite{card2013a} for similar results on the vorticity equation on the rotating sphere.

Analogously to the barotropic vorticity equation,
the shallow water equations model a single constant-density fluid layer,
but take into account depth variations and gravity wave dynamics for cases where
the horizontal scale is significantly greater than the vertical one.
The shallow water equations with flat bottom topography coincide with
a particular case of the equations of motion of polytropic gas with the parameter $\gamma=2$;
see~\cite{ibra1985A} for Lie symmetries, some exact solutions and some conservation laws of these equations.
Lie symmetries and zeroth-order conservation laws
of the shallow water equations with variable bottom topography were classified
in~\cite{bihl2020a} and in~\cite{bihl2020b}, respectively.

Increasing the precision of the approximation to the multi-layer quasi-geostrophic model,
one considers the fluid as stacked layers of constant densities
allowing for simulating baroclinic instability and synoptic-scale weather patterns.
The classical symmetry analysis for the two-layer model was carried out in~\cite{bihl2011a},
while its Lax representation and its conservation laws were constructed in~\cite{moro2023a}.

The primitive equations consist of three main sets of balance equations:
the continuity, momentum and thermal energy equations.
This is a standard model for global forecasting.
The partial classical symmetry analysis of the primitive equations in~\cite{card2021a}
included computing their maximal Lie invariance algebra, point-symmetry pseudogroup
and some families of Lie invariant solutions.
The most complicated and complete model for meteorology and oceanology
is given by the full nonhydrostatic Euler equations,
which resolves vertical accelerations and acoustic waves to account for small-scale,
high-energy phenomena like deep convection and turbulence.

For other relevant results on application of Lie group analysis in fluid dynamics,
see \cite[Chapters~9--11]{CRChandbook1995V2},
\cite{andr1998A,chev2014a,ibra2011A,kelb2013a,levi1989a,popo2002b,szat2014a}
and references therein.

The purpose of the present paper is to carry out extended Lie symmetry analysis
of the quasi-geostrophic model with an arbitrary number of layers,
essentially generalizing the study of the case of two layers
initiated in~\cite{bihl2011a} and continued in~\cite{moro2023a}.
This includes the construction of conservation laws and a Hamiltonian structure of the model,
the computation of its maximal Lie invariance algebra~$\mathfrak g$ and its complete pseudogroup~$G$ of point symmetries,
the classification of one- and two-dimensional subalgebras of the algebra~$\mathfrak g$,
the exhaustive study of codimension-one, -two and -three Lie reductions and hidden Lie symmetries
and finding exact invariant solutions of the model.
Beyond its clear physical relevance to large-scale oceanic and atmospheric dynamics,
the model possesses distinguishing symmetry features that have never been represented
in the literature on Lie group analysis of differential equations.
The multi-layer quasi-geostrophic model is governed
by a system of $m\in\mathbb N$ coupled barotropic vorticity equations,
where the coupling between layers is given by a specific tridiagonal matrix.
To the best of our knowledge,
systems of nonlinear partial differential equations involving an arbitrary number of coupled equations have not yet been
considered in the literature from the perspective of Lie group analysis.
The arbitrariness of $m$ in particular makes it impossible to directly apply specialized computer algebra packages
for finding the algebra~$\mathfrak g$.

This paper is organized as follows.
In Section~\ref{sec:mLaysMod}, we provide a comprehensive description of
the multi-layer quasi-geostrophic problem,
which is given by the system of equations~\eqref{eq:mLaysMod}.
In particular, we present its vector representation~\eqref{eq:mLaysModMatrixForm}
and study structural and spectral properties of the vertical coupling matrix~$\mathsf F$.
We show that the matrix~$\mathsf F$ has rank $m-1$,
is diagonalizable and has pairwise distinct negative eigenvalues,
except one eigenvalue, which is zero.
These results are of great relevance to the physical interpretation of the model
since even distinguishing between barotropic and baroclinic modes is based on
the $\mathsf F$-invariant subspace decomposition of the underlying space
for values of the tuple of layer stream functions.
We also present real geophysical data from \cite{nauw2004a} for parameters
of the model of a stratified ocean with three layers,
which are later used for physically relevant illustration of found invariant solutions.

Furthermore, in Section~\ref{sec:HamiltonianStructure}, we correctly construct for the first time
families of conservation laws of the general case of the model~\eqref{eq:mLaysMod}
and its Hamiltonian structure.

Section~\ref{sec:LieInvAlg} is devoted to computing the maximal Lie invariance algebra~$\mathfrak g$
of the multi-layer quasi-geostrophic model~\eqref{eq:mLaysMod}.
This computation is made possible due to two original tricks.
The first trick is to formally change the independent spatial variables $(x,y)$
to the complex conjugated ones $z=x+{\rm i}y$ and $\bar z=x-{\rm i}y$,
which allows us to introduce a natural ordering for jet variables
and reduce the size of involved expressions.
The second trick is to consider the class of systems of general form~\eqref{eq:mLaysMod}
as the lowest member in a chain of nested classes of systems of differential equations.
In this way, the derivation of determining equations for Lie symmetries
is simplified via the successive application of the infinitesimal invariance criterion
to the classes in the chain starting from the largest one.
Since a general system from a superclass has less constraints on its arbitrary elements
than one from a proper subclass,
it is easier to treat the corresponding invariance condition
and split it only with respect to higher-order jet variables
that are not involved in the arbitrary elements of the superclass as their arguments.
At the same time, the determining equations for the superclass are a subsystem
of the analogous system for the subclass.
Higher-level classes in the hierarchy give more important determining equations;
taking them into account significantly simplifies
deriving the complete systems of determining equations in lower levels.

In Section~\ref{sec:mLaysPointSymGroup}, using the megaideal-based version of the algebraic method,
we compute the pseudogroup~$G$ of point symmetries of the system~\eqref{eq:mLaysMod}
and thus exhaustively classify its discrete elements.

Then in Section~\ref{sec:GroupClassification},
we consider the collection of systems of the form~\eqref{eq:mLaysMod},
where the vertical coupling matrix~$\mathsf F$ and the Rossby parameter~$\beta$
run through the sets of admitted values, as a class $\mathcal M$ of systems of differential equations.
Based on the results of Sections~\ref{sec:LieInvAlg} and~\ref{sec:mLaysPointSymGroup},
we compute the equivalence groupoid, the equivalence group and the equivalence algebra
of this class and carry out its exhaustive group classification.
In particular, we show that the class $\mathcal M$ is normalized in the usual sense
and describe its generalized equivalence group.

Towards the study of Lie reductions of the multi-layer quasi-geostrophic problem~\eqref{eq:mLaysMod},
in Section~\ref{sec:Subalgebras}
we classify one- and two-dimensional subalgebras of the Lie algebra~$\mathfrak g$
with respect to the action of the pseudogroup~$G$.
We also show that codimension-three Lie reductions of~\eqref{eq:mLaysMod}
lead to no interesting exact solutions of~\eqref{eq:mLaysMod}.

Section~\ref{sec:CodimOneLieReds} is devoted to the systematic investigation
of codimension-one Lie reductions of~\eqref{eq:mLaysMod}.
We carry out the group classification of each class of the reduced systems,
thus computing their maximal Lie invariance algebras,
and determine their induced Lie symmetries.
This allows us to exhaustively describe hidden Lie symmetries of the original nonlinear system~\eqref{eq:mLaysMod}.
We thoroughly study the cases when the codimension-one reduced systems are linear and completely decoupled
(at least under additional differential constraints).
In this setting, we derive solutions of the multi-layer quasi-geostrophic problem~\eqref{eq:mLaysMod}
expressed in terms of solutions of decoupled systems of well-known linear equations, including
Helmholtz, modified Helmholtz, Laplace, Klein--Gordon, Whittaker, Bessel and linearized Benjamin--Bona--Mahony equations.
Among them, the most physically prominent solutions are represented in terms of {\it Herglotz wave functions};
such functions associated with finite measures on the circle constitute
the class of entirely bounded solutions of the Helmholtz equations.
Particular cases of this representation include baroclinic Rossby waves,
coherent baroclinic eddies, coherent hetons and their superpositions.
Furthermore, by partitioning the domain and merging different solutions
(even those originating from different reduced systems),
we reconstruct representations for dipolar vortices, also known in the oceanological literature as modons.
The illustrations of the above physically relevant solutions using real-world geophysical data from~\cite{nauw2004a}
for a three-layered ocean are presented as well.
We also consider the generically coupled case of codimension-one reduced systems
and further reduce them using appropriate one-dimensional subalgebras of their Lie invariance algebras,
which, in some cases, allowed us to construct explicit exact solutions.

The study of codimension-two Lie reductions of the system~\eqref{eq:mLaysMod}
is carried out in Section~\ref{sec:CodimTwoLieReds}.
It turns out that each of them is just a two-step reduction,
where the first step is a codimension-one Lie reduction of~\eqref{eq:mLaysMod}
and the second step is the further Lie reduction
of the obtained reduced systems with respect to their induced Lie symmetry.
Nevertheless, we present results on codimension-two reductions for the following two reasons.
First and foremost, the codimension-two reduced systems
are systems of linear ordinary differential equations (of order at most three)
with constant coefficients, and thus they can be completely integrated explicitly or at least in quadratures.
At the same time, it is optimal not to consider them as two-step reductions~\cite[Section~2]{vinn2024a}.
Given that the codimension-one reduced systems in Section~\ref{sec:CodimOneLieReds} are integrated
in terms of elementary or special functions only for certain special cases or specific values of subalgebra parameters,
the codimension-two Lie reduction may result in solutions that are not found in Section~\ref{sec:CodimOneLieReds},
and this is indeed the case for integrating three out of four families of corresponding reduced systems.
The second reason is that including the results on codimension-two reductions in the paper
maintains the systematic and comprehensive nature of the present symmetry study
and provides a complete solution to the problem of Lie reductions
for the multi-layer quasi-geostrophic problem~\eqref{eq:mLaysMod}.

We summarize the results of the present paper
and outline research perspectives on the multi-layer quasi-geostrophic problem~\eqref{eq:mLaysMod} in Section~\ref{sec:Conclusion}.

For readers' convenience,
the constructed exact solutions of the multi-layer quasi-geostrophic problem~\eqref{eq:mLaysMod}
are marked by the bullet symbol~$\bullet$\,
and the constructed exact solutions of the reduced systems of~\eqref{eq:mLaysMod}
are marked by the circle symbol~$\circ$\,.

\section{Multi-layer quasi-geostrophic problem}\label{sec:mLaysMod}

The multi-layer quasi-geostrophic problem describes the dynamics of a stratified incompressible fluid with~$m$ layers
of constant densities $\rho_1<\dots<\rho_m$,
stacked according to increasing density (i.e., the density of the top layer is~$\rho_1$),
in the ``rigid-lid'' and flat-bottom approximation
(i.e., the upper surface of the fluid is assumed to be fixed \cite[Section~16.A.2]{vall2006A} and the variation of bottom is negligible).

\subsection{Mathematical model}

The quasi-geostrophic approximation assumes that the vertical component of the velocity field is negligibly small
and in the horizontal direction the motion satisfies
a system of~$m$ ($m\geqslant2$) coupled vorticity equations, that is, the third-order partial differential equations
\begin{gather}\label{eq:mLaysMod}
\begin{split}
&q^i_t+\{\psi^i,q^i\}=0,\\
&q^i:=\psi^i_{xx}+\psi^i_{yy}+f_{i,i-1}(\psi^{i-1}-\psi^i)-f_{i,i+1}(\psi^i-\psi^{i+1})+\beta y,\quad
i=1,\dots, m.
\end{split}
\end{gather}
Here $\psi^i=\psi^i(t,x,y)$ and $q^i=q^i(t,x,y)$
are the stream function and the quasi-geostrophic potential vorticity for the $i$th layer, respectively,
and their Poisson bracket is given by their Jacobian with respect to~$(x,y)$, $\{\psi^i,q^i\}:=\psi^i_xq^i_y-\psi^i_yq^i_x$.
We define the stream function as in meteorology and oceanography
such that the tuple $(-\psi^i_y,\psi^i_x)$ is the velocity field on the $i$th~layer.
The $i$th layer of the stratified fluid is characterized by its mean thickness $H_i>0$,
the fluid density~$\rho_i$ in this layer
and the corresponding reduced gravity $g_i'=g(\rho_{i+1}-\rho_i)/\rho_0>0$, ${i=1,\dots,m-1}$,
where~$\rho_0$ denotes the reference density%
\footnote{%
In practice, there are several ways to choose the reference density depending on physical context.
An option is to take the weighted average of the layer densities with weights given by the layer thicknesses.
In some ocean models, the density of the densest (bottom) layer is used as the reference one.
Alternatively, depending on the medium under study,
one may use a standard constant like
the density of fresh water $1000\,{\rm kg}/{\rm m}^3$,
the density of salt water $1025\,{\rm kg}/{\rm m}^3$ or
the air density $1.225\,{\rm kg}/{\rm m}^3$
at the sea level under the International Standard Atmosphere (ISA) model.
Note that within the multi-layer quasi-geostrophic framework, which is based on the Boussinesq approximation,
the relative variations of the layer densities in comparison with~$\rho_0$ are usually small ($\approx 1\%$),
hence the results are quite insensitive to the specific choice of $\rho_0$.
The differences between the individual layer densities~$\rho_i$
are essential in the vertical direction in the interactions between neighboring layers because of gravity,
which is captured via the definition of the reduced gravities or, equivalently, of the vertical coupling coefficients.
In the horizontal direction, the density differences affect the inertia of the fluid negligibly,
and thus the layer densities~$\rho_i$ are then replaced by the reference density~$\rho_0$.
}
of the fluid.
The coupling between $i$th and $(i-1)$th (resp.\ $(i+1)$th) layer
is provided by the term $f_{i,i-1}(\psi^i-\psi^{i-1})$ with $f_{i,i-1}:=f_0^{\,2}/(H_ig_{i-1}')>0$
(resp.\ $f_{i,i+1}(\psi^{i+1}-\psi^i)$ with $f_{i,i+1}:=f_0^{\,2}/(H_ig_i')>0$),
where $f_{1,0}:=0$ and $f_{m,m+1}:=0$.
Therefore, for each~$i$ both the constants~$f_{i,i-1}$ and~$f_{i,i+1}$ are positive real numbers,
except for~$f_{1,0}$ and~$f_{m,m+1}$, which are zero.
Here, $f_0$ and $\beta$ are the coefficients in the beta-plane approximation
of the Coriolis parameter $f=2\Omega\sin\vartheta$ about the latitude $\vartheta=\vartheta_0$, $f=f_0+\beta y$.
More specifically, the Coriolis parameter at $\vartheta_0$ is $f_0=2\Omega\sin\vartheta_0$,
the Rossby parameter at $\vartheta_0$ is $\beta=2\Omega\cos\vartheta_0/R>0$, $y=R(\vartheta-\vartheta_0)$,
where~$R$ is the Earth's radius and~$\Omega$ is the angular rotation rate of the Earth.
It is worth to note that constants $f_{i,i\pm1}$ and $H_i$ satisfy a nontrivial relation,
$H_if_{i,i+1}=H_{i+1}f_{i+1,i}$.
See \cite[Section~16.A.2]{vall2006A} and \cite[Section~6.16]{pedo1987A}
for the more detailed description of the model
as well as Fig.~\ref{fig:3LayerModel} for graphic illustration
in the particular case of three-layer models, $m=3$.

\begin{figure}[!ht]
\centering
\includegraphics[width=\linewidth]{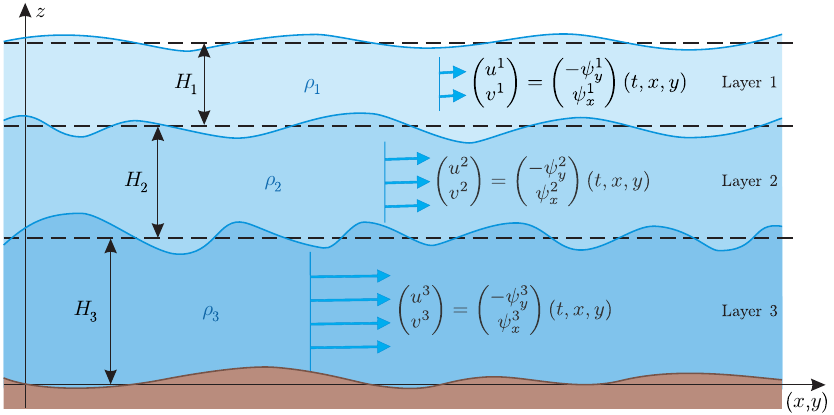}
\caption{Three-layer model}
\label{fig:3LayerModel}
\end{figure}

\begin{notation}
Throughout the paper,
the Latin indices $i$, $j$ and $k$ run through the sets $\{1,\dots,m\}$ unless the other range is specified.
The summation over the repeated indices whenever they appear in the different multipliers is understood
unless the other is specified, except for the index~$i$.
Capital sans serif letters denote, usually $m\times m$, matrices.
\end{notation}

The expanded form of the system~\eqref{eq:mLaysMod} is
\begin{gather*}
\begin{split}
&\psi^i_{txx}+\psi^i_{tyy}+\psi^i_x(\psi^i_{xxy}+\psi^i_{yyy})-\psi^i_y(\psi^i_{xxx}+\psi^i_{yyx})
\\
&\qquad{}+f_{i,i-1}\big(\psi^{i-1}_t-\psi^i_t+\psi^i_x(\psi^{i-1}_y-\psi^i_y)-\psi^i_y(\psi^{i-1}_x-\psi^i_x)\big)
\\
&\qquad{}-f_{i,i+1}\big(\psi^i_t-\psi^{i+1}_t+\psi^i_x(\psi^i_y-\psi^{i+1}_y)-\psi^i_y(\psi^i_x-\psi^{i+1}_x)\big)+\beta\psi^i_x=0.
\end{split}
\end{gather*}
More rigorously, this represents the class~$\mathcal M$ of systems of
$m$ third-order partial differential equations for $m$ unknown functions $\psi:=(\psi^1,\dots,\psi^m)^{\mathsf T}$
in the three independent variables $(t,x,y)$,
where the constant parameters $f_{i,i-1}$, $f_{i,i+1}$,
and~$\beta$ plays the role of arbitrary elements of the class and are positive,
except for the zero artificial parameters~$f_{1,0}$ and~$f_{m,m+1}$.
In other words, the arbitrary-element tuple of the class~$\mathcal M$ is
\[\theta:=(f_{i+1,i},f_{i,i+1},i=1,\dots,m-1,\beta),\]
and the system from~$\mathcal M$ that corresponds to a value~$\theta$
of the arbitrary-element tuple is denoted by~$\mathcal M_\theta$.
This is why all problems of finding symmetry-like objects associated with the system~\eqref{eq:mLaysMod}
should be considered classification problems for such objects
within the class of systems of the form~\eqref{eq:mLaysMod}.
Nevertheless, it turns out that all these systems have the same maximal Lie invariance algebra
and the point symmetry pseudogroups, see Theorems~\ref{thm:mLaysLieInvAlg} and~\ref{thm:mLaysPointSymGroup},
respectively, which leads to the degeneration of the corresponding classification problems
to the regular computations of the above structures as for a single system.
Note that the physical constraint of the parameters' positivity
can be weakened in the course of these computations
in Sections~\ref{sec:LieInvAlg} and~\ref{sec:mLaysPointSymGroup}
to the condition of being nonzero.

\subsection{Vector representation of the model}\label{sec:VectorFormOfSyst}

One more convenient representation of the system~\eqref{eq:mLaysMod},
in particular for studying properties of invariant submodels,
is that in the matrix form:
\begin{gather}\label{eq:mLaysModMatrixForm}
\begin{split}
&q_t+\{\psi,q\}=0,\\
&q:=\psi_{xx}+\psi_{yy}+{\mathsf F}\psi+\beta y\bar1,
\end{split}
\end{gather}
where $\psi:=(\psi^1,\dots,\psi^m)^{\mathsf T}$, $q:=(q^1,\dots,q^m)^{\mathsf T}$,
the Poisson bracket $\{\psi,q\}$ is defined componentwise,
$\{\psi,q\}:=(\{\psi^i,q^i\})^{\mathsf T}=\psi_x\odot q_y-\psi_y\odot q_x$,
``$\odot$'' is the Hadamard (componentwise) product of matrices,
the vertical coupling matrix
$\mathsf F:=(f_{ij})$ is a constant $m$-by-$m$ tridiagonal matrix with $f_{ii}:=-f_{i,i-1}-f_{i,i+1}$,
 $f_{10}:=0$, $f_{m,m+1}:=0$ and $f_{ij}:=0$ if $|i-j|>1$,
\begin{gather*}
\mathsf F:=\begin{pmatrix}
-f_{12} & f_{12}        & 0             & \cdots&0           & 0                     & 0\\
 f_{21} &-f_{21}-f_{23} & f_{23}        & \cdots&0           & 0                     & 0\\
\vdots  & \vdots        & \vdots        &\ddots &\vdots      &\vdots                 &\vdots\\
0       & 0             & 0             &\cdots &f_{m-1,m-2} &-f_{m-1,m-2}-f_{m-1,m} &f_{m-1,m}\\
0       & 0             & 0             &\cdots &0            & f_{m,m-1}             &-f_{m,m-1}
\end{pmatrix},
\end{gather*}
and $\bar 1$ denotes the all-ones $m$-column,
$\bar 1:=(1,\dots,1)^{\mathsf T}$.
We call the components~$f_{i+1,i}$ and~$f_{i,i+1}$, $i=1,\dots,m-1$, of the matrix~$\mathsf F$ the essential ones.
In what follows, $\mathsf E$ denotes the identity matrix of the same size as $\mathsf F$.

\subsection{Properties of the vertical coupling matrix}\label{sec:PropertiesOfMatrixF}

The matrix~$\mathsf F$ is identified with the corresponding linear operator on the space $\mathbb R^m$.
Under the minimal constraints of nonvanishing for the essential components of $\mathsf F$,
its rank is equal to $m-1$
and its kernel~$\ker\mathsf F$ and image~$\im\mathsf F$
are spanned by the vector~$\bar 1$ and by the columns of~$\mathsf F$, respectively.
In other words, $0$ is an eigenvalue of~$\mathsf F$ whose eigenspace is $\langle\bar 1\rangle$,
and $m-1$ columns of~$\mathsf F$ can be chosen as a basis of $\im\mathsf F$
since they are linearly independent.

In the course of carrying out Lie reductions and finding exact solutions
of systems from the class~$\mathcal M$,
we also need less trivial properties of the matrix~$\mathsf F$.
Under the physical constraints of positivity on the essential components of $\mathsf F$,
we can reduce the matrix~$\mathsf F$ by the similarity transformation $\mathsf F\to\mathsf D^{-1}\mathsf F\mathsf D$
with
\begin{gather*}
\mathsf D:=\mathop{\rm diag}(d_1,\dots,d_m),\quad
d_i:=\begin{cases}
1,                                                            &i=1,\\
\sqrt{\dfrac{f_{21}\cdots f_{i,i-1}}{f_{12}\cdots f_{i-1,i}}}
=\sqrt{\dfrac{H_1}{H_i}}, &i=2,\dots, m,
\end{cases}
\end{gather*}
where we use the relation $H_if_{i,i+1}=H_{i+1}f_{i+1,i}$,
to the symmetric tridiagonal matrix $\mathsf S$ with components
\begin{gather*}
s_{ii}=-f_{i,i-1}-f_{i,i+1},\quad
s_{i,i+1}=s_{i+1,i}=\sqrt{f_{i,i+1}f_{i+1,i}},\quad
s_{ij}=0,\ \ |i-j|>1.
\end{gather*}
(We can use the dimensional matrices~$\hat{\mathsf D}:=\mathop{\rm diag}(1/\sqrt{H_1},\dots,1/\sqrt{H_m})$
instead of the dimensionless matrix~$\mathsf D$.)
This implies that $\mathsf F$ is diagonalizable and its eigenvalues are real.
Moreover, since $f_{i,i+1}$ or $f_{i,i-1}$ are nonzero,
for any $\lambda\in\mathbb R$, the matrix $\mathsf F-\lambda\mathsf E$ contains a submatrix of rank $m-1$
(e.g., the submatrix of the matrix $\mathsf F-\lambda\mathsf E$ obtained by deleting its first row and its last column),
and thus by the rank-nullity theorem,
the dimension of the kernel of $\mathsf F-\lambda\mathsf E$ is at most $m-(m-1)=1$.
If $\lambda$ is an eigenvalue of $\mathsf F$, then this dimension must be $1$,
i.e., the geometric multiplicity of each eigenvalue of $\mathsf F$ is equal to one.
Since the matrix~$\mathsf F$ is diagonalizable, for any eigenvalue of~$\mathsf F$,
its algebraic multiplicity coincides with its geometric multiplicity and is thus equal to one as well.
In other words, eigenvalues of~$\mathsf F$ are pairwise distinct.

In fact, we can say more about the properties of eigenvalues of $\mathsf F$.
For this purpose, we need to study the characteristic polynomial $P(\lambda):=\det(\lambda\mathsf E-\mathsf F)$
of the matrix $\mathsf F$.
The above discussion implies that $P$ is divisible by $\lambda$ and has no multiple roots.

\begin{lemma}\label{lem:PositivityCharPoly}
If the physical constraints of positivity
on the essential components~$f_{i+1,i}$ and~$f_{i,i+1}$, $i=1,\dots,m-1$, of the matrix~$\mathsf F$
is imposed,
the coefficients of its characteristic polynomial $P(\lambda):=\det(\lambda\mathsf E-\mathsf F)$ are positive,
except the zero-degree one, which is equal to zero.
\end{lemma}

\begin{proof}
Denote by $\hat{\mathsf F}_k$ the $k$th order leading principal submatrix of~$\mathsf F$,
\[
\mathsf F_k:=\hat{\mathsf F}_k|^{}_{f_{k,k+1}=0},\quad
\hat P_k(\lambda):=\det(\lambda\mathsf E-\hat{\mathsf F}_k),\quad
P_k(\lambda):=\det(\lambda\mathsf E-\mathsf F_k).
\]
Thus, $\hat{\mathsf F}_m=\mathsf F_m=\mathsf F$ and $\hat P_m=P_m=P$.
We prove by induction with respect to the submatrix size~$k$
that the coefficients of the polynomials $\hat P_k$, $k=1,\dots,m-1$,
and the coefficients of the polynomials~$P_k$, $k=1,\dots,m$, except the zero-degree one, are positive,
and the zero-degree coefficient of~$P_k$ is equal to zero.

The induction base, where $k\leqslant2$, is obvious since
\begin{gather*}
P_1(\lambda)=\lambda,\quad
\hat P_1(\lambda)=\lambda+f_{12},\\
P_2(\lambda)=\lambda^2+(f_{12}+f_{21})\lambda,\quad
\hat P_2(\lambda)=\lambda^2+(f_{12}+f_{21}+f_{23})\lambda+f_{12}f_{23}.
\end{gather*}
For the induction step, suppose that the statement holds true for the submatrix sizes less than~$k$.
Then the recurrent formulas
\begin{gather*}
P_k(\lambda)=(\lambda+f_{k,k-1})P_{k-1}(\lambda)+\lambda\hat P_{k-2}(\lambda),\\
\hat P_k(\lambda)=P_k(\lambda)+f_{k,k-1}\hat P_{k-1}(\lambda).
\end{gather*}
imply the statement for the submatrix size equal to~$k$.
\end{proof}

\begin{corollary}\label{cor:EigenvaluesAreNegaive}
All nonzero eigenvalues of $\mathsf F$ are negative real.%
\footnote{%
In some oceanological literature, see, e.g., \cite[p.~56]{kame1986A},
the paper~\cite{kame1981a} is indicated as the source
with the first proof of Corollary~\ref{cor:EigenvaluesAreNegaive}
for an arbitrary value of the number of layers~$m$,
but, unfortunately, this paper is not well accessible.
This is why we decided to include the proofs of Lemma~\ref{lem:PositivityCharPoly}
and Corollary~\ref{cor:EigenvaluesAreNegaive} in the present paper.
The statement on the diagonalizability of tridiagonal matrices is commonly known.
}
\end{corollary}

\begin{proof}
Since all the eigenvalues of~$\mathsf F$ are real and
the polynomial $\lambda^{-1}P(\lambda)$ has strictly positive coefficients,
the nonzero eigenvalues should be negative.
\end{proof}

\begin{remark}
Corollary~\ref{cor:EigenvaluesAreNegaive} also follows from the Gershgorin circle theorem.
From this theorem, we can derive even a stronger statement that
any eigenvalue of the matrix~$\mathsf F$ belongs to the interval $[-R,0]$,
where $R:=\min\big(2\max_i(f_{i,i-1}+f_{i,i+1}),\max_i(f_{i,i-1}+f_{i,i+1}+f_{i-1,i}+f_{i+1,i})\big)$,
$f_{01}:=0$ and $f_{m+1,m}:=0$.
In other words, the spectral radius $\rho(\mathsf F)=|\lambda_1|$ of~$\mathsf F$ is not greater than $R$,
and, moreover, $-R\leqslant\lambda_1<\dots<\lambda_{m-1}<\lambda_m=0$.
At the same time, $\lambda_1<-\max_i(f_{i,i-1}+f_{i,i+1})$ according to the Rayleigh quotient argument.
\end{remark}

To present the solutions of~\eqref{eq:mLaysMod} derived via Lie reductions in a convenient form,
we introduce the following useful notation.
Denote by $e_1$,~\dots, $e_m$ the eigenvectors of the matrix~$\mathsf F$
corresponding to its ordered eigenvalues~$\lambda_1$, \dots $\lambda_m$ of~$\mathsf F$,
$\lambda_1<\dots<\lambda_m$,
and let $\mathsf\Lambda:=\mathop{\rm diag}(\lambda_1,\dots,\lambda_m)$ be the matrix
constituted by these eigenvalues.
Thus, $\lambda_m=0$, and we can choose $e_m=\bar 1$.
Since the matrix~$\mathsf F$ is tridiagonal and irreducible,
the first and last components of any of its eigenvectors are guaranteed to be nonzero.
We define $(\cdot,\cdot)_{\mathsf W}^{}$ to be the weighted inner product on the space~$\mathbb R^m$
with the weight matrix $\mathsf W:=\mathsf D^{-2}=\mathop{\rm diag}(d_1^{-2},\dots,d_m^{-2})$.
The dimensional analogue of~$\mathsf W$ is
$\hat{\mathsf W}:=\hat{\mathsf D}^{-2}=\mathop{\rm diag}(H_1,\dots,H_m)$.
In this Euclidean space, the eigenvectors of~$\mathsf F$ are orthogonal, $(e_i,e_j)_{\mathsf W}=0$ for $i\ne j$.
The eigenvectors~$e_1$,~\dots, $e_m$ are columns of the transition matrix~$\mathsf P$
from the matrix~$\mathsf F$ to the matrix~$\mathsf\Lambda$, $\mathsf\Lambda:=\mathsf P^{-1}\mathsf F\mathsf P$.
It is obvious that
$\mathop{\rm im}\mathsf F=\langle e_1,\dots,e_{m-1}\rangle$,
$\ker\mathsf F=\langle e_m\rangle$,
$\mathbb R^m=\mathop{\rm im}\mathsf F\oplus\ker\mathsf F$ and, moreover,
$\ker\mathsf F\perp_{\mathsf W}^{}\mathop{\rm im}\mathsf F$.
Due to this decomposition, the Moore--Penrose inverse%
\footnote{\label{fnt:Moore-PenroseInverse}
Given a linear operator~$\mathsf A$ on a finite-dimensional space~$V$ with inner product,
its Moore--Penrose inverse~$\mathsf A^+$ can be defined
in terms of the inverse of the restriction~$\hat{\mathsf A}$ of~$\mathsf A$ to $(\ker\mathsf F)^\perp$
with codomain restricted to $\mathop{\rm im}\mathsf A$,
${\mathsf A^+v=\hat{\mathsf A}^{-1}\mathop{\rm proj}_{\mathop{\rm im}\mathsf A}^{}v}$ for any $v\in V$,
where $\mathop{\rm proj}_{\mathop{\rm im}\mathsf A}$ denotes
the operator of orthogonal projection on~$\mathop{\rm im}\mathsf A$,
and hence $\mathsf A^+(\mathop{\rm im}\mathsf A)^\perp=\{0\}$.
If $\ker\mathsf A\perp\mathop{\rm im}\mathsf A$, then we have simpler expression for
the Moore--Penrose inverse of~$\mathsf A$,
$\mathsf A^+=(\mathsf A|_{\mathop{\rm im}\mathsf A}^{})^{-1}\oplus\mathsf0|_{\ker\mathsf F}^{}$,
where $\mathsf0$ is the zero operator on~$V$.
Algorithms for finding the Moore--Penrose inverse of an operator are implemented in
different packages for symbolic or numerical computations,
like {\sf Maple} and {\sf Sage} or the package {\sf Numpy} for Python.
}%
~$\mathsf F^+$ of the matrix~$\mathsf F$ interpreted as an operator on $\mathbb R^m$
can be simply represented in the decomposed form,
\begin{gather}\label{eq:MNInverse}
\mathsf F^+=(\mathsf F|_{\mathop{\rm im}\mathsf F}^{})^{-1}\oplus\mathsf0|_{\ker\mathsf F}^{},
\end{gather}
where $\mathsf0$ denotes the zero $m\times m$ matrix or, equivalently, the zero operator on~$\mathbb R^m$.

\begin{remark}
For real-world multi-layer quasi-geostrophic problems,
the vertical coupling coefficients $f_{i,i-1}$ and $f_{i,i+1}$ are all in general pairwise different
and, moreover, the parameter $f_{12}$ is usually much larger than $f_{m,m-1}$,
cf.\ Section~\ref{sec:NumericalExample} below.
For simplifying the interaction terms for each single inner layer while allowing the model to capture the asymmetric global structure (e.g., a highly stratified surface layer overlying a weakly stratified deep layer),
the constraint $f_{i,i-1}=f_{i,i+1}:=\alpha_i$, $i=2,\dots,m-1$ is often set, see, e.g., \cite{cott2020a}.
Additionally denoting $\alpha_1:=f_{12}$ and $\alpha_m:=f_{m,m-1}$,
we can decompose the matrix~$\mathsf F$ in the form $\mathsf F=\mathsf A\mathsf L$,
where
$\mathsf L$ is the $m\times m$ tridiagonal matrix with the main diagonal $(-1,-2,\dots ,-2,-1)$
and the same sub- and superdiagonals $(1,\dots ,1)$, i.e.,
the matrix of the discrete Laplacian operator for an interval with Neumann boundary conditions at both ends,
and the matrix $\mathsf A:=\mathop{\rm diag}(\alpha_1,\dots,\alpha_m)$
can be interpreted as a weight matrix to this operator.
Then $d_i=\sqrt{\alpha_i/\alpha_1}$,
but there are still no explicit closed-form expressions
for nonzero eigenvalues and corresponding eigenvectors of the matrix~$\mathsf F$
for general $(\alpha_1,\dots,\alpha_m)$.
In the course of discretizing a uniform fluid model in the vertical direction with equal steps
for further numerical analysis,
the constructed multi-layer quasi-geostrophic problem is even more simplified.
More specifically, since then $H_i=H_j$ and $\rho_{i+1}-\rho_i=\rho_{j+1}-\rho_j$, $i,j=1,\dots,m-1$,
all the vertical coupling coefficients $f_{i,i-1}$ and $f_{i,i+1}$ are the same,
except the artificial ones~$f_{1,0}$ and~$f_{m,m+1}$,
and hence the matrix~$\mathsf F$, its eigenvalues and eigenvectors take the form
\[
\mathsf F=f_{12}\mathsf L,\quad
\lambda_i=-4 f_{12}\sin^2\frac{m-i}{2m}\pi,\quad
e_i=\left(\cos\frac{m-i}m\Big(j-\frac12\Big)\pi,\,j=1,\dots,m\right)^{\mathsf T}.
\]
\end{remark}

\subsection{Barotropic and baroclinic modes}

It is often convenient to expand $\psi$ and~$q$ with respect to the eigenbasis of the matrix~$\mathsf F$,
\[
\psi=\tilde\psi^ke_k,\quad q=\tilde q^ke_k,\quad\mbox{where}\quad
\tilde\psi^k=\frac{(\psi,e_k)_{\mathsf W}^{}}{(e_k,e_k)_{\mathsf W}^{}},\quad
\tilde   q^k=\frac{(   q,e_k)_{\mathsf W}^{}}{(e_k,e_k)_{\mathsf W}^{}}.
\]
In other words,
$\tilde\psi:=(\tilde\psi^1,\dots,\tilde\psi^m)=\mathsf P^{-1}\psi$ and
$\tilde q:=(\tilde q^1,\dots,\tilde q^m)=\mathsf P^{-1}q$.
In contrast to the pairs $(\psi^i,q^i)$,
the pairs $(\tilde\psi^i,\tilde q^i)$ are decoupled in the sense that
for each fixed~$i$, the expression for the $i$th recombined vorticity~$\tilde q^i$ involves
the $i$th recombined stream function~$\tilde\psi^i$ only,
\begin{gather*}
\tilde q^i=\tilde\psi^i_{xx}+\tilde\psi^i_{yy}+\lambda_i\tilde\psi^i,\ \ i=1,\dots,m-1,\quad
\tilde q^m=\tilde\psi^m_{xx}+\tilde\psi^m_{yy}+\beta y.
\end{gather*}
The expansion with respect to the eigenbasis of~$\mathsf F$ allows one
to separate the \emph{barotropic mode}, which corresponds to the zero eigenvalue~$\lambda_m$
and represents the depth-averaged flow of the stratified fluid,
from the \emph{baroclinic modes}, which correspond to the nonzero eigenvalues~$\lambda_1$,~\dots $\lambda_{m-1}$
and represent movements of the fluid layers with respect to each other.
In terms of $\tilde\psi$ and~$\tilde q$, the system~\eqref{eq:mLaysModMatrixForm} takes the form
\begin{gather}\label{eq:mLaysModMatrixFormMod}
\begin{split}
&\tilde q_t+\mathsf P^{-1}\{\mathsf P\tilde\psi,\mathsf P\tilde q\}=0,\\
&\tilde q:=\tilde\psi_{xx}+\tilde\psi_{yy}+{\mathsf\Lambda}\tilde\psi+(0,\dots,0,\beta y)^{\mathsf T}.
\end{split}
\end{gather}
The decoupling between the components in the expression for~$\tilde q$
leads to the decoupling in derivatives with respect to~$t$,
but this is not the case for the nonlinearity given by the Poisson bracket,
where the coupling becomes even more complicated.

At the same time, most of the Lie submodels of the multi-layer quasi-geostrophic problem~\eqref{eq:mLaysModMatrixForm}
are linear, and the vertical couplings for a number of them is given by modifications
of the matrix~$\mathsf F$ with diagonal summands.
Eigenvectors of these modified matrices are in general different from those of the matrix~$\mathsf F$,
but the corresponding eigenbases are orthogonal in the same weighted inner product $(\cdot,\cdot)_{\mathsf W}^{}$.
For each submodel with such vertical couplings,
the expansion of the tuple of unknown invariant functions with respect to the associated eigenbasis
leads to the total decoupling of the submodel,
see Sections~\ref{sec:CodimOneLieReds} and~\ref{sec:CodimTwoLieReds} below.

\subsection{Numerical example}\label{sec:NumericalExample}

As a specific illustrative example of the multi-layer quasi-geostrophic problem,
throughout the paper we consider the case of three layers ($m=3$)
with realistic geophysical data from \cite[Table~1]{nauw2004a} for model parameters,
\begin{gather}\label{eq:3LaysModData}\arraycolsep=0ex
\begin{array}{lll}
f_0   =1.0\cdot10^{-4}\   \mbox{s}^{-1},                     & H_1=6.0\cdot10^2\ \mbox{m},\qquad
&g_1'=2.0\cdot10^{-2}\ \mbox{m}\cdot\mbox{s}^{-2},\\[1ex]
\beta =1.6\cdot10^{-11}\  (\mbox{m}\cdot\mbox{s})^{-1},\qquad& H_2=1.4\cdot10^3\ \mbox{m},
&g_2'=3.0\cdot10^{-2}\ \mbox{m}\cdot\mbox{s}^{-2},\\[1ex]
\rho_0=1.0\cdot10^3\      \mbox{kg}\cdot\mbox{m}^{-3},       & H_3=2.0\cdot10^3\ \mbox{m}.&
\end{array}
\end{gather}
In particular, the plots of found exact solutions are constructed for the values~\eqref{eq:3LaysModData}.
If $m=3$, there are only four essential vertical coupling parameters, $f_{12}$, $f_{21}$, $f_{23}$ and~$f_{32}$.
For the specific model data~\eqref{eq:3LaysModData}, they take the values
\begin{gather*}
f_{12}=\frac{f_0^2}{H_1g_1'}\approx8.33\cdot10^{-10}\ \text{m}^{-2},
\quad
f_{21}=\frac{f_0^2}{H_2g_1'}\approx3.57\cdot10^{-10}\ \text{m}^{-2},
\\
f_{23}=\frac{f_0^2}{H_2g_2'}\approx2.38\cdot10^{-10}\ \text{m}^{-2},
\quad
f_{32}=\frac{f_0^2}{H_3g_2'}\approx1.67\cdot10^{-10}\ \text{m}^{-2}.
\end{gather*}
Thus, the vertical coupling matrix~$\mathsf F$, its eigenvalues, the corresponding eigenvectors,
the diagonal entries of the matrix $\mathsf D:=\mathop{\rm diag}(d_1,d_2,d_3)$
and the weighted inner product $(\cdot,\cdot)_{\mathsf W}^{}$
are given by
\begin{gather*}
\mathsf F\approx10^{-10}\cdot\!
\begin{pmatrix}
-8.33 &  8.33 & 0\\
 3.57 & \!\!\!-5.95\!\!\! & 2.38\\
  0  &  1.67 &-1.67
\end{pmatrix}\!,\quad
\begin{array}{l}
\lambda_1\approx-12.9\cdot10^{-10},\\[1ex]
\lambda_2\approx -3.1\cdot10^{-10},\\[1ex]
\lambda_3=0,
\end{array}\ \,
\begin{array}{l}
e_1\approx(1,-0.54,0.08)^{\mathsf T},\\[1ex]
e_2\approx(1,0.63,-0.74)^{\mathsf T},\!\\[1ex]
e_3:=(1,1,1)^{\mathsf T},
\end{array}
\\[1ex]
d_1=1,\ \ d_2\approx0.65,\ \ d_3\approx0.55,
\quad
(b,c)_{\mathsf W}^{}\approx b_1c_1+2.33b_2c_2+3.33b_3c_3,
\end{gather*}
respectively, where $b:=(b_1,b_2,b_3)^{\mathsf T}$ and $c:=(c_1,c_2,c_3)^{\mathsf T}$.

\section{Conservation laws and Hamiltonian structure}\label{sec:HamiltonianStructure}

We extend the well-known results on local conservation laws and the Hamiltonian structure
of the $\beta$-plane vorticity equation~\cite{shep1990a}
to the multi-layer quasi-geostrophic problem~\eqref{eq:mLaysModMatrixForm}.
Below we list the names of conservation laws of~\eqref{eq:mLaysModMatrixForm}
jointly with related objects, which include their canonical characteristics,%
\footnote{%
Conservation-law characteristics~\cite[Section~4.3]{olve1993A} are also called multipliers for conservation laws,
see, e.g., \cite{anco2002a,anco2002b} and~\cite[Section~1.3.3]{blum2010A}.
}
equivalent canonical forms of their conserved currents, which are usually not unique,
and the corresponding conserved values.

In this section, we interpret functionals in the spirit of~\cite[Section~7.1]{olve1993A}
as equivalence classes of differential functions up to adding total divergences
and use the notation $\iint F[\psi]\,{\rm d}x{\rm d}y$ for them,
where $F[\psi]$ is a differential function from the corresponding equivalence class.
For such a functional to be turned into the associated physical value,
we need to fix a domain of integration and reflect this choice in the notation.

\medskip\par\noindent
{\it Conservation of generalized total weighted circulations.} \ $\lambda=\kappa(t)\mathsf W\bar1$,
\begin{gather*}
\kappa\big(0,\,(\bar1,\psi_{xt})_{\mathsf W}^{}-(\psi_y,q)_{\mathsf W}^{},\,
(\bar1,\psi_{yt})_{\mathsf W}^{}+(\psi_x,q)_{\mathsf W}^{}\big)
\\[.5ex]
\sim\big(\kappa(\bar1,\triangle\psi)_{\mathsf W}^{},\,
-\kappa_t(\bar1,\psi_x)_{\mathsf W}^{}-\kappa(\psi_y,q)_{\mathsf W}^{},\,
-\kappa_t(\bar1,\psi_y)_{\mathsf W}^{}+\kappa(\psi_x,q)_{\mathsf W}^{}\big)
\\[.5ex]
\sim\big(\kappa(\bar1,q)_{\mathsf W}^{},\,
-\kappa(\psi_y,q)_{\mathsf W}^{}-\kappa_t(\bar1,\psi_x)_{\mathsf W}^{},\,
\kappa(\psi_x,q)_{\mathsf W}^{}-(\kappa_t\bar1,\psi_y+\tfrac12\beta y^2\bar1)_{\mathsf W}^{}\big).
\end{gather*}
The first conserved current shows
that this conservation law for each fixed~$\kappa$ can be interpreted as a ``short'' one whose density is zero.
This explains the presence and the negligibility of the multiplier~$\kappa$
that is an arbitrary sufficiently smooth function of~$t$.
The second conserved current multiplied by $H_1/(H_1+\dots+H_m)$
can be interpreted as that representing the conservation of the generalized total weighted circulation,
which is the weighted average of the generalized total circulations on single layers,
\[
\iint\frac{\kappa(\bar1,\triangle\psi)_{\hat{\mathsf W}}}{H_1+\dots+H_m}\,{\rm d}x{\rm d}y
=
\sum_{j=1}^m\frac{H_j}{H_1+\dots+H_m}\iint\kappa\triangle\psi^j\,{\rm d}x{\rm d}y.
\]
Here and in what follows, we use the obvious equality
$(\cdot,\cdot)_{\hat{\mathsf W}}=H_1(\cdot,\cdot)_{\mathsf W\vphantom{\hat{\mathsf W}}}$.
The third conserved current leads to the representation of this conserved functional
in terms of the potential vorticity,
and thus this functional can be re-interpreted as a time-dependent distinguished functional
\cite[Section~7.2]{olve1993A}
for the Hamiltonian operator~$\mathfrak H$ of the system~\eqref{eq:mLaysModMatrixForm},
see below on a Hamiltonian structure for~\eqref{eq:mLaysModMatrixForm}.

\medskip\par\noindent
{\it Conservation of generalized total zonal momentums.} \ $\lambda=-\chi(t)y\mathsf W\bar1$,
\begin{gather*}
\chi\big(0,\,
y(\psi_y,q)_{\mathsf W}^{}+\tfrac12|\psi_x|_{\mathsf W}^2+\tfrac12|\psi_y|_{\mathsf W}^2
+\tfrac12(\psi,\mathsf F\psi)_{\mathsf W}^{}+y(\bar1,\beta\psi-\psi_{tx})_{\mathsf W}^{},\\\quad
-(\psi_x,y q-\psi_y)_{\mathsf W}^{}-(\bar1,y\psi_{ty}-\psi_t)_{\mathsf W}^{}\big)
\\[1ex]
\sim\big(-(\chi y\bar1,\triangle\psi)_{\mathsf W}^{},\,
\chi y(\psi_y,q)_{\mathsf W}^{}
+\tfrac12\chi(|\psi_x|_{\mathsf W}^2+|\psi_y|_{\mathsf W}^2+(\psi,\mathsf F\psi)_{\mathsf W}^{})
+y(\bar1,\beta\chi\psi+\chi_t\psi_x)_{\mathsf W}^{},\\\ \quad
-\chi(\psi_x,yq-\psi_y)_{\mathsf W}^{}+\chi_t(\bar1,y\psi_y-\psi)_{\mathsf W}^{}\big)
\\[1ex]
\sim\big(\chi(\bar1,\psi_y)_{\mathsf W}^{},\,
\chi\big(y(\psi_y,q)_{\mathsf W}^{}-\tfrac12|\psi_x|_{\mathsf W}^2-\tfrac12|\psi_y|_{\mathsf W}^2
-\tfrac12(\psi,\mathsf F\psi)_{\mathsf W}^{}-y(\bar1,\psi_{tx}+\beta\psi)_{\mathsf W}^{}\big),\\\ \quad
-\chi(\psi_x,yq-\psi_y)_{\mathsf W}^{}-(\bar1,\chi y\psi_{ty}+\chi_t\psi)_{\mathsf W}^{}\big).
\end{gather*}
Analogously, the first conserved current shows
that this conservation law  for each fixed~$\chi$ can be interpreted as a ``short'' one whose density is zero.
This again explains the presence and the negligibility of the multiplier~$\chi$
that is an arbitrary sufficiently smooth function of~$t$.
At the same time, in contrast to the previous family of conservation laws,
the short form for this conservation law exists only in terms of the stream function,
becoming nonlocal in terms of the fluid velocity.
The density of the second conserved current is nicely represented in terms of the vorticity.
For proper physical interpretation, we multiply the third conserved current by $\rho_0H_1$.
The first component of the obtained tuple is the density of the generalized total zonal momentum,
which is the sum of the total zonal momenta on single layers,
\[
\iint\chi\rho_0(\bar1,\psi_y)_{\hat{\mathsf W}}^{}\,{\rm d}x{\rm d}y
=
\sum_{i=1}^m\iint\chi\rho_0 H_i\psi^i_y\,{\rm d}x{\rm d}y.
\]

\medskip\par\noindent
{\it Conservation of total energy.} \ $\lambda=-\mathsf W\psi$,
\begin{gather*}
\big(\tfrac12(|\psi_x|_{\mathsf W}^2+|\psi_y|_{\mathsf W}^2-(\psi,\mathsf F\psi)_{\mathsf W}^{}),\,
   -(\psi,\psi_{tx}+\tfrac12\psi\odot q_y)_{\mathsf W}^{},\,
   -(\psi,\psi_{ty}-\tfrac12\psi\odot q_x)_{\mathsf W}^{}\big)
\\[.5ex]
\sim\tfrac12\big(-(\psi,q-\beta y\bar1)_{\mathsf W}^{},\,
   -(\psi,\psi\odot q_y+\psi_{xt})_{\mathsf W}^{}+(\psi_x,\psi_t)_{\mathsf W}^{},\,
   (\psi,\psi\odot q_x+\psi_{yt})_{\mathsf W}^{}-(\psi_y,\psi_t)_{\mathsf W}^{}\big).\!
\end{gather*}
After multiplying the first conserved current by $\rho_0H_1$,
its first component is the density of the total energy
\begin{gather*}
\mathcal E:=\iint\frac{\rho_0}2\big(|\psi_x|_{\hat{\mathsf W}}^2+|\psi_y|_{\hat{\mathsf W}}^2
-(\psi,\mathsf F\psi)_{\hat{\mathsf W}}^{}\big){\rm d}x{\rm d}y
=\sum_{i=1}^m\iint\frac{\rho_0 H_i}2\big((\nabla\psi^i)^2-\psi^i(\mathsf F\psi)^i\big){\rm d}x{\rm d}y,
\end{gather*}
which can of course be represented in several ways, e.g.,
as the sum of the total kinetic energy and the total energy of interaction between neighboring layers or
as the sum of total energies of single layers.

The system~\eqref{eq:mLaysMod} is Hamiltonian with a non-canonical Hamiltonian structure
under treating the potential vorticity~$q$ as the tuple of unknown functions of this system,
which is a common approach to models of fluid dynamics, cf.\ \cite[Example~7.10]{olve1993A}.
Thus, a Hamiltonian representations for~\eqref{eq:mLaysMod} is
\[
q_t=\mathfrak H\,\frac{\delta\mathcal H}{\delta q},
\]
where $\delta/\delta q:=(\delta/\delta q^1,\dots,\delta/\delta q^m)^{\mathsf T}$,
$\delta/\delta q^i$ denotes the variational derivative with respect to~$q^i$,
\[\mathfrak H:=\mathsf W^{-1}\mathop{\rm diag}(q_y\mathrm D_x-q_x\mathrm D_y)
=\mathop{\rm diag}\big((\mathsf D^2q_y)\mathrm D_x-(\mathsf D^2q_x)\mathrm D_y\big)\]
is the associated Hamiltonian differential operator with $\mathsf D:=\mathop{\rm diag}(d_1,\dots,d_m)$
and the total derivative operators~$\mathrm D_x$ and~$\mathrm D_y$ with respect to~$x$ and~$y$, respectively,
and the Hamiltonian~$\mathcal H$ of is given by the ``mathematical total energy''
(without the multiplier $\rho_0H_1$),
\begin{gather*}
\begin{split}
\mathcal H:=&\iint\frac12\big(|\psi_x|_{\mathsf W}^2+|\psi_y|_{\mathsf W}^2
-(\psi,\mathsf F\psi)_{\mathsf W}^{}\big){\rm d}x{\rm d}y
=\iint\sum_{i=1}^m\frac1{2d_i^2}\big((\nabla\psi^i)^2-\psi^i(\mathsf F\psi)^i\big){\rm d}x{\rm d}y\\
=&-\frac12\iint(q-\beta y\bar 1,\psi)_{\mathsf W}^{}\,{\rm d}x{\rm d}y
\end{split}
\end{gather*}
with the constant parameters $d_i$ defined in Section~\ref{sec:PropertiesOfMatrixF},
cf.\ \cite[Section~2.1]{cott2020a}.
(The last expression of the Hamiltonian~$\mathcal H$, which involves the potential vorticity~$q$,
justifies the presentation of the second conserved current.)
Hence $\delta\mathcal H/\delta q=-\mathsf W\psi$.
The corresponding Lie--Poisson vorticity bracket is given by
\[
\{\mathcal F,\mathcal G\}:=\sum_{i=1}^m d_i^2\iint
q^i\left\{\frac{\delta\mathcal F}{\delta q^i},\frac{\delta\mathcal G}{\delta q^i}\right\}{\rm d}x{\rm d}y,
\]
where $\mathcal F$ and $\mathcal G$ are interpreted as functionals of $q$.
The constructed Hamiltonian structure extends the (one-layer) vorticity equation~\cite{shep1990a}
to the general multi-layer model~\eqref{eq:mLaysMod}.

\medskip\par\noindent
{\it Conservation of generalized potential enstrophies on the $i$th layer, $i\in\{1,\dots,m\}$.}
\[
\lambda=\frac{{\rm d}\Phi^i}{{\rm d}q^i}\delta_i,\quad
(\Phi^i,-\psi^i_y\Phi^i,\psi^i_x\Phi^i),
\]
where $\delta_i=(\delta_{ij},j=1,\dots,m)^{\mathsf T}$ with the Kronecker delta~$\delta_{ij}$,
$\Phi^i$ is an arbitrary smooth function of~$q^i$.
For each fixed~$i$, particular elements of the corresponding family of conserved functionals
\[\mathcal C_{i,\Phi^i}:=\iint\Phi^i(q^i){\rm d}x{\rm d}y\]
are the total potential circulation on the $i$th layer, where $\Phi^i=q^i$,
and the potential enstrophy on the $i$th layer, where $\Phi^i=\frac12(q^i)^2$.
For any $(i,\Phi^i)$, the conserved functional~$\mathcal C_{i,\Phi^i}$ is
a Casimir (or, in terminology of~\cite[Section~7.2]{olve1993A}, distinguished)
functional for the Hamiltonian operator~$\mathfrak H$,
and the space of Casimir functionals of~$\mathfrak H$
is exhausted by the sums $\mathcal C_{1,\Phi^1}+\dots+\mathcal C_{m,\Phi^m}$,
where~$\Phi^i$ runs through the set of arbitrary smooth functions of~$q^i$ for each $i\in\{1,\dots,m\}$.
The space of time-dependent distinguished functionals of the Hamiltonian operator~$\mathfrak H$
that are conserved for the system~\eqref{eq:mLaysModMatrixForm}
is exhausted by the sums of generalized potential enstrophies and generalized total weighted circulations,
$\mathcal C_{1,\Phi^1}+\dots+\mathcal C_{m,\Phi^m}
+\iint\kappa(\bar1,q)_{\mathsf W}^{}\,{\rm d}x{\rm d}y$,
where in addition $\kappa$ runs through the set of arbitrary smooth functions of~$t$.

\section{Lie invariance algebra and Hamiltonian symmetries}\label{sec:LieInvAlg}

To find the maximal Lie invariance algebra~$\mathfrak g$ of the multi-layer quasi-geostrophic problem~\eqref{eq:mLaysMod},
we employ the infinitesimal invariance criterion, see, e.g., \cite[Theorem~4.1.1-1]{blum2002A} or \cite[Theorem~2.31]{olve1993A}.
In general, finding symmetry-like objects for a system of differential equations
involves sophisticated, intricate and cumbersome computations.
For low values of the number of layers~$m$,
the Lie invariance algebra~$\mathfrak g$
can be found by one of the numerous packages for symbolically computing Lie symmetries,
like {\sf DESOLV}~\cite{carm2000a} or GeM~\cite{chev2007a},
but this is not the case for an arbitrary value of~$m$.
Below we construct the algebra~$\mathfrak g$ by hand.
Since the straightforward performance of this construction is highly inefficient,
we use two tricks that significantly arrange and simplify the computations.

The first of these tricks is the formal replacement of the real independent variables~$(x,y)$
by the complex conjugate variables $z=x+{\rm i}y$ and $\bar z=x-{\rm i}y$, where ${\rm i}$ is the imaginary unit.%
\footnote{%
In fact, the variables $z$ and $\bar z$ are not independent to each other in the canonical sense
since they are uniquely related by the complex conjugation.
However, the equalities $\p_{\bar z}z=\p_z\bar z=0$ and $\p_zz=\p_{\bar z}\bar z=1$
with the Wirtinger derivatives $\p_z=\tfrac12(\p_x-{\rm i}\p_y)$ and $\p_{\bar z}=\tfrac12(\p_x+{\rm i}\p_y)$
allow one to formally treat $z$ and $\bar z$ as independent variables
and to simultaneously indicate them as arguments of relevant functions instead of $(x,y)$.
}
Throughout this section, bar denotes the complex conjugation.
In the independent variables $(t,z,\bar z)$, the system~\eqref{eq:mLaysMod} takes the form
\begin{gather}\label{eq:mLaysModCompl}
\begin{split}
&q^i_t-2{\rm i}(\psi^i_zq^i_{\bar z}-\psi^i_{\bar z}q^i_z)=0,\\
&q^i:=4\psi^i_{z\bar z}+f_{i,i-1}(\psi^i-\psi^{i-1})-f_{i,i+1}(\psi^{i+1}-\psi^i)+\frac\beta{2\rm i}(z-\bar z),\quad
i=1,\dots, m,	
\end{split}
\end{gather}
where $\psi^i=\psi^i(t,z,\bar z)$ are real-valued functions of their arguments,
$f_{1,0}=f_{m,m+1}=0$ and the rest of the constants $f_{i,i-1}$, $f_{i,i+1}$ are nonzero.
Having found the maximal Lie invariance algebra for the system~\eqref{eq:mLaysModCompl}
and pulling it back with respect to the change of variables $(t,x,y,\psi)\mapsto(t,z,\bar z,\psi)$,
we obtain the maximal Lie invariance algebra of the system~\eqref{eq:mLaysMod}.
There are two advantages of the system~\eqref{eq:mLaysModCompl} in comparison with the system~\eqref{eq:mLaysMod}.
The expansion of the former is more concise
due to replacing each of the Laplacians $\psi^i_{xx}+\psi^i_{yy}$ by the single mixed derivative $\psi^i_{z\bar z}$
(up to the multiplier~4).
Moreover, in contrast to the system~\eqref{eq:mLaysMod}, its counterpart~\eqref{eq:mLaysModCompl}
allows for confining to its solution set while preserving the parity of the variables~$z$ and~$\bar z$.

\begin{notation}
In addition in this section,
the Greek indices $\mu$, $\nu$, $\kappa$ and $\lambda$ run through the set $\{t,z,\bar z\}$.
\end{notation}

The second trick consists in introducing a chain of nested superclasses that contain the class~\eqref{eq:mLaysModCompl}.
The largest considered class $\mathcal V$ is constituted by the systems of third-order partial differential equations of the form
\begin{gather}\label{eq:GenVorClass}
\mathcal V_{F^i}\colon\quad
\psi^i_{tz\bar z}+F^i(t,z,\bar z,\psi,\psi_t,\psi_z,\psi_{\bar z},\psi_{z\bar z},\psi_{zz\bar z},\psi_{z\bar z\bar z})=0
\end{gather}
for the real-valued unknown functions $\psi=(\psi^1,\dots,\psi^m)$ of independent variables $(t,z,\bar z)$.
In view of their dependence at most on
$(t,z,\bar z,\psi,\psi_t,\psi_z,\psi_{\bar z},\psi_{z\bar z},\psi_{zz\bar z},\psi_{z\bar z\bar z})$,
the arbitrary elements $F^1$, \dots, $F^m$ can be interpreted as
real-valued third-order differential functions of~$\psi$
that run through the solution set of the auxiliary system of differential equations
\begin{gather*}
F^i_{\psi^j_{\mu\mu}}=F^i_{\psi^j_{t\mu}}=F^i_{\psi^j_{t\nu\lambda}}=F^i_{\psi^j_{zzz}}=F^i_{\psi^j_{\bar z\bar z\bar z}}=0.
\end{gather*}
The second, narrower class~$\hat{\mathcal V}$ consists of the systems of the form
\begin{gather}\label{eq:HatGenVorClass}
\hat{\mathcal V}_{(a_i,\hat F^i)}\colon\quad
\psi^i_{tz\bar z}+a_i\psi^i_{\bar z}\psi^i_{zz\bar z}+\bar a_i\psi^i_z\psi^i_{z\bar z\bar z}
+\hat F^i(t,z,\bar z,\psi,\psi_t,\psi_z,\psi_{\bar z},\psi_{z\bar z})=0.
\end{gather}
The class~$\hat{\mathcal V}$ can be considered as the subclass of the class~$\mathcal V$
that is singled out by the auxiliary system
\begin{gather*}
F^i_{\psi^j_{zz\bar z}}=F^i_{\psi^j_{z\bar z\bar z}}=0,\quad i\ne j,
\qquad
\left(\frac{F^i_{\psi^i_{zz\bar z}}}{\psi^i_{\bar z}}\right)_\upsilon=
\left(\frac{F^i_{\psi^i_{z\bar z\bar z}}}{\psi^i_z}\right)_\upsilon=0,\quad
\upsilon\in\{\mu,\psi_\mu,\psi_{\mu\nu},\psi_{\mu\nu\kappa}\}.
\end{gather*}
We can also reparameterize the class~$\hat{\mathcal V}$
and declare, instead of $F^i=a_i\psi^i_{\bar z}\psi^i_{zz\bar z}+\bar a_i\psi^i_z\psi^i_{z\bar z\bar z}+\hat F^i$,
the nonzero complex constants~$a_i$ and the real-valued second-order differential functions $\hat F^i$ to be the arbitrary elements,
which thus satisfy the auxiliary system
\begin{gather*}
a_{i,\upsilon}=0,\quad \upsilon\in\{\mu,\psi_\mu,\psi_{\mu\nu},\psi_{\mu\nu\kappa}\},\quad a_i\ne0, \quad
\hat F^i_{\psi_{\mu\nu\kappa}}=\hat F^i_{\psi_{t\mu}}=\hat F^i_{\psi_{zz}}=\hat F^i_{\psi_{\bar z\bar z}}=0.
\end{gather*}

The introduction of the classes $\hat{\mathcal V}$ and $\mathcal V$
allows us to simplify the computation of the maximal Lie invariance algebra
of the system of the form~\eqref{eq:mLaysModCompl} via splitting it into natural steps in the following way.
Applying the infinitesimal invariance criterion to an arbitrary system from the most general class~$\mathcal V$,
we derive the principal constraints on the components of Lie symmetry vector fields,
which hold for any system from this class and, consequently, for any system belonging to~$\hat{\mathcal V}$.
Using the obtained constraints, we proceed with successively specifying the form of Lie symmetry vector fields
for systems from the class~$\hat{\mathcal V}$ and for systems of the form~\eqref{eq:mLaysModCompl}.
As a result, we need to expand the cumbersome expressions for~$F^i$ in the latter systems
only on the very last stage of computing, which is one of the sources of the caused simplification.
Moreover, sequentially deriving more and more restrictions on the components of the Lie symmetry vector fields
allows us to avoid solving at once a vast overdetermined system of determining equations,
which arise in the straightforward application of the infinitesimal invariance criterion
to systems of the form~\eqref{eq:mLaysModCompl}.

The suggested approach results in the next chain of lemmas.

\begin{lemma}\label{lem:LieSymVecFieldClassV}
Any Lie symmetry vector field of an arbitrary system from the class $\mathcal V$ is of the form
\[
Q=\tau(t)\p_t+\xi(t,z)\p_z+\bar\xi(t,\bar z)\p_{\bar z}+\big(\Psi^{jk}(t)\psi^k+\Phi^j(t,z,\bar z)\big)\p_{\psi^j},
\]
where $\tau$, $\Psi^{jk}$ and $\Phi^j$ are smooth real-valued functions of their arguments,
and $\xi$ is a smooth complex-valued function of $(t,z)$.
\end{lemma}

\begin{proof}
In order to find the above form,
we start from a vector field~$Q$ of the most general form on the space with the coordinates $(t,z,\bar z,\psi)$,
\[
Q=\tau(t,z,\bar z,\psi)\p_t+\xi(t,z,\bar z,\psi)\p_z+\bar\xi(t,z,\bar z,\psi)\p_{\bar z}+\eta^j(t,z,\bar z,\psi)\p_{\psi^j},
\]
where $\tau$ and~$\eta^1$, \dots,~$\eta^m$ (resp.\ $\xi$) are real-valued (resp. is complex-valued) smooth function(s) of $(t,z,\bar z,\psi)$.
The necessary and sufficient condition for $Q$ to be a Lie symmetry vector field of a fixed system
$\mathcal V_F$: $\psi^i_{tz\bar z}+F^i=0$ from the class $\mathcal V$ with a fixed tuple of parameter-functions $F=(F^1,\dots,F^m)$
is given by the infinitesimal invariance criterion, see, e.g.,~\cite[Theorem~2.31]{olve1993A}.
In the specific case of $Q$ and $\mathcal V_F$, the criterion implies the condition
$Q_{(3)}(\psi^i_{tz\bar z}+F^i)|_{\mathcal V_F}^{}=0$, where~$Q_{(3)}$ is the third prolongation of~$Q$,
\[
Q_{(3)}=Q+\eta^{j,\mu}\p_{\psi^j_\mu}+\eta^{j,\mu\nu}\p_{\psi^j_{\mu\nu}}+\eta^{j,\mu\nu\kappa}\p_{\psi^j_{\mu\nu\kappa}},
\]
and the notation~\smash{$|_{\mathcal V_F}^{}$} means that the condition holds only on the solution set of the system~$\mathcal V_F$.
Here and in what follows, tuples of Greek indices are assumed unordered.
The coefficients $\eta^{j,\mu}$, $\eta^{j,\mu\nu}$ and $\eta^{j,\mu\nu\kappa}$
are defined by the general prolongation formulas for vector fields.
In particular,
\[
\eta^{i,\mu\nu\kappa}={\rm D_\mu D_\nu D_\kappa}(\eta^i-\tau\psi^i_t-\xi\psi^i_z-\bar\xi\psi^i_{\bar z})
+\tau\psi^i_{t\mu\nu\kappa}+\xi\psi^i_{z\mu\nu\kappa}+\bar\xi\psi^i_{\bar z\mu\nu\kappa}.
\]
Here ${\rm D_\mu}$ denotes the total derivative operator with respect to $\mu\in\{t,z,\bar z\}$,
\[
{\rm D_\mu}=\p_\mu+\psi^j_\mu\p_{\psi^j}+\psi^j_{\mu\nu}\p_{\psi^j_\nu}
+\psi^j_{\mu\nu\kappa}\p_{\psi^j_{\nu\kappa}}+\psi^j_{\mu\nu\kappa\lambda}\p_{\psi^j_{\nu\kappa\lambda}}+\cdots.
\]
Expanding the invariance condition for~$\mathcal V_F$, we obtain
\begin{gather}\label{eq:InvCondGenVortClass}
\eta^{i,tz\bar z}\!
+\eta^{j,zz\bar z}F^i_{\psi^j_{zz\bar z}}\!\!+\eta^{j,z\bar z\bar z}F^i_{\psi^j_{z\bar z\bar z}}\!\!
+\eta^{j,z\bar z}F^i_{\psi^j_{z\bar z}}\!\!
+\eta^{j,\mu}F^i_{\psi^j_\mu}+\eta^jF^i_{\psi^j}\!
+\tau F^i_t+\xi F^i_z+\bar\xi F^i_{\bar z}=0
\end{gather}
whenever $\psi^i_{tz\bar z}+F^i=0$.
After further expanding the system~\eqref{eq:InvCondGenVortClass} and substituting $-F^i$ for $\psi^i_{tz\bar z}$,
we split it with respect to the third- and second-order derivatives of $\psi$ that are not among the arguments of $F$;
the other jet variables are not appropriate to be used in the course of splitting.
More specifically, separately collecting coefficients of summands involving
$\psi^i_{tt\bar z}$ and $\psi^i_{ttz}$,
$\psi^i_{t\bar z\bar z}$ and $\psi^i_{tzz}$,
$\psi^i_{t\bar z}$ and $\psi^i_{tz}$,
each of which originates only from the expansion of~$\eta^{i,tz\bar z}$,
results in the determining equations
\begin{gather*}
{\rm D}_z\tau={\rm D}_{\bar z}\tau=0,\quad
{\rm D}_z\bar\xi={\rm D}_{\bar z}\xi=0,\quad
\eta^i_{\psi^i\psi^j}=\eta^i_{z\psi^i}=\eta^i_{\bar z\psi^i}=0,
\end{gather*}
respectively.
The simultaneous integration of these equations gives the required form of the components of~$Q$.
\end{proof}

In other words, for an arbitrary system from the class~$\mathcal V$,
any of its Lie symmetry vector fields is projectable to the space coordinatized by the independent variables $(t,z,\bar z)$
and is affine in $\psi$.

\begin{lemma}\label{lem:LieSymVecFieldClassHatV}
Any Lie symmetry vector field of an arbitrary system from the class~$\hat{\mathcal V}$ is of the form
\begin{gather*}
Q=\tau\p_t+\big(\gamma z+\delta\big)\p_z+\big(\bar\gamma\bar z+\bar\delta\big)\p_{\bar z}\\
\hphantom{Q=}
+\big((\gamma+\bar\gamma-\tau_t)\psi^j+\tfrac12(a_j^{\,-1}\gamma_t+\bar a_j^{\,-1}\bar\gamma_t)z\bar z+a_j^{\,-1}\delta_t\bar z+\bar a_j^{\,-1}\bar\delta_tz+\phi^j\big)\p_{\psi^j},
\end{gather*}
where $\tau=\tau(t)$ and $\phi^j=\phi^j(t)$ (resp.\ $\gamma=\gamma(t)$ and $\delta=\delta(t)$)
are real-valued (resp.\ complex-valued) smooth functions of~$t$,
and $\mathop{\rm Arg}\gamma_t=\pm\mathop{\rm Arg}a_j$
with $\mathop{\rm Arg}$ denoting the principal value of the corresponding argument.
\end{lemma}

\begin{proof}
The class~$\hat{\mathcal V}$ is a subclass of~$\mathcal V$.
Hence we can substitute the expressions for the components of Lie symmetry vector fields from Lemma~\ref{lem:LieSymVecFieldClassV}
for the arbitrary elements~$F^i$ within the subclass~$\hat{\mathcal V}$,
$F^i=a_i\psi^i_{\bar z}\psi^i_{zz\bar z}+\bar a_i\psi^i_z\psi^i_{z\bar z\bar z}+\hat F^i$, into~\eqref{eq:InvCondGenVortClass}
to derive the condition of invariance of a system in~$\hat{\mathcal V}$ with respect a vector field~$Q$,
\begin{gather}\label{eq:InvCondHatGenVortClass}
\begin{split}
&\eta^{i,tz\bar z}
+a_i\eta^{i,zz\bar z}\psi^i_{\bar z}+\bar a_i\eta^{i,z\bar z\bar z}\psi^i_z
+a_i\eta^{i,\bar z}\psi^i_{zz\bar z}+\bar a_i\eta^{i,z}\psi^i_{z\bar z\bar z}
\\
&
\hphantom{\eta^{i,tz\bar z}}
+\eta^{j,z\bar z}\hat F^i_{\psi^j_{z\bar z}}+\eta^{j,\alpha}\hat F^i_{\psi^j_\alpha}
+\eta^j\hat F^i_{\psi^j}
+\tau\hat F^i_t+\xi\hat F^i_z+\bar\xi\hat F^i_{\bar z}=0
\end{split}
\end{gather}
whenever $\psi^i_{tz\bar z}+a_i\psi^i_{\bar z}\psi^i_{zz\bar z}+\bar a_i\psi^i_z\psi^i_{z\bar z\bar z}+\hat F^i=0$.

We expand the invariance condition~\eqref{eq:InvCondHatGenVortClass},
substitute $-a_i\psi^i_{\bar z}\psi^i_{zz\bar z}-\bar a_i\psi^i_z\psi^i_{z\bar z\bar z}-\hat F^i$
for the derivative~$\psi^i_{tz\bar z}$ whenever it appears
and split the obtained equality with respect to the third-order derivatives~$\psi^j_{zz\bar z}$ and~$\psi^j_{z\bar z\bar z}$.
(The rest of the jet variables are not appropriate for splitting since they are in the arguments of $\hat{F}^i$
or have been used in the course of splitting in Lemma~\ref{lem:LieSymVecFieldClassV}.)
The equations corresponding summands involving~$\psi^j_{zz\bar z}$ or~$\psi^j_{z\bar z\bar z}$
can be further split with respect to~$(\psi^k_z,\psi^k_{\bar z})$,
which leads to the following system on the parameters in $\eta^i$:
\begin{subequations}
\begin{gather}\label{eq:constPsi}
\Psi^{ij}=-\tau_t+\xi_z+\bar{\xi}_{\bar z},\quad i=j,\qquad \Psi^{ij}=0,\quad i\ne j,
\\
a_i\Phi^i_{\bar z}=\xi_t,\quad
\bar a_i\Phi^i_z=\bar{\xi}_t.\label{eq:Phii}
\end{gather}
Since $\Psi^{ii}$ depends only on $t$,
differentiating~\eqref{eq:constPsi} with respect to~$z$ results in the equation~$\xi_{zz}=0$,
whose general solution is $\xi=\gamma z+\delta$,
where $\gamma=\gamma(t)$ and~$\delta=\delta(t)$ are smooth complex-valued functions.
Therefore, $\Psi^{ii}=-\tau_t+\gamma+\bar\gamma$,
and the equations~\eqref{eq:Phii} are expanded as
\begin{gather}\label{eq:PhiiExpanded}
a_i\Phi^i_{\bar z}=\gamma_tz+\delta_t,\quad
\bar a_i\Phi^i_z=\bar\gamma_t\bar z+\bar\delta_t.
\end{gather}
\end{subequations}
The compatibility condition $\Phi^i_{\bar zz}=\Phi^i_{z\bar z}$ reduces to
$\bar a_i\gamma_t=a_i\bar\gamma_t$, which means that $\bar a_i\gamma_t\in\mathbb R$,
i.e., $\mathop{\rm Arg}\gamma_t=\pm\mathop{\rm Arg}a_i$.
Integrating~\eqref{eq:PhiiExpanded} leads to the required form of the vector field~$Q$.
\end{proof}

Any system of the form~\eqref{eq:mLaysModCompl} belongs to the class~$\hat{\mathcal V}$ and expands into
\begin{gather}\label{eq:MLayModEnhanced}
\begin{split}
\mathcal M^{\rm c}_\theta\colon\quad
&\psi^i_{tz\bar z}-2{\rm i}(\psi^i_z\psi^i_{z\bar z\bar z}-\psi^i_{\bar z}\psi^i_{zz\bar z})+\check F^i=0,\quad
i=1,\dots,m,\\
&\check F^i:=\frac\beta4(\psi^i_z+\psi^i_{\bar z})
+\frac14\sum_{k\in I_i}f_{ik}\big(\psi^i_t-\psi^k_t-2{\rm i}(\psi^k_z\psi^i_{\bar z}-\psi^i_z\psi^k_{\bar z})\big)
\end{split}
\end{gather}
with $I_i:=\{i-1,i+1\}$, $\beta,f_{i,i-1},f_{i,i+1}\ne0$, except $f_{1,0}:=0$ and~$f_{m,m+1}=0$.
Recall that $\theta:=(f_{i+1,i},f_{i,i+1},i=1,\dots,m-1,\beta)$.	

\begin{lemma}\label{lem:LieInvAlgMLaysModelCompl}
The maximal Lie invariance algebra~$\mathfrak g^{\rm c}$ of any system~$\mathcal M^{\rm c}_\theta$ of the form~\eqref{eq:MLayModEnhanced} is spanned by the vector fields
\begin{gather*}
\p_t,\quad
\p_z-\p_{\bar z},\quad
\chi(t)(\p_z+\p_{\bar z})-\frac1{2\rm i}(z-\bar z)\chi_t\big(\p_{\psi^1}+\dots+\p_{\psi^m}\big),\\
\p_{\psi^1},\ \dots,\ \p_{\psi^m},\quad
\kappa(t)\big(\p_{\psi^1}+\dots+\p_{\psi^m}\big),
\end{gather*}
where $\chi$ and $\kappa$ are arbitrary real-valued smooth functions of~$t$.
\end{lemma}

\begin{proof}
We use the fact the system~$\mathcal M^{\rm c}_\theta$ belongs to the class~$\hat{\mathcal V}$
and specify Lemma~\ref{lem:LieSymVecFieldClassHatV} for~$\mathcal M^{\rm c}_\theta$,
substituting $\check F^i$ and $2{\rm i}$ for $\hat F^i$ and~$a_i$, respectively.
In particular, we proceed with further deriving the constraints on the parameter functions~$\tau$, $\gamma$, $\delta$ and $\phi^j$
from the specified invariance condition~\eqref{eq:InvCondHatGenVortClass},
\begin{gather}\label{eq:InvCondMLayers}
\eta^{i,tz\bar z}
+2{\rm i}\eta^{i,zz\bar z}\psi^i_{\bar z}-2{\rm i}\eta^{i,z\bar z\bar z}\psi^i_z
+2{\rm i}\eta^{i,\bar z}\psi^i_{zz\bar z}-2{\rm i}\eta^{i,z}\psi^i_{z\bar z\bar z}
+\sum_{j=i-1}^{i+1}\eta^{j,\mu}\check F^i_{\psi^j_\mu}=0
\end{gather}
whenever $\psi^i_{tz\bar z}+2{\rm i}\psi^i_{zz\bar z}-2{\rm i}\psi^i_{z\bar z\bar z}+\check F^i=0$.

After expanding, substituting $-2{\rm i}\psi^i_{zz\bar z}+2{\rm i}\psi^i_{z\bar z\bar z}-\check F^i$ for $\psi^i_{tz\bar z}$
and simplifying, the system~\eqref{eq:InvCondMLayers} does not contain the second- and third-order derivatives of~$\psi$
and can be split with respect to the first-order derivatives of~$\psi$,
which results in the equations
\begin{gather}\label{eq:SplitInvCond}
\begin{split}&
f_{i,i-1}(\bar\gamma+\tau_t)=0,\quad
f_{i,i+1}(\gamma+\tau_t)=0,\quad
\beta(\bar\gamma+\tau_t)=0,\quad
\beta(\gamma+\tau_t)=0,\\&
f_{i,i-1}(\gamma+\bar\gamma)=0,\quad
f_{i,i+1}(\gamma+\bar\gamma)=0.
\end{split}
\\\label{eq:ReducedInvCond}
f_{i,i-1}(\phi^{i-1}_t-\phi^i_t)
+f_{i,i+1}(\phi^{i+1}_t-\phi^i_t)=
\frac{\rm i}2\beta(\bar\gamma_t\bar z-\gamma_tz+\bar\delta_t-\delta_t).
\end{gather}
In view of the fact that all constants $f_{i,i-1}$, $f_{i,i+1}$ (positive) and $\beta$ are nonzero,
the equations~\eqref{eq:SplitInvCond} imply that $\gamma=-\tau_t$, $\bar\gamma=-\tau_t$ and $\gamma=-\bar\gamma$,
which in turn means $\gamma=\bar\gamma=\tau_t=0$.

We consider the system~\eqref{eq:ReducedInvCond} as a system of linear (algebraic) equations
with respect to the unknowns $(\phi^1_t,\phi^2_t,\dots,\phi^m_t)$.
Its coefficient matrix $A$ is given by $A:=(f_{ij})_{i,j=1}^m$,
where $f_{ii}=-f_{i,i-1}-f_{i,i+1}$ and $f_{ij}=0$ when $|i-j|>1$.
The rank of~$A$ is equal to $m-1$,
and the linear span of the columns of $A$ over~$\mathbb R$
coincides with that of the columns $v_j$, $j\in\{1,\dots,m-1\}$
with the only nonzero entries $v_{j,j}=-f_{j,j+1}$ and $v_{j,j+1}=f_{j+1,j}$.
In view of the Rouch\'e--Capelli theorem and the consistency of the system~\eqref{eq:ReducedInvCond},
the rank of the augmented matrix of this system is equal to the rank of~$A$, i.e.,
the column $\frac{\rm i}2\beta(\bar\delta_t-\delta_t)(1,\dots,1)^{\mathsf T}$,
whose entries are real, belongs to the linear span $\langle v_1,\dots,v_{m-1}\rangle$ over~$\mathbb R$.
Suppose that $\bar\delta_t\ne\delta_t$.
Then the above claim also holds true for the column $(1,\dots,1)^{\mathsf T}$,
which is equivalent to the consistency of the system
\begin{gather*}
-\lambda_1f_{1,2}=1,\quad
-\lambda_2f_{2,3}+\lambda_1f_{2,1}=1,\quad \dots,
\\ \qquad
-\lambda_{m-1}f_{m-1,m}+\lambda_{m-2}f_{m-1,m-2}=1,\quad
\lambda_{m-1}f_{m,m-1}=1
\end{gather*}
with respect to the unknowns $\lambda_1$, \dots, $\lambda_{m-1}$.
Successively solving the first $(m-1)$ equations, we obtain the solution in the recursive form
\[
\lambda_1=-\frac1{f_{1,2}},\quad
\lambda_j=-\frac{1-\lambda_{j-1}f_{j,j-1}}{f_{j,j+1}},\quad j\in\{2,\dots,m-1\},
\]
which inductively implies $\lambda_1\leqslant 0$, \dots, $\lambda_{m-1}\leqslant 0$.
However, this contradicts the last equation of the system.
Hence $\delta_t=\bar\delta_t$, i.e., $\delta=\chi+{\rm i}c$.
This means that in fact the right-hand sides of equations of the system~\eqref{eq:ReducedInvCond} vanish,
and its general solution is $\phi^i=\kappa+b^i$.
Here $\chi=\chi(t)$ and $\kappa=\kappa(t)$ are arbitrary smooth real-valued functions of $t$,
and $b^i$ and $c$ are arbitrary real constants.
\noprint{
Summing up the above arguments we have that
$\tau_t=0$,
$\xi=\chi(t)+{\rm i}c$,
$\bar\xi=\chi(t)-{\rm i}c$ and
$\eta^i=\tfrac1{2\rm i}\chi_t(\bar z-z)+\kappa(t)+b^i$.
This results in the lemma's statement.
}
\end{proof}

Pulling back the vectors fields given in Lemma~\ref{lem:LieInvAlgMLaysModelCompl}
with respect to the change of variables $\tilde t=t$, $z=x+{\rm i}y$, $\bar z=x-{\rm i}y$,
we obtain a spanning set of vector fields for the maximal Lie invariance algebra of the system~\eqref{eq:mLaysMod}.

\begin{theorem}\label{thm:mLaysLieInvAlg}
The maximal Lie invariance algebra~$\mathfrak g$ of the system~\eqref{eq:mLaysMod} is spanned by the vector fields
\begin{gather}\label{eq:mLaysLieInvAlg}
\begin{split}&
\mathcal P^t=\p_t,\quad
\mathcal P^y=\p_y,\quad
\mathcal P^x(\chi)=\chi(t)\p_x-y\chi_t(\p_{\psi^1}+\dots+\p_{\psi^m}),\\&
\mathcal J^1=\p_{\psi^1},\quad
\dots,\quad
\mathcal J^m=\p_{\psi^m},\quad
\mathcal Z(\kappa)=\kappa(t)(\p_{\psi^1}+\dots+\p_{\psi^m}),
\end{split}
\end{gather}
where $\chi$ and $\kappa$ are arbitrary smooth functions of~$t$.
\end{theorem}

\begin{remark}
The algebra~$\mathfrak g$ is infinite-dimensional.
The vector fields~\eqref{eq:mLaysLieInvAlg} do not constitute a~basis of~$\mathfrak g$
since they are linearly dependent.
There is no natural way to choose a basis in~$\mathfrak g$.
The main reason for this is that there is no natural basis in the space of smooth functions of~$t$.
\end{remark}

\begin{remark}\label{rem:mLayerMLIAGaugeTrans}
The coefficients in the expansion of any element of the algebra~$\mathfrak g$
with respect to the collection of vector fields~\eqref{eq:mLaysLieInvAlg},
$a_1\mathcal P^t+a_2\mathcal P^y+\mathcal P^x(\chi)+c_k\mathcal J^k+\mathcal Z(\kappa)$,
is defined up to the gauge transformations
$\tilde c=c+\epsilon\bar 1$, $\tilde\kappa=\kappa-\epsilon$ with arbitrary $\epsilon\in\mathbb R$,
where $c:=(c_1,\dots,c_m)^{\mathsf T}$, and similarly for~$\tilde c$.
Due to the presence of these gauge transformations, we can set without loss of generality that
{\it
$c\perp_{\mathsf W}^{}\bar 1$ or, equivalently, $c\in\mathop{\rm im}\mathsf F$
in the expansion of any element of the algebra~$\mathfrak g$
with respect to the collection of vector fields~\eqref{eq:mLaysLieInvAlg}.}
At the same time, in particular cases when constant summand~$\alpha$ is singled out in the parameter function~$\kappa$
or, moreover, $\kappa=\alpha$, it may be convenient to change
$c+\alpha\bar 1\to c$ and $\kappa-\alpha\to\kappa$ and neglect the gauge $c\perp_{\mathsf W}^{}\bar 1$.
\end{remark}

Up to antisymmetry of the Lie bracket, the nontrivial commutation relations
between vector fields from the spanning set of~$\mathfrak g$
from Theorem~\ref{thm:mLaysLieInvAlg} are exhausted by
\begin{gather*}
[\mathcal P^t,\mathcal P^x(\chi)]=\mathcal P^x(\chi_t),\quad
[\mathcal P^t,\mathcal Z(\kappa)]=\mathcal Z(\kappa_t),\quad
[\mathcal P^x(\chi),\mathcal P^y]=\mathcal Z(\chi_t).
\end{gather*}
It is thus obvious that the algebra~$\mathfrak g$ is solvable and not nilpotent.
Its nilradical is
$\mathfrak n:=\langle\mathcal P^y,\mathcal P^x(\chi),\mathcal J^1,\dots,\mathcal J^m,\mathcal Z(\kappa)\rangle$.
To classify subalgebras of $\mathfrak g$ it is convenient to
decompose the algebra $\mathfrak g$ into the semidirect sum of
the subalgebras as follows:
\begin{gather*}
\mathfrak g=\langle\mathcal P^t,\mathcal P^y\rangle
\lsemioplus\langle\mathcal P^x(\chi),\mathcal J^1,\dots,\mathcal J^m,\mathcal Z(\kappa)\rangle.
\end{gather*}

Since the Hamiltonian structure of the system~\eqref{eq:mLaysModMatrixForm} is not canonical,
see Section~\ref{sec:HamiltonianStructure},
it establishes a connection between conserved quantities and Hamiltonian symmetries in a specific way.
The Hamiltonian operator~$\mathfrak H$ of the system~\eqref{eq:mLaysModMatrixForm}
maps characteristics of local conservation laws
to the $q$-component tuple, $\hat{\mathcal Q}q=(\hat{\mathcal Q}q^1,\dots,\hat{\mathcal Q}q^m)^{\mathsf T}$,
of prolonged evolution forms, $\hat{\mathcal Q}$, of local infinitesimal Hamiltonian symmetries, $\mathcal Q$.
Given a generalized symmetry vector field ${\mathcal Q=\tau\p_t+\xi^x\p_x+\xi^y\p_y+\eta^i\p_{\psi^i}}$,
each of whose components is a differential functions of~$\psi$, i.e.,
a smooth function of the independent variables~$t$, $x$ and~$y$
and a finite number of derivatives of~$\psi$ with respect to these variables,
we have $\hat{\mathcal Q}q=(\mathrm D_x^2+\mathrm D_y^2+\mathsf F)(\eta-\tau\psi_t-\xi^x\psi_x-\xi^y\psi_y)$.
Hence
\begin{gather*}
-\hat{\mathcal P}^tq=q_t=\mathfrak H(-\mathsf W\psi),\quad
-\hat{\mathcal P}^x(\chi)q=\chi(t)q_x=\mathfrak H\big(-\chi(t)y\mathsf W\bar1\big),
\\
-\hat{\mathcal Z}(\kappa)q=0=\mathfrak H\big(\kappa(t)\mathsf W\bar1\big),\quad
0=\mathfrak H\big(\Phi^i_{q^i}\delta_i\big),
\\
-\hat{\mathcal P}^yq=q_y,\quad -\hat{\mathcal J}^iq=-\mathsf F\delta_i,
\end{gather*}
where $\delta_i=(\delta_{ij},j=1,\dots,m)^{\mathsf T}$ with the Kronecker delta~$\delta_{ij}$.
Thus, we have
\begin{itemize}\itemsep=0ex
\item[$\diamond$]
the standard relation of the invariance with respect to shifts in time ($\mathcal Q=\mathcal P^t$)
with the conservation of energy (${\lambda=-\mathsf W\psi}$),
\item[$\diamond$]
the standard relation of the invariance with respect to generalized (time-dependent) shifts in~$x$
($\mathcal Q=\mathcal P^x(\chi)$)
with the conservation of generalized total zonal momentums ($\lambda=-\chi(t)y\mathsf W\bar1$) and
\item[$\diamond$]
the degenerate (via zero) association of simultaneous equal time-dependent shifts of~$\psi^i$
($\mathcal Q=\mathcal Z(\kappa)$)
with generalized total weighted circulations ($\lambda=\kappa(t)\mathsf W\bar1$),
which are conserved time-dependent Casimir functionals.
\end{itemize}
All the Casimir functionals (\smash{$\lambda=\Phi^i_{q^i}\delta_i$}) are mapped by~$\mathfrak H$ to zero
and, except the total weighted circulation ($\lambda=\mathsf W\bar1$), has no relation to symmetries.
The algebra of Hamiltonian Lie symmetries of the system~\eqref{eq:mLaysModMatrixForm}
is $\langle\mathcal P^t,\mathcal P^x(\chi),\mathcal Z(\kappa)\rangle$,
whereas $\mathcal P^y$ and $\mathcal J^1$, \dots, $\mathcal J^m$ are not such symmetries
and thus have no counterparts among conserved quantities of~\eqref{eq:mLaysModMatrixForm}.

\section{Point symmetry pseudogroup}\label{sec:mLaysPointSymGroup}

To find the point symmetry pseudogroup of the multi-layer quasi-geostrophic model~\eqref{eq:mLaysMod},
we use the megaideal-based version of the algebraic method~\cite{bihl2012a}.
The description and applications of this technique can be found, e.g., in~\cite{boyk2024a,opan2020a}
and references therein.
It is an extension of the automorphism-based version of the algebraic method 
first introduced by Hydon~\cite{hydo2000A,hydo2000b}.

Both versions are based on the idea that the pushforward~$\Phi_*$ of any point-symmetry transformation~$\Phi$
of the system of differential equations is an automorphism of its maximal Lie invariance algebra~$\mathfrak g$,
$\Phi_*\mathfrak g=\mathfrak g$.
But Hydon's approach requires the explicit computation of the automorphisms of the algebra~$\mathfrak g$,
which certainly cannot be done if its dimension is not finite;
even for a finite-dimensional Lie algebra, the explicit construction of its automorphisms can be highly complicated.
In contrast, the megaideal-based version of the algebraic method employs the notion of a megaideal~$\mathfrak m$ of a Lie algebra~$\mathfrak g$
(which is also called a fully characteristic ideal in \cite[Exrecise~14.1.1]{hilg2012A}).
Recall that a vector subspace $\mathfrak m$ of a Lie algebra $\mathfrak g$ is called {\it megaideal}
if it is invariant under all automorphisms of $\mathfrak g$~\cite{popo2003a},
that is, $\mathfrak T\mathfrak m\subseteq\mathfrak m$ for all $\mathfrak T\in{\rm Aut}(\mathfrak g)$.

The first step in the megaideal-based version of the algebraic method
is to construct a certain set of megaideals of the maximal Lie invariance algebra $\mathfrak g$.
Then we consider a point transformation $\Phi$ acting on the space of independent and dependent variables
and using property $\Phi_*\mathfrak m\subseteq\mathfrak m$ find the general constraints imposed on components of $\Phi$
for each megaideal $\mathfrak m$ from the found set.
After deriving all possible constraints on the components of $\Phi$ in such a manner using all found megaideals,
we check if the condition $\Phi_*\mathfrak g\subseteq\mathfrak g$ is satisfied,
which guarantees that no more constraints within the algebraic method can be found.
Then one finishes the computation by applying the direct method,
which is usually not computationally expensive.

Elementary properties of megaideals allow us to find a large set of megaideals of the algebra~$\mathfrak g$
without any knowledge of the automorphism group~${\rm Aut}(\mathfrak g)$ of~$\mathfrak g$.
For example, all elements of the derived series and lower and upper central series of~$\mathfrak g$,
including the center~$\mathfrak z$ and the derivative~$\mathfrak g'$ of~$\mathfrak g$,
as well as the radical~$\mathfrak r$ and nilradical~$\mathfrak n$ of~$\mathfrak g$, are megaideals of~$\mathfrak g$.
Moreover, if $\mathfrak i_1$ and $\mathfrak i_2$ are megaideals of~$\mathfrak g$,
then so are $\mathfrak i_1+\mathfrak i_2$, $\mathfrak i_1\cap\mathfrak i_2$ and $[\mathfrak i_1, \mathfrak i_2]$,
i.e., sums, intersections and commutators of megaideals are again megaideals.
Note that for computing point symmetry transformations,
we should exclude megaideals that are sums of two proper submegaideals from the consideration,
since they do not provide essentially new constraints on the components of the transformation 
in comparison to those obtained from the summands.
If $\mathfrak i_2$ is a megaideal of $\mathfrak i_1$ and $\mathfrak i_1$ is a megaideal of~$\mathfrak g$,
then $\mathfrak i_2$ is a megaideal of~$\mathfrak g$,
i.e., megaideals of megaideals are also megaideals.
The following lemma from~\cite{card2013a} is also used to construct megaideals from known ones.

\begin{lemma}\label{lem:SpaceOfCommutElements}
If~$\mathfrak i_0$, $\mathfrak i_1$ and $\mathfrak i_2$ are megaideals of~$\mathfrak g$,
then the set~$\mathfrak s$ of elements from~$\mathfrak i_0$ whose commutators with
arbitrary elements from~$\mathfrak i_1$ belong to~$\mathfrak i_2$ is also a megaideal of~$\mathfrak g$.
\end{lemma}

Since the algebra~$\mathfrak g$ is solvable the radical~$\mathfrak r$ of~$\mathfrak g$ coincides with~$\mathfrak g$,
\[\mathfrak m_1:=\mathfrak g=
\langle \mathcal P^t,
\mathcal P^y,
\mathcal P^x(\chi),
\mathcal J^1,
\dots,
\mathcal J^m,
\mathcal Z(\kappa)
\rangle.
\]
The next obvious proper megaideals are the derived algebra~$\mathfrak g'$,
the nilradical~$\mathfrak n$ of~$\mathfrak g$,
its derived algebra $\mathfrak n'$
and the center~$\mathfrak z$ of $\mathfrak g$,
which respectively are
\begin{gather*}
\mathfrak g'
=\langle\mathcal P^x(\chi),\mathcal Z(\kappa)\rangle,
\quad
\mathfrak m_2:=\mathfrak n
=\langle\mathcal P^y,\mathcal P^x(\chi),\mathcal J^1,\dots,\mathcal J^m,\mathcal Z(\kappa)\rangle,
\\
\mathfrak n'
=\langle\mathcal Z(\kappa)\rangle,
\quad
\mathfrak m_3:=\mathfrak z
=\langle\mathcal P^x(1),\mathcal J^1,\dots,\mathcal J^m\rangle.
\end{gather*}
Pairwise intersecting the above megaideals, we get the next new ones
\begin{gather*}
\mathfrak m_4:=\mathfrak g'\cap\mathfrak m_3
=\langle\mathcal P^x(1),\mathcal Z(1)\rangle,
\quad
\mathfrak m_5:=\mathfrak n'\cap\mathfrak m_3
=\langle\mathcal Z(1)\rangle.
\end{gather*}
In notations of Lemma~\ref{lem:SpaceOfCommutElements}, let the megaideals~$\mathfrak i_0$ and $\mathfrak i_1$
be~$\mathfrak m_2$ and $\mathfrak m_1$, respectively.
Taking $\mathfrak i_2$ to be $\mathfrak m_5$
we obtain
\begin{gather*}
\mathfrak s=\langle\mathcal P^x(1),\mathcal P^x(t),\mathcal Z(1),\mathcal Z(t)\rangle,
\quad
\mathfrak m_6:=\mathfrak s\cap\mathfrak n'=\langle\mathcal Z(1),\mathcal Z(t)\rangle.
\end{gather*}
The found megaideals that are denoted by $\mathfrak m_i$, $i\in\{1,\dots,6\}$
are essential in computing the point symmetry pseudogroup $G$.

\begin{theorem}\label{thm:mLaysPointSymGroup}
The point symmetry pseudogroup $G$ of the system~\eqref{eq:mLaysMod}
consists of the point transformations of the form
\begin{gather*}
\tilde t    =\varepsilon_1t+T^0,\quad
\tilde x    =\varepsilon_1x+h(t),\quad
\tilde y    =\varepsilon_2y+Y^0,\quad
\tilde{\psi}^i=\varepsilon_2\psi^i-\varepsilon_1\varepsilon_2h_t(t)y+g(t)+\Psi^i,
\end{gather*}
where $\varepsilon_1,\varepsilon_2=\pm1$, $T^0$, $Y^0$ and $\Psi^i$, $i=1,\dots,m$, are arbitrary real constants,
and $h$ and~$g$ are arbitrary smooth functions of~$t$.
\end{theorem}

\begin{proof}
The general form of point transformations acting on the space $\mathbb R^3_{t,x,y}\times\mathbb R^m_\psi$ is
\begin{gather*}
\Phi\colon\quad
(\tilde t,\tilde x,\tilde y,\tilde\psi^1,\dots,\tilde\psi^m)=
(T,X,Y,\Psi^1,\dots,\Psi^m),
\end{gather*}
where $T$, $X$, $Y$, $\Psi^1$, \dots, $\Psi^m$ are smooth functions of $t$, $x$, $y$, $\psi^1$, \dots, $\psi^m$,
whose Jacobian is nonzero.
If $\Phi$ is a point-symmetry transformation of the system~\eqref{eq:mLaysMod},
then $\Phi_*\mathfrak m_i\subseteq\mathfrak m_i$ for all $i\in\{1,\dots,6\}$,
where $\Phi_*$ denotes the pushforward of vector fields by~$\Phi$.
We use this necessary condition for deriving principal constraints on the components of such transformations.

We choose
$\mathcal Z(1)$, $\mathcal Z(t)$, $\mathcal P^x(1)$, $\mathcal J^i$, $\mathcal P^y$ and $\mathcal P^t$
as linearly independent vector fields from $\mathfrak g$ for checking the above necessary condition.
Since
$\mathcal Z(1)\in\mathfrak m_5$,
$\mathcal Z(t)\in\mathfrak m_6$,
$\mathcal P^x(1)\in\mathfrak m_4$,
$\mathcal J^i\in\mathfrak m_3$,
$\mathcal P^y\in\mathfrak m_2$ and
$\mathcal P^t\in\mathfrak m_1$,
the above necessary condition implies that
\begin{subequations}\label{eq:PushFwd}
\begin{gather}
\Phi_*\mathcal Z(1)   =\tilde{\mathcal Z}(\tilde\kappa^1),\label{eq:PushFwdQ1}
\\
\Phi_*\mathcal Z(t)   =\tilde{\mathcal Z}(\tilde\kappa^2)+\tilde{\mathcal Z}(\hat\kappa^2\tilde t),\label{eq:PushFwdQ2}
\\
\Phi_*\mathcal P^x(1)   =\tilde{\mathcal P}^x(\tilde\chi^3)+\tilde{\mathcal Z}(\tilde\kappa^3),\label{eq:PushFwdQ3}
\\
\Phi_*\mathcal J^i=\tilde{\mathcal P}^x(\tilde\chi^4)+\alpha^{4j}\tilde{\mathcal J}^j,\label{eq:PushFwdQ4}
\\
\Phi_*\mathcal P^y   =\beta^5\tilde{\mathcal P}^y+\tilde{\mathcal P}^x(\tilde\chi^5)
+\alpha^{5j}\tilde{\mathcal J}^j+\tilde{\mathcal Z}(\tilde\kappa^5),\label{eq:PushFwdQ5}
\\
\Phi_*\mathcal P^t   =\gamma^6\tilde{\mathcal P}^t+\beta^6\tilde{\mathcal P}^y
+\tilde{\mathcal P^x}(\tilde\chi^6)+\alpha^{6j}\tilde{\mathcal J}^j
+\tilde{\mathcal Z}(\tilde\kappa^6),\label{eq:PushFwdQ6}
\end{gather}
\end{subequations}
where $\tilde\chi^i$ and $\tilde\kappa^i$, $i\in\{5,6\}$,
are smooth functions of $\tilde t$, the other parameters are real constants,
and the tildes in the notation of the vector fields on the right-hand sides mean
that these vector fields are defined in the coordinates with tildes.
We expand the system~\eqref{eq:PushFwd} and split each of its equations componentwise
In total, we obtain the following collection of constraints for the components of $\Phi$:
\begin{gather}\label{eq:ConstraintsOnTrnsfG}
\begin{split}
&T_{tt}=0,\quad
T_x=T_y=T_{\psi^i}=0,
\\
&X_x=X^1,\quad
X_y=X^2,\quad
X_{\psi^i}=X^{0i},
\\
&Y_t=Y^2,\quad
Y_y=Y^1,\quad
Y_x=0,\quad
Y_{\psi^i}=0
\\
&\Psi^i_{\psi^j}=A^{ij},\quad
\Psi^i_x=A^0,\quad
\Psi^i_{yy}=0,\quad
\Psi^i_{ty}=-T_t^{-1}X_{tt}Y_y,
\end{split}
\end{gather}
where $X^1$, $X^2$, $X^{0i}$, $Y^1$, $Y^2$, $A^{ij}$ and $A^0$ are arbitrary real constants with
\begin{gather}\label{eq:ConstraintsOnTrnsfGB}
\sum_{k=1}^m X^{0k}=0,\quad
\sum_{k=1}^mA^{ik}=\sum_{k=1}^mA^{jk}.
\end{gather}
The solution set of this collection of constraints consists of the tuples with components of the form
\begin{gather}\label{eq:ConstraintsOnPhi}
\begin{split}
&T=T^1t+T^0,\quad
X=X^1x+X^2y+X^{0i}\psi^i+h(t),\quad
Y=Y^1y+Y^2t+Y^0,\\
&\Psi^i=A^{ij}\psi^j+A^0x-\frac{Y^1}{T^1}f_ty+g(t)+b_i,
\end{split}
\end{gather}
where in addition to the above,
$T^0$ and $T^1$ are arbitrary real constants,
$h$ and~$g$ are arbitrary smooth functions of~$t$,
and the nondegeneracy of the Jacobian matrix of~$\Phi$ requires that $T^1X^1Y^1\det A^{ij}\ne0$.

One more further specification for the obtained form~\eqref{eq:ConstraintsOnTrnsfG}
for components of the arbitrary point-symmetry transformation~$\Phi$
can be obtained using the megaideal~$\mathfrak g'$,
\begin{gather}\label{eq:ConstraintsOnPhiB}
X^1Y^1=T^1\sum_{k=1}^mA^{ik}.
\end{gather}

Under the derived restrictions on the components of $\Phi$, we have $\Phi_*\mathfrak g\subseteq\mathfrak g$.
This is why the further application of the algebraic method does not result in more constraints
on the components of~$\Phi$.

We proceed with applying the direct method.
In view of the obtained form~\eqref{eq:ConstraintsOnPhi} for~$\Phi$,
the operators of partial derivatives are transformed according to the rule
\begin{gather*}
\p_{\tilde t}=\frac1{{\rm D}_tT}\left({\rm D}_t-\frac{{\rm D}_tX}{{\rm D}_xX}{\rm D}_x
-\frac{{\rm D}_tY}{{\rm D}_yY}\left({\rm D}_y-\frac{{\rm D}_yX}{{\rm D}_xX}{\rm D}_x\right)\right),
\\
\p_{\tilde x}=\frac1{{\rm D}_xX}{\rm D}_x,\quad
\p_{\tilde y}=\frac1{{\rm D}_yY}\left({\rm D}_y-\frac{{\rm D}_yX}{{\rm D}_xX}{\rm D}_x\right).
\end{gather*}
We label by ($\tilde{\ref{eq:mLaysMod}}$) the system~\eqref{eq:mLaysMod} written in the variables with tildes.
The system obtained from~($\tilde{\ref{eq:mLaysMod}}$) by substituting
the expressions for jet variables with tildes in terms of jet variables without tildes
and then substituting the expressions for $\psi^i_{txx}$ in view of the system~\eqref{eq:mLaysMod}
is called the expanded system~($\tilde{\ref{eq:mLaysMod}}$).
The derivative $\psi^i_{txy}$ appears only in $i$th equation of this system in the term
arising from~$\tilde\psi^i_{\tilde t\tilde y\tilde y}$,
\begin{gather*}
\tilde\psi^i_{\tilde t\tilde y\tilde y}=-\frac2{T^1(Y^1)^2}\frac{{\rm D}_yX}{{\rm D}_xX}A^{ij}\psi^j_{txy}+\cdots.
\end{gather*}
Hence ${\rm D}_yX=0$ since $\det A^{ij}\ne0$, and thus $X^2=X^{0i}=0$, which simplifies the expressions
for the transformed operators of partial derivatives,
\begin{gather*}
\p_{\tilde t}=\frac1{T^1}\left({\rm D}_t-\frac{f_t}{X^1}{\rm D}_x-\frac{Y^2}{Y^1}{\rm D}_y\right),\quad
\p_{\tilde x}=\frac1{X^1}{\rm D}_x,\quad
\p_{\tilde y}=\frac1{Y^1}{\rm D}_y.
\end{gather*}

Collecting the coefficients of $\psi^j_{txx}$
in the $i$th equation of the expanded system~($\tilde{\ref{eq:mLaysMod}}$) with $i\ne j$
gives the constraint $A^{ij}=0$.
Thus, the equation~\eqref{eq:ConstraintsOnPhiB} implies $A^{ii}=X^1Y^1/T^1$.
More constraints on $\Phi$,
\begin{gather*}
(f_{i,i-1}+f_{i,i+1})\frac{Y^2}{Y^1}=0,\quad
f_{i,i-1}\frac{A^0}{X^1Y^1}=0,
\end{gather*}
are obtained after collecting, in the $i$th equation of the expanded system~($\tilde{\ref{eq:mLaysMod}}$) for fixed~$i$,
the coefficients of~$\psi^i_y$ and~$\psi^{i-1}_y$, which arise from the terms
\smash{$(f_{i,i-1}+f_{i,i+1})\tilde\psi^i_{\tilde t}$}
and \smash{$f_{i,i-1}\tilde\psi^i_{\tilde x}\tilde\psi^{i-1}_{\tilde y}$}, respectively.
Therefore, $Y^2=A^0=0$.
Under the derived constraints,
the expanded system~($\tilde{\ref{eq:mLaysMod}}$) takes the form
\begin{gather}\label{eq:MappedExpandedSystForPSG}
\begin{split}
&(\psi^i_{tyy}+\psi^i_x\psi^i_{yyy}-\psi^i_y\psi^i_{yyx})\big((X^1)^2(Y^1)^{-2}-1\big)
\\
&\qquad{}
+f_{i,i-1}\big(\psi^i_t-\psi^{i-1}_t+\psi^{i-1}_x\psi^i_y-\psi^{i}_x\psi^{i-1}_y\big)\big((X^1)^2-1\big)
\\&\qquad{}
+f_{i,i+1}\big(\psi^i_t-\psi^{i+1}_t+\psi^{i+1}_x\psi^i_y-\psi^{i}_x\psi^{i+1}_y\big)\big((X^1)^2-1\big)
+\beta\psi^i_x(T^1X^1-1)=0.
\end{split}
\end{gather}
The obvious total splitting of~\eqref{eq:MappedExpandedSystForPSG} implies that $(X^1)^2=(Y^1)^2=1$ and $T^1=X^1$.
\end{proof}

\begin{corollary}\label{cor:DiscrTransfs}
The quotient of the pseudogroup $G$ by its identity component $G_{\rm id}$ is isomorphic to the group $\mathbb Z_2^2$
and can be naturally identified with the subgroup of $G$ generated by two involutions
$\mathscr I^{tx}$ and $\mathscr I^{y\psi}$,
\begin{gather*}
\mathscr I^{tx}\colon\quad
(t,x,y,\psi^1,\dots,\psi^m)\mapsto(-t,-x,y,\psi^1,\dots,\psi^m),
\\
\mathscr I^{y\psi}\colon\quad
(t,x,y,\psi^1,\dots,\psi^m)\mapsto(t,x,-y,-\psi^1,\dots,-\psi^m).
\end{gather*}
\end{corollary}
The pseudosubgroups of the pseudogroup~$G$
each of which is parameterized by a single functional or constant canonical parameter
and associated with a family of vector fields listed in Theorem~\ref{thm:mLaysLieInvAlg}
are of the form
\begin{gather}\label{eq:ElementTransfs}
\begin{split}
\mathscr P^t(\epsilon)&:=
(t,x,y,\psi^1,\dots,\psi^m)\mapsto(t+\epsilon,x,y,\psi^1,\dots,\psi^m),
\\
\mathscr P^x(h)&:=
(t,x,y,\psi^1,\dots,\psi^m)\mapsto(t,x+h(t),y,\psi^1-h_t(t)y,\dots,\psi^m-h_t(t)y),
\\
\mathscr P^y(\epsilon)&:=
(t,x,y,\psi^1,\dots,\psi^m)\mapsto(t,x,y+\epsilon,\psi^1,\dots,\psi^m),
\\
\mathscr J^i(\epsilon)&:=
(t,x,y,\psi^1,\dots,\psi^i,\dots,\psi^m)\mapsto(t,x,y,\psi^1,\dots,\psi^i+\epsilon,\dots,\psi^m),
\\
\mathscr Z(g)&:=
(t,x,y,\psi^1,\dots,\psi^m)\mapsto(t,x,y,\psi^1+g(t),\dots,\psi^m+g(t)),
\end{split}
\end{gather}
where the constant~$\epsilon$ plays the role of constant group parameter,
and $h$ and~$g$ are arbitrary smooth functions of~$t$.
The transformations~$\mathscr P^t(\epsilon)$, $\mathscr P^y(\epsilon)$ and~$\mathscr P^x(h)$
are respectively shifts with respect to~$t$ and~$y$
and generalized time-dependent shifts with respect to~$x$.
The transformations~$\mathscr J^i(\epsilon)$ and~$\mathscr Z(g)$ correspond to
gaugings of the stream function tuple,
where gauging summands for the stream functions on different layers differ at most by constants.

\section{Equivalence groupoid, pseudogroup and algebra}\label{sec:GroupClassification}

Theorems~\ref{thm:mLaysLieInvAlg} and~\ref{thm:mLaysPointSymGroup} provide some additional insights
about transformational properties of the system~\eqref{eq:mLaysMod}
when it is viewed as a class~$\mathcal M$ (of systems) of differential equations,
which is defined in Section~\ref{sec:mLaysMod}.
See \cite{opan2017a,opan2020b,opan2022a, vane2020b}
for the required theoretical background on classes of differential equations
and on point transformations within them.
More specifically, Theorem~\ref{thm:mLaysLieInvAlg} implies that
the algebra $\mathfrak g$ actually coincides with
the kernel of the maximal Lie invariance algebras of systems from the class~$\mathcal M$.
Furthermore, the trivial prolongation of the kernel algebra of a class of differential equations
to its arbitrary elements is invariant under the action of the corresponding usual equivalence (pseudo)group,
see~\cite[Section~2.2]{opan2022a} for all the necessary notions and the proof of the mentioned property.
This observation allows us to compute the equivalence pseudogroup $G^\sim$ of the class~$\mathcal M$
using the megaideal-based version of the algebraic method, which was suggested in~\cite{bihl2015a}.
At the same time, the algebra $\mathfrak g$ is the maximal Lie invariance algebra
for each system from the class~$\mathcal M$,
and hence it is preserved when pushed forward by the transformational part of any admissible transformation
within this class.
This provides an algebraic tool for computing the equivalence groupoid~$\mathcal G^\sim$ of~$\mathcal M$.
A nice feature of both computations, of~$G^\sim$ and of~$\mathcal G^\sim$,
is that their major part repeats the proof of Theorem~\ref{thm:mLaysPointSymGroup}.

\begin{theorem}\label{thm:ClassMEquivGroupoid}
The equivalence groupoid $\mathcal G^\sim$ of the class $\mathcal M$ consists of the triples of the form
$(\theta,\Phi,\tilde\theta)$,
where $\Phi$ is a point transformation in the space of independent and dependent variables,
whose components are
\begin{subequations}\label{eq:ClassMEquivGroupoid}
\begin{gather}\label{eq:ClassMEquivGroupoidA}
\tilde t    =T^1t+T^0,\quad
\tilde x    =X^1x+h(t),\quad
\tilde y    =\varepsilon X^1y+Y^0,
\\\label{eq:ClassMEquivGroupoidB}
\tilde{\psi}^i=\varepsilon\frac{(X^1)^2}{T^1}\psi^i-\varepsilon\frac{X^1}{T^1}h_t(t)y+g(t)+\Psi^i,
\end{gather}
and the target and source arbitrary-element tuples $\tilde\theta=(\tilde f_{i,i-1},\tilde f_{i,i+1},\tilde\beta)$
and $\theta=(f_{i,i-1},f_{i,i+1},\beta)$ are related according to
\begin{gather}\label{eq:ClassMEquivGroupoidC}
\tilde f_{i,i-1}=\frac{f_{i,i-1}}{(X^1)^2},\quad
\tilde f_{i,i+1}=\frac{f_{i,i+1}}{(X^1)^2},\quad
\tilde\beta=\frac\beta{T^1X^1}.
\end{gather}
\end{subequations}
Here $\varepsilon=\pm1$, $T^1$, $T^0$, $X^1$ $Y^0$ and $\Psi^i$, $i=1,\dots,m$, are arbitrary real constants satisfying $T^1X^1\ne0$,
and $h$ and~$g$ are arbitrary smooth functions of~$t$.
\end{theorem}

\begin{proof}
Given any two arbitrary-element tuples $\theta$ and $\tilde\theta$, let a point transformation
\begin{gather*}
\Phi\colon\quad
(\tilde t,\tilde x,\tilde y,\tilde\psi^1,\dots,\tilde\psi^m)=
(T,X,Y,\Psi^1,\dots,\Psi^m)
(t,x,y,\psi^1,\dots,\psi^m)
\end{gather*}
connect the corresponding systems~$\mathcal M_\theta$ and~$\mathcal M_{\tilde\theta}$ from the class $\mathcal M$.
In light of Theorem~\ref{thm:mLaysLieInvAlg}, the algebra $\mathfrak g$
is the maximal Lie invariance algebra of each system from the class $\mathcal M$.
Therefore, $\Phi_*\mathfrak g=\mathfrak g$, which thus implies $\Phi_*\mathfrak m=\mathfrak m$
for every megaideal $\mathfrak m\subset\mathfrak g$.
This returns us to the setting of the algebraic part of the proof of Theorem~\ref{thm:mLaysPointSymGroup},
which straightforwardly implies that the components of the transformation $\Phi$
are necessarily of the form~\eqref{eq:ConstraintsOnPhi},
where the parameters satisfy the restrictions, indicated after this form, including~\eqref{eq:ConstraintsOnPhiB}.

The next step in the course of computing the explicit form of $\Phi$ is to apply the direct method.
Within the framework of this method, we should express the transformed derivatives in terms of the initial ones
and substitute the obtained expressions into the target system $\mathcal M_{\tilde\theta}$
and then substitute the expressions for $\psi^i_{txx}$ in view of the source system~$\mathcal M_\theta$.
The computations in this part are completely analogous to those carried out
in the proof of Theorem~\ref{thm:mLaysPointSymGroup}.
The only difference is that the expanded target system $\mathcal M_{\tilde\theta}$,
which is the counterpart of~\eqref{eq:MappedExpandedSystForPSG},
takes the form
\begin{gather}\label{eq:MappedExpandedSystForGroupoid}
\begin{split}
&(\psi^i_{tyy}+\psi^i_x\psi^i_{yyy}-\psi^i_y\psi^i_{yyx})\big((X^1)^2(Y^1)^{-2}-1\big)
\\
&\qquad{}
+\big(\psi^i_t-\psi^{i-1}_t+\psi^{i-1}_x\psi^i_y-\psi^{i}_x\psi^{i-1}_y\big)\big(\tilde f_{i,i-1}(X^1)^2-f_{i,i-1}\big)
\\&\qquad{}
+\big(\psi^i_t-\psi^{i+1}_t+\psi^{i+1}_x\psi^i_y-\psi^{i}_x\psi^{i+1}_y\big)\big(\tilde f_{i,i+1}(X^1)^2-f_{i,i+1}\big)
+\psi^i_x(\tilde\beta T^1X^1-\beta)=0.
\end{split}
\end{gather}
This results in the remaining constraints,
$(X^1)^2=(Y^1)^2$ and~\eqref{eq:ClassMEquivGroupoidC}.
\end{proof}

\begin{theorem}\label{thm:ClassMEquivGroup}
The usual equivalence pseudogroup~$G^\sim$ of the class~$\mathcal M$ consists of the point transformations
in the foliated space \smash{$\mathbb R^{m+3}_{t,x,y,\psi}\times\mathbb R^{2m-1}_\theta$}
with components of the form~\eqref{eq:ClassMEquivGroupoid},
where $\varepsilon=\pm1$, $T^1$, $T^0$, $X^1$ $Y^0$ and $\Psi^i$, $i=1,\dots,m$,
are arbitrary real constants with $T^1X^1\ne0$,
and $f$ and~$g$ are arbitrary smooth functions of~$t$.
\end{theorem}

\begin{proof}
We start with the general form of fibre-preserving point transformation $\mathscr T$
acting on the foliated space $\mathbb R^{m+3}_{t,x,y,\psi}\times\mathbb R^{2m-1}_\theta$,
\noprint{
$\mathbb R^3_{t,x,y}\times\mathbb R^m_\psi\times\mathbb R^{2(m-1)}_{f_{i,j}}\times\mathbb R_\beta$
}
\begin{align*}
\mathscr T\colon\quad
&(\tilde t,\tilde x,\tilde y,\tilde\psi^1,\dots,\tilde\psi^m)=
(T,X,Y,\Psi^1,\dots,\Psi^m)
(t,x,y,\psi^1,\dots,\psi^m),
\\
&(\tilde f_{i,i-1},\tilde f_{i,i+1},\tilde\beta)=(F^{i,i-1},F^{i,i+1},B)
(t,x,y,\psi^1,\dots,\psi^m,f_{i,i-1},f_{i,i+1},\beta),
\end{align*}
where $i=1,\dots,m$ with the excluded index pairs $(1,0)$ and~$(m,m+1)$.
Denote by hat the trivial prolongation of objects defined in the space~$\mathbb R^{m+3}_{t,x,y,\psi}$
to the space~$\mathbb R^{m+3}_{t,x,y,\psi}\times\mathbb R^{2m-1}_\theta$.
Theorem~\ref{thm:mLaysLieInvAlg} implies that the algebra $\mathfrak g$
coincides with the kernel of maximal Lie invariance algebras of systems from the class $\mathcal M$.
This is why the trivial prolongation~$\hat{\mathfrak g}$ of $\mathfrak g$ is preserved
under the adjoint action of $G^\sim$, and thus $G^\sim$ also preserves the megaideals of $\hat{\mathfrak g}$.
Hence the form of the $(t,x,y,\psi)$-components of~$\mathscr T$ is given by~\eqref{eq:ConstraintsOnPhi},
where the parameters satisfy the equations~\eqref{eq:ConstraintsOnPhiB} and~\eqref{eq:ConstraintsOnTrnsfGB}
the inequality $T^1X^1Y^1\det A^{ij}\ne0$.
In notation of the proof of Theorem~\ref{thm:mLaysPointSymGroup},
the consideration of the necessary conditions
$\mathscr T_*\hat{\mathcal P}^x(1)\in\hat{\mathfrak m}_4$,
$\mathscr T_*\hat{\mathcal J}^i\in\hat{\mathfrak m}_3$,
$\mathscr T_*\hat{\mathcal P}^y\in\hat{\mathfrak m}_2$ and
$\mathscr T_*\hat{\mathcal P}^t\in\hat{\mathfrak m}_1$
also results in the following equations for the $\theta$-components of~$\mathscr T$:
\begin{gather*}
F^{i,i-1}_t=F^{i,i-1}_x=F^{i,i-1}_y=F^{i,i-1}_{\psi^j}=0,\quad
F^{i,i+1}_t=F^{i,i+1}_x=F^{i,i+1}_y=F^{i,i+1}_{\psi^j}=0,
\\
B_t=B_x=B_y=B_{\psi^j}=0,
\end{gather*}
i.e., $F^{i,i-1}=F^{i,i-1}(f_{i,i-1},f_{i,i+1},\beta)$ and $B=B(f_{i,i-1},f_{i,i+1},\beta)$.
The next step is to apply the direct method to find the remaining restrictions on the components on $\mathscr T$.
These computations are carried out analogously to those
in the last part of the proof of Theorem~\ref{thm:ClassMEquivGroupoid},
with the only difference that the parameters $\tilde f_{i,i-1}$, $\tilde f_{i,i+1}$ and $\tilde \beta$ in~\eqref{eq:MappedExpandedSystForGroupoid}
are replaced with $F^{i,i-1}$, $F^{i,i+1}$ and $B$, respectively.
\end{proof}

\begin{remark}
Under the physical constraint $\beta>0$, we should set $T^1X^1>0$
in Theorems~\ref{thm:ClassMEquivGroupoid} and~\ref{thm:ClassMEquivGroup}.
\end{remark}

\begin{corollary}\label{cor:ClassMNorm}
The class~$\mathcal M$ is normalized in the usual sense.
\end{corollary}

\begin{proof}
Theorems~\ref{thm:ClassMEquivGroupoid} and~\ref{thm:ClassMEquivGroup} jointly imply
that the action groupoid of the equivalence group $G^\sim$
coincides with the equivalence groupoid~$\mathcal G^\sim$.
According to~\cite[Definition~2]{vane2020b}, this means
that the class~$\mathcal M$ is normalized in the usual sense.
\end{proof}

\begin{remark}
We could also derive Theorem~\ref{thm:ClassMEquivGroup}
as a corollary of Theorem~\ref{thm:ClassMEquivGroupoid}
since analyzing the equivalence groupoid~$\mathcal G^\sim$ leads to
the first claim in the proof of Corollary~\ref{cor:ClassMNorm},
which implies the form of equivalence transformations.
\end{remark}

\begin{remark}\label{rem:GeneralizedEquivGroup}
The generalized equivalence group of the class~$\mathcal M$ consists of the point transformations
in the foliated space \smash{$\mathbb R^{m+3}_{t,x,y,\psi}\times\mathbb R^{2m-1}_\theta$}
with components of the form~\eqref{eq:ClassMEquivGroupoid},
where all parameters are allowed to depend on the arbitrary-element tuple~$\theta$.
The effective generalized equivalence group of~$\mathcal M$ coincides with
its usual equivalence group~$G^\sim$.
\end{remark}

\begin{corollary}\label{cor:ClassMEquivAlgebra}
The equivalence algebra $\mathfrak g^\sim$ of the class~$\mathcal M$ of systems of the form~\eqref{eq:mLaysMod}
is spanned by the vector fields
\begin{gather*}
\hat{\mathcal P}^t,\quad
\hat{\mathcal P}^y,\quad
\hat{\mathcal P}^x(\chi),\quad
\hat{\mathcal J}^1,\quad
\dots,\quad
\hat{\mathcal J^m},\quad
\hat{\mathcal Z}(\kappa),\quad
t\p_t-(\psi^1\p_{\psi^1}+\dots+\psi^m\p_{\psi^m})-\beta\p_{\beta},\\
x\p_x+y\p_y+2(\psi^1\p_{\psi^1}+\dots+\psi^m\p_{\psi^m})-2\sum_{i=1}^m(f_{i,i-1}\p_{f_{i,i-1}}+f_{i,i+1}\p_{f_{i,i+1}})-\beta\p_\beta,
\end{gather*}
where $\chi(t)$ and $\kappa(t)$ are arbitrary smooth functions of their argument.
\end{corollary}

We can summarize the results of Section~\ref{sec:LieInvAlg} and this section
as the solution of the group classification problem for the class~$\mathcal M$.

\begin{theorem}\label{thm:ClassMGroupClassification}
The kernel of the maximal Lie invariance algebras of systems from the class~$\mathcal M$
coincides with the algebra~$\mathfrak g$.
There are no Lie symmetry extensions in this class,
i.e., each system from~$\mathcal M$ is invariant exactly with respect to the algebra~$\mathfrak g$.
The usual equivalence pseudogroup~$G^\sim$ of the class~$\mathcal M$
is described in Theorem~\ref{thm:ClassMEquivGroup}.
As canonical representatives of $G^\sim$-equivalence classes of systems from the class~$\mathcal M$,
one can choose the systems with $\beta=1$ and $f_{12}$ equal to $1$ or $\pm1$
depending on which condition, of the positivity or of nonvanishing, is assumed for the arbitrary elements.
\end{theorem}

\section{Classification of subalgebras}\label{sec:Subalgebras}

Since the multi-layer quasi-geostrophic model~\eqref{eq:mLaysMod}
is a system of partial differential equations with three independent variables,
to classify its Lie reductions
in an optimal way it suffices to classify one- and two-dimensional Lie subalgebras
of its Lie invariance algebra $\mathfrak g$
with respect to the adjoint action of its point-symmetry pseudogroup $G$.
See~\cite[Section~2]{vinn2024a} for a description of the optimized procedure of Lie reductions
in the case of three independent variables
and Remark~\ref{rem:mLaysMod3DSubalgs} below
for the exhaustive study of codimension-three Lie reductions of the system~\eqref{eq:mLaysMod}.
The pseudogroup $G$ is generated by
the continuous point transformations listed in~\eqref{eq:ElementTransfs} and
the discrete point transformations listed in Corollary~\ref{cor:DiscrTransfs}.
They act on~$\mathfrak g$ via pushing forward its elements
as vector fields in the space $\mathbb R^3_{t,x,y}\times\mathbb R^m_\psi$,
and the nonidentity actions on the elements spanning~$\mathfrak g$ are exhausted by the following:
\begin{gather*}
\mathscr P^x(h)_*\mathcal P^t=\mathcal P^t+\mathcal P^x(h_t),
\quad
\mathscr Z(g)_*\mathcal P^t=\mathcal P^t+\mathcal Z(g_t),
\\
\mathscr P^x(h)_*\mathcal P^y=\mathcal P^y-\mathcal Z(h_t),
\quad
\mathscr P^y(\epsilon)_*\mathcal P^x(\chi)=\mathcal P^x(\chi)+\epsilon\mathcal Z(\chi_t),
\\
\mathscr P^t(\epsilon)_*\mathcal Z(\kappa)=\mathcal Z(\tilde\kappa^\epsilon),\quad
\mathscr P^t(\epsilon)_*\mathcal P^x(\chi)=\mathcal P^x(\tilde\chi^\epsilon),
\\
\mathscr I^{tx}_*\mathcal P^t=-\mathcal P^t,\quad
\mathscr I^{tx}_*\mathcal P^x(\chi)=-\mathcal P^x(\hat\chi),\quad
\mathscr I^{tx}_*\mathcal Z(\kappa)=\mathcal Z(\hat\kappa),
\\
\mathscr I^{y\psi}_*\mathcal P^y=-\mathcal P^y,\quad
\mathscr I^{y\psi}_*\mathcal J^k=-\mathcal J^k,\quad
\mathscr I^{y\psi}_*\mathcal Z(\kappa)=-\mathcal Z(\kappa),
\end{gather*}
where
$\tilde\kappa^\epsilon(t):=\kappa(t-\epsilon)$,
$\tilde\chi^\epsilon(t):=\chi(t-\epsilon)$,
$\hat\chi(t):=\chi(-t)$ and $\hat\kappa(t):=\kappa(-t)$.

The algebra $\mathfrak g$ can be represented as the semidirect sum of its subalgebra
$\mathfrak g_1=\langle\mathcal P^t,\mathcal P^y\rangle$ acting on the ideal
$\mathfrak g_2=\langle\mathcal P^x(\chi),\mathcal J^1,\dots,\mathcal J^m,\mathcal Z(\kappa)\rangle$.
Using this decomposition, we can apply the method for classifying subalgebras of semidirect sums of Lie algebras,
which was suggested in~\cite{pate1975a},
see also~\cite{chap2024b} and references therein for the modern view on this method.
Since each of the Lie algebras~$\mathfrak g_1$ and~$\mathfrak g_2$ is abelian,
its optimal list of subalgebras with respect to its inner automorphisms
coincides with the list of its subspaces.
Denote by $\pi_1$ the canonical projection $\pi_1\colon\mathfrak g\to\mathfrak g_1$
with respect to the vector space decomposition $\mathfrak g=\mathfrak g_1\dotplus\mathfrak g_2$.

\begin{theorem}\label{thm:1D-subalgebras}{}
A complete list of $G$-inequivalent one-dimensional subalgebras of the algebra~$\mathfrak g$
spanned by the vector fields~\eqref{eq:mLaysLieInvAlg} is exhausted by the following families.
Below $c:=(c_1,\dots,c_m)^{\mathsf T}$ and~$a$ are arbitrary constants with $c\in\mathop{\rm im}\mathsf F$
and $\chi$ and~$\kappa$ are arbitrary smooth functions of~$t$
that satisfy the indicated constraints.

\begin{enumerate}

\item
$\mathfrak s_{1.1}^{ac}:=
\langle\mathcal P^t+a\mathcal P^y+c_k\mathcal J^k\rangle$.
The subalgebras~$\mathfrak s_{1.1}^{ac}$ and~$\mathfrak s_{1.1}^{\tilde a\tilde c}$
are equivalent if and only if $\tilde a=\varepsilon a$ and $\tilde c=\varepsilon c$
for some $\varepsilon\in\{-1,1\}$.

\item
$\mathfrak s_{1.2}^{\chi c}:=
\langle\mathcal P^y+\mathcal P^x(\chi)+c_k\mathcal J^k\rangle$.
The subalgebras~$\mathfrak s_{1.2}^{\chi c}$ and~$\mathfrak s_{1.2}^{\tilde\chi\tilde c}$
are equivalent if and only if $\tilde\chi(t)=\varepsilon'\chi(\varepsilon t-\epsilon)$
and $\tilde c=c$
for some $\varepsilon,\varepsilon'\in\{-1,1\}$ and $\epsilon\in\mathbb R$.

\item
$\mathfrak s_{1.3}^{\chi c\kappa}:=
\langle \mathcal P^x(\chi)+c_k\mathcal J^k+\mathcal Z(\kappa)\rangle$
with nonzero parameter tuple $(\chi,c,\kappa)$.
The subalgebras~$\mathfrak s_{1.3}^{\chi c\kappa}$ and~$\mathfrak s_{1.3}^{\tilde\chi\tilde c\tilde\kappa}$
are equivalent if and only if
$\tilde\chi(t)=\varepsilon\beta\chi(\varepsilon t-\epsilon)$,
$\tilde\kappa(t)=\varepsilon'\beta\kappa(\varepsilon t-\epsilon)
+\epsilon'\beta\dot\chi(\varepsilon t-\epsilon)$
and $\tilde c=\varepsilon'\beta c$
for some $\varepsilon,\varepsilon'\in\{-1,1\}$, $\epsilon,\epsilon'\in\mathbb R$ and $\beta\in\mathbb R_{\ne0}$
with $\dot\chi$ denoting the derivative of~$\chi$.\looseness=-1

\end{enumerate}
\end{theorem}

\begin{proof}

Let $\mathfrak s\subset\mathfrak g$ be a one-dimensional subalgebra of~$\mathfrak g$.
Since $\dim\pi_1\mathfrak s\leqslant1$, modulo the $G$-equivalence we can set
\smash{$\pi_1\mathfrak s\in\big\{\langle\mathcal P^t+a\mathcal P^y\rangle,\langle\mathcal P^y\rangle,\{0\}\big\}$}.
We can set the gauge $a\geqslant 0$, acting by \smash{$\mathscr I^{y\psi}_*$} if necessary.
According to Remark~\ref{rem:mLayerMLIAGaugeTrans},
we can also assume from the very beginning in each case below
that $c\in\mathop{\rm im}\mathsf F$ in the fixed basis element~$Q$ of~$\mathfrak s$.

\medskip\par\noindent	
$\boldsymbol{\pi_1\mathfrak s=\langle\mathcal P^t+a\mathcal P^y\rangle}$.
Thus, we can choose
$Q=\mathcal P^t+a\mathcal P^y+c_i\mathcal J^i+\mathcal P^x(\chi)+\mathcal Z(\kappa)$.
Pushing forward~$Q$ by the transformation
$\mathscr Z\big(-\int(\kappa+a\chi)\,{\rm d}t\big)\circ\mathscr P^x\big(-\int\chi\,{\rm d}t\big)$,
we can set $\chi=\kappa=0$,
which results in the first subalgebra family in the theorem's statement.
The equivalence within this family is generated by
the pushforwards $\mathscr Z(\epsilon t)_*$ and $\mathscr I^{y\psi}_*$.

\medskip\par\noindent	
$\boldsymbol{\pi_1\mathfrak s=\langle\mathcal P^y\rangle}$.
Taking $Q=\mathcal P^y+c_i\mathcal J^i+\mathcal P^x(\chi)+\mathcal Z(\kappa)$ as the basis vector of $\mathfrak s$,
we simplify it by $\mathscr P^x\big(\int\kappa\,{\rm d}t\big)_*$, setting the gauge $\kappa=0$
and thus obtaining the second listed subalgebra family.
Furthermore, the pushforwards $\mathscr P^t(\epsilon)_*$, $\mathscr I^{tx}_*$
and  $\mathscr P^x(-\epsilon't)_*$ generate the equivalence within this family.

\medskip\par\noindent	
$\boldsymbol{\pi_1\mathfrak s=\{0\}}$.
In other words, $\mathfrak s\subset\mathfrak g_2$.
The one-dimensional subspaces of~$\mathfrak g_2$ constitute the third listed subalgebra family,
the equivalence between which is generated by the action of~$G$.
\end{proof}

\begin{theorem}\label{thm:2D-subalgebras}
A complete list of $G$-inequivalent two-dimensional subalgebras of the algebra~$\mathfrak g$
spanned by the vector fields~\eqref{eq:mLaysLieInvAlg} is exhausted by the following subalgebra families.
Below $b:=(b_1,\dots,b_m)^{\mathsf T}$, $c:=(c_1,\dots,c_m)^{\mathsf T}$,
$a$ and~$\sigma$ are arbitrary constants with $\sigma\ne0$,
$b\in\mathop{\rm im}\mathsf F$ and, except the first three algebras, $c\in\mathop{\rm im}\mathsf F$,
and $\chi$, $\chi^1$, $\chi^2$, $\kappa$, $\kappa^1$ and~$\kappa^2$
are arbitrary smooth functions of $t$ that satisfy the indicated constraints.
\begin{enumerate}\itemsep=0.5ex
\item
$
\mathfrak s_{2.1}^{abc}:=
\langle
\mathcal P^t+b_k\mathcal J^k,\,
\mathcal P^y+\mathcal P^x(a)+c_k\mathcal J^k
\rangle
$.
The subalgebras~$\mathfrak s_{2.1}^{abc}$ and~$\mathfrak s_{2.1}^{\tilde a\tilde b\tilde c}$
are equivalent if and only if $\tilde a=\varepsilon a$,
$\tilde b=\varepsilon b$ and $\tilde c=c$ for some $\varepsilon\in\{-1,1\}$.
		
\item
$
\mathfrak s_{2.2}^{abc}:=\langle
\mathcal P^t+a\mathcal P^y+b_k\mathcal J^k,\,
\mathcal P^x(1)+c_k\mathcal J^k
\rangle
$.
The subalgebras~$\mathfrak s_{2.2}^{abc}$ and~$\mathfrak s_{2.2}^{\tilde a\tilde b\tilde c}$
are equivalent if and only if $\tilde a=\varepsilon a$,
$\tilde b=\varepsilon b$ and
$\tilde c=\varepsilon c$
for some $\varepsilon\in\{-1,1\}$.
		
\item
$
\mathfrak s_{2.3}^{abc}:=\langle
\mathcal P^t+a\mathcal P^y+b_k\mathcal J^k,\,
c_k\mathcal J^k
\rangle
$,
where the tuple $c$ is nonzero.
The subalgebras~$\mathfrak s_{2.3}^{abc}$ and~\smash{$\mathfrak s_{2.3}^{\tilde a\tilde b\tilde c}$}
are equivalent if and only if $\tilde a=\varepsilon a$,
$\tilde b=\varepsilon b+\epsilon c$ and
$\tilde c=\gamma c$
for some $\varepsilon\in\{-1,1\}$, $\epsilon\in\mathbb R$ and $\gamma\in\mathbb R_{\ne0}$.
		
\item
$
\mathfrak s_{2.4}^{ab\sigma}:=\langle
\mathcal P^t+a\mathcal P^y+b_k\mathcal J^k,\,
\mathcal P^x({\rm e}^{\sigma t})+\mathcal Z(a\sigma t{\rm e}^{\sigma t})
\rangle
$.
The subalgebras~$\mathfrak s_{2.4}^{ab\sigma}$ and~$\mathfrak s_{2.4}^{\tilde a\tilde b\tilde\sigma}$
are equivalent if and only if $\tilde a=\varepsilon a$,
$\tilde b=\varepsilon b$
and $\tilde\sigma=\varepsilon'\sigma$		
for some $\varepsilon,\varepsilon'\in\{-1,1\}$.

\item
$
\mathfrak s_{2.5}^{ab\sigma}:=\langle
\mathcal P^t+a\mathcal P^y+b_k\mathcal J^k,\,
\mathcal Z({\rm e}^{\sigma t})
\rangle
$,
The subalgebras~$\mathfrak s_{2.5}^{ab\sigma}$ and~$\mathfrak s_{2.5}^{\tilde a\tilde b\tilde\sigma}$
are equivalent if and only if $\tilde a=\varepsilon a$,
$\tilde b=\varepsilon b$ and $\tilde\sigma=\varepsilon'\sigma$		
for some $\varepsilon,\varepsilon'\in\{-1,1\}$.

\item
$
\mathfrak s_{2.6}^{\chi bc\kappa}:=\langle
\mathcal P^y+\mathcal P^x(\chi)+b_k\mathcal J^k,\,
\mathcal P^x(1)+c_k\mathcal J^k+\mathcal Z(\kappa)
\rangle
$.
Such subalgebras with parameter values $(\chi,b,c,\kappa)$ and $(\tilde\chi,\tilde b,\tilde c,\tilde\kappa)$
are equivalent if and only if
$\tilde\chi(t)=\varepsilon\varepsilon'\chi(\varepsilon t-\epsilon)+\epsilon'$,
$\tilde b=b+\varepsilon'\epsilon'c$,
$\tilde c=\varepsilon'c$,
$\tilde\kappa(t)=\varepsilon'\kappa(\varepsilon t-\epsilon)$
for some $\varepsilon,\varepsilon'\in\{-1,1\}$ and $\epsilon,\epsilon'\in\mathbb R$.

\item
$
\mathfrak s_{2.7}^{\chi bc\kappa}:=\langle
\mathcal P^y+\mathcal P^x(\chi)+b_k\mathcal J^k,\,
c_k\mathcal J^k+\mathcal Z(\kappa)
\rangle
$,
where the tuple $(c,\kappa)$ is nonzero.
Such subalgebras with parameter values $(\chi,b,c,\kappa)$ and $(\tilde\chi,\tilde b,\tilde c,\tilde\kappa)$
are equivalent if and only if
$\tilde\chi(t)=\varepsilon'\chi(\varepsilon t-\epsilon)$,
$\tilde b=b+\epsilon'c$,
$\tilde c=\alpha c$,
$\tilde\kappa(t)=\alpha\kappa(\varepsilon t-\epsilon)$
for some $\varepsilon,\varepsilon'\in\{-1,1\}$, $\epsilon,\epsilon'\in\mathbb R$
and $\alpha\in\mathbb R_{\ne0}$.

\item
$\mathfrak s_{2.8}^\upsilon:=\big\langle
\mathcal P^x(\chi^1)+b_k\mathcal J^k+\mathcal Z(\kappa^1),\,
\mathcal P^x(\chi^2)+c_k\mathcal J^k+\mathcal Z(\kappa^2)\big\rangle$,
where $\upsilon:=(\chi^1,b,\kappa^1,\chi^2,c,\kappa^2)$, and
the subtuples $(\chi^1,b,\kappa^1)$ and $(\chi^2,c,\kappa^2)$ are linearly independent.
Such subalgebras with parameter values $\upsilon$ and $\tilde\upsilon$ are equivalent if and only if
\begin{gather*}
\tilde\chi^k(t)=\varepsilon\beta_{kl}\chi^l(\varepsilon t-\epsilon),\quad
\tilde\kappa^k(t)=\beta_{kl}
\big(\varepsilon'\kappa^l(\varepsilon t-\epsilon)+\epsilon'\chi^l_t(\varepsilon t-\epsilon)\big),
\\
\tilde b=\varepsilon'\beta_{11}b+\varepsilon'\beta_{12}c,\quad
\tilde c=\varepsilon'\beta_{21}b+\varepsilon'\beta_{22}c,
\end{gather*}
where $\varepsilon,\varepsilon'\in\{-1,1\}$, $\epsilon,\epsilon'\in\mathbb R$,
$(\beta_{kl})$ is a constant nondegenerate matrix,
and the indices~$k$ and~$l$ run from 1 to~2.

\end{enumerate}
\end{theorem}

\begin{proof}
Let $\mathfrak s$ be a two-dimensional subalgebra of $\mathfrak g$.
Modulo the $G$-equivalence we can set
$\pi_1\mathfrak s\in\{\langle\mathcal P^t,\mathcal P^y\rangle,\langle\mathcal P^t+a\mathcal P^y\rangle,\langle\mathcal P^y\rangle,\{0\}\}$.
Assume that $\mathfrak s$ is spanned by vector fields $Q_1$ and $Q_2$,
whose forms will be specified depending on the case under the consideration.

\medskip\par\noindent	
$\boldsymbol{\pi_1\mathfrak s=\langle\mathcal P^t,\mathcal P^y\rangle}$.
The subalgebra $\mathfrak s$ is spanned by the vector fields
\begin{gather*}
Q_1=\mathcal P^t+b_k\mathcal J^k+\mathcal P^x(\chi^1)+\mathcal Z(\kappa^1),
\quad
Q_2=\mathcal P^y+c_k\mathcal J^k+\mathcal P^x(\chi^2)+\mathcal Z(\kappa^2).
\end{gather*}
Pushing forward the basis vector fields by the transformation
$\mathscr Z\big(-\int\kappa^1{\rm d}t\big)\circ\mathscr P^x\big(-\int\chi^1{\rm d}t\big)$,
we can set $\kappa^1=\chi^1=0$.
Since $\pi_1\mathfrak s$ is abelian, the algebra $\mathfrak s$ is abelian as well.
The commutation relation $[Q_1,Q_2]=0$ implies $\chi^2_t=\kappa^2_t=0$.
Denoting the constants $c_i+\kappa^2$ by $c_i$ and~$\chi^2$ by~$a$,
we obtain the first subalgebra family listed in the theorem's statement.
The equivalence within this family is generated by
the pushforwards $\mathscr Z(\epsilon t)_*$ and $\mathscr I^{y\psi}_*$.

\medskip\par\noindent	
$\boldsymbol{\pi_1\mathfrak s=\langle\mathcal P^t+a\mathcal P^y\rangle}$.
We can assume the vector fields $Q_1$ and $Q_2$ take the following form:
\begin{gather*}
Q_1=\mathcal P^t+a\mathcal P^y+\mathcal P^x(\chi^1)+b_k\mathcal J^k+\mathcal Z(\kappa^1),
\quad
Q_2=\mathcal P^x(\chi^2)+c_k\mathcal J^k+\mathcal Z(\kappa^2).
\end{gather*}
Acting on $\mathfrak s$ by $\mathscr Z\big(-\int(\kappa^1+a\chi^1)\,{\rm d}t\big)_*\circ\mathscr P^x\big(-\int\chi^1\,{\rm d}t\big)_*$,
we set $\chi^1=\kappa^1=0$.

If the subalgebra $\mathfrak s$ is abelian,
the commutation relation $[Q_1,Q_2]=0$ implies $\chi^2_t=0$ and $\kappa^2_t=0$.
Denoting the constants $c_i+\kappa^2$ by $c_i$,
we obtain the second and the third listed subalgebra families
depending on whether the constant $\chi^2$ is nonzero or zero, respectively.
If $\chi^2\ne0$, then we can set $\chi^2=1$ by rescaling the vector field~$Q_2$,
and the equivalence within the corresponding family of subalgebras
is generated by the pushforward $\mathscr I^{y\psi}_*$.
Within the family with $\chi^2=0$, the equivalence additionally involves linear recombination of the basis elements.

If the subalgebra $\mathfrak s$ is nonabelian,
then the commutation relation $[Q_1,Q_2]=\sigma Q_2$ with nonzero~$\sigma$ gives the system
\[
c_1=\dots=c_m=0,\quad
\chi^2_t=\sigma\chi^2,\quad
\kappa^2_t-a\chi^2_t=\sigma\kappa^2.
\]
Hence $\chi^2=\beta{\rm e}^{\sigma t}$, $\kappa^2=(\beta a\sigma t+\delta){\rm e}^{\sigma t}$
for some real constants $\beta$ and $\delta$ with $(\beta,\delta)\ne(0,0)$.
Depending on whether $\beta$ is zero or not, we obtain the fifth and fourth listed families.
In the former case, we rescale~$Q_2$ to set $\delta=1$.
In the latter case, we can set $\delta=0$ and $\beta=1$
by the pushforward $\mathscr P^y(-\beta^{-1}\delta)_*$ and rescaling~$Q_2$.
The equivalence within each of these families is generated by
the pushforwards $\mathscr I^{tx}_*$ and $\mathscr I^{y\psi}_*$.

\medskip\par\noindent	
$\boldsymbol{\pi_1\mathfrak s=\langle\mathcal P^y\rangle}$.
In view of Theorem~\ref{thm:1D-subalgebras} and modulo linearly combining~$Q_1$ and~$Q_2$,
we can set $Q_1=\mathcal P^y+\mathcal P^x(\chi^1)+b_k\mathcal J^k$
and $Q_2=\mathcal P^x(\chi^2)+c_k\mathcal J^k+\mathcal Z(\kappa^2)$.
The commutation relation $[Q_1,Q_2]=\mathcal Z(\chi^2_t)$ implies $\chi^2_t=0$.
The cases with $\chi^2\ne0$ and $\chi^2=0$
correspond to the sixth and the seventh subalgebra families, respectively.
We also relabel $(\chi^1,\kappa^2)$ as $(\chi,\kappa)$ and,
for the sixth subalgebra family, set $\chi^2=1$ by scaling the vector field $Q_2$.
Note that in addition to the $G$-equivalence within each of these subalgebra families,
we can also replace $Q_1$ with the linear combination $Q_1+\gamma Q_2$, where $\gamma\in\mathbb R$,
and use the pushforward $\mathscr P^x\big(\int\gamma\kappa\,{\rm d}t\big)_*$
to make the $\mathcal Z$-summand in this vector field zero.

\medskip\par\noindent	
$\boldsymbol{\pi_1\mathfrak s=\{0\}}$.
This means that~$\mathfrak s$ is a subalgebra of the abelian algebra~$\mathfrak g_2$,
which gives us the eighth subalgebra family.
The equivalence within this subalgebra family is generated by the action of~$G$
and the linear recombination of the basis vector fields.
This equivalence is described in the theorem's statement.
\end{proof}

\begin{remark}\label{rem:mLaysMod3DSubalgs}
A three-dimensional subalgebra of~$\mathfrak g$ satisfies the transversality condition
and thus is appropriate for Lie reduction if and only if
it is $G$-equivalent to a subalgebra of the form
\[
\langle
\mathcal P^t+a_k\mathcal J^k,\,
\mathcal P^y+b_k\mathcal J^k,\,
\mathcal P^x(1)+c_k\mathcal J^k
\rangle,
\]
where $a:=(a_1,\dots,a_m)^{\mathsf T}$, $b:=(b_1,\dots,b_m)^{\mathsf T}$ and $c:=(c_1,\dots,c_m)^{\mathsf T}$
are arbitrary constant tuples with $a,b\in\mathop{\rm im}\mathsf F$.
An ansatz constructed with respect to such a subalgebra is
$\psi=\phi+xc+yb+ta$, where $\phi$ is the $m$-tuple of unknown constants,
which can be set to zero by the point symmetry transformation $\psi-\phi\to\psi$ of the system~\eqref{eq:mLaysMod}.
The corresponding reduced system is thus trivial;
this is just the compatibility condition of the ansatz with the system~\eqref{eq:mLaysMod},
which constraints the subalgebra parameters,
$\mathsf Fa+c\odot\mathsf Fb-b\odot\mathsf Fc+\beta c=0$.
This implies that $c\perp_{\mathsf W}^{}\langle\bar1\rangle=\ker\mathsf F$,
i.e., $c\in\mathop{\rm im}\mathsf F$ as well,
and $a=-\mathsf F^+(c\odot\mathsf Fb-b\odot\mathsf Fc+\beta c)$,
where the Moore--Penrose inverse~$\mathsf F^+$ of~$\mathsf F$ is given by~\eqref{eq:MNInverse}.
All the solutions of~\eqref{eq:mLaysMod} constructed in this way are very simple.
They are $\mathfrak s_{2.6}^{0bc0}$-invariant
and belong to the solution family~\eqref{eq:MultilayerLinInXY}, see Section~\ref{sec:ReductionS26} below.
This completes the study of codimension-three Lie reductions of the system~\eqref{eq:mLaysMod}.
\end{remark}

\section{Codimension-one Lie reductions}\label{sec:CodimOneLieReds}

To construct a Lie invariant ansatz for a system~$\mathcal L$ of differential equations
with $n$ independent variables $x=(x_1,\dots,x_n)$ and $m$ dependent variables $u=(u^1,\dots,u^m)$,
one selects a Lie subalgebra $\mathfrak s$ of the maximal Lie invariance algebra $\mathfrak g$ of~$\mathcal L$
satisfying the local transversality condition.
This condition means that the dimension of $\mathfrak s$
is equal to the dimension of the pushforward of $\mathfrak s$ by the canonical projection
from the space with the coordinates $(x,u)$ to the space with the coordinates~$x$.
If $(Q_1,\dots,Q_d)$ is a basis of the subalgebra~$\mathfrak s$,
then an $\mathfrak s$-invariant ansatz is a representation of the general solution
of the system~$\mathcal Q$ of quasilinear first-order partial differential equations $Q_1[u]=0$, \dots, $Q_d[u]=0$,
where~$Q[u]$ denotes the characteristic of a vector field~$Q$.
More specifically, given $n+m-d$ functionally independent particular solutions of~$\mathcal Q$,
which constitute a functional basis of $\mathfrak s$-invariants,
$m$ of them are assumed to be the new (invariant) dependent variables $v=(v^1,\dots,v^m)$
of the other $n-d$ solutions,
which are considered as the new (invariant) independent variables $z=(z_1,\dots,z_{n-d})$.
Substituting the ansatz into the system~$\mathcal L$, one obtains
a codimension-$d$ reduced system in terms of the invariant variables $(x,v)$.
The substitution of any solution of this reduced system into the ansatz
gives an $\mathfrak s$-invariant solution of~$\mathcal L$.
See \cite[Section~5.2.1]{blum2010A}, \cite[Chapter~3]{olve1993A} or \cite[Chapter~V]{ovsi1982A}
for a related theoretical background
as well as \cite[Section~2]{vinn2024a} and \cite[Section~B]{kova2023b}
for the descriptions of the optimized procedures of Lie reduction in the case of three independent variables
and of step-by-step reductions with involving hidden symmetries, respectively.

Using the subalgebras listed in Theorem~\ref{thm:1D-subalgebras},
we construct a complete list of $G$-inequivalent group-invariant codimension-one submodels
of the system~\eqref{eq:mLaysMod}.
The systematic study of these submodels, even within the framework of classical symmetry analysis, is too involved.
In particular, for each reduced system, the complexity of computing its Lie invariance algebra
is comparable to that presented in Section~\ref{sec:LieInvAlg} for the system~\eqref{eq:mLaysMod}.
Nevertheless, we were able to compute the Lie invariance algebras for
all of the above submodels by hand,
and in addition, we checked these results using several packages for symbolic computation of Lie symmetries
for low values of the number of layers~$m$.

In this section,~$v^1$, \dots, $v^m$ denote the new dependent variables
of the two new independent variables $(z_1,z_2)$, $v:=(v^1,\dots,v^m)^{\mathsf T}$
and the subscripts~1 and~2 of functions
denote derivatives with respect to~$z_1$ and~$z_2$, respectively

\subsection{Subalgebra family 1.1}\label{sec:ReductionS11}

\subsubsection{Reduced system}

An ansatz constructed with the one-dimensional subalgebra
$\mathfrak s_{1.1}^{ac}:=\langle\mathcal P^t+a\mathcal P^y+c_k\mathcal J^k\rangle$
of~$\mathfrak g$,
where $a$ and $c:=(c_1,\dots,c_m)^{\mathsf T}$ are arbitrary constants with $c\in\mathop{\rm im}\mathsf F$,
is $\psi^i=v^i+c_it+ax$ or,~in the vector notation,
\[
\psi=v+tc+ax\bar 1\quad\mbox{with}\quad z_1=x,\quad z_2=y-at,
\]
where $\bar 1$ is the all-ones $m$-column defined in Section~\ref{sec:VectorFormOfSyst}.
The corresponding reduced system takes the~form
\begin{gather}\label{eq:RedSystSubalgS11}
\begin{split}
&\{v^i,w^i\}
+f_{i,i-1}\big(c_{i-1}-c_i\big)
-f_{i,i+1}\big(c_i-c_{i+1}\big)+a\beta=0,
\\
&w^i:=v^i_{11}+v^i_{22}+f_{i,i-1}(v^{i-1}-v^i)-f_{i,i+1}(v^i-v^{i+1})+\beta z_2,
\end{split}
\end{gather}
where $\{v^i,w^i\}:=v^i_1w^i_2-v^i_2w^i_1$ is the Poisson bracket of~$v^i$ and~$w^i$, and $c_0,c_{m+1}:=0$.
The matrix form of the system~\eqref{eq:RedSystSubalgS11} is
\begin{gather}\label{eq:RedSystSubalgS11MatrixForm}
\begin{split}
&\{v,w\}+\mathsf Fc+a\beta\bar 1=0,\\
&w:=v_{11}+v_{22}+\mathsf Fv+\beta z_2\bar 1
\end{split}
\end{gather}
with the componentwise Poisson bracket $\{v,w\}$, $\{v,w\}:=(\{v^i,w^i\})^{\mathsf T}$
and the matrix $\mathsf F=(f_{ij})$ defined in Section~\ref{sec:VectorFormOfSyst}.
Using the package {\sf DESOLV}~\cite{carm2000a} for the computer algebra system {\sf Maple}
we computed the maximal Lie invariance algebra of the system~\eqref{eq:RedSystSubalgS11} for
low values of the number of layers~$m$, $m\in\{2,3,4\}$,
but such a computation is not possible for an arbitrary value of~$m$.
Nevertheless, we can make it by hand.

\begin{theorem}\label{thm:mLaysLieInvAlgRedSystem1.1}
The maximal Lie invariance algebra of the system~\eqref{eq:RedSystSubalgS11}
is the linear span \[\langle \p_{z_1},\p_{z_2},\p_{v^1},\dots,\p_{v^m}\rangle,\]
and it is entirely induced by
the maximal Lie invariance algebra~$\mathfrak g$ of the system~\eqref{eq:mLaysMod}.
\end{theorem}

\begin{proof}
We follow the consideration of Section~\ref{sec:LieInvAlg}.
More specifically, we formally replace the real independent variables~$(z_1,z_2)$
by the complex conjugate variables $z=z_1+{\rm i}z_2$ and $\bar z=z_1-{\rm i}z_2$, where ${\rm i}$ is the imaginary unit,
which reduces the system~\eqref{eq:RedSystSubalgS11} to the form
\begin{gather}\label{eq:RedSystSubalgS11ComplexForm}
\begin{split}
&v^i_zv^i_{z\bar z\bar z}-v^i_{\bar z}v^i_{zz\bar z}+G^i(z,\bar z,v,v_z,v_{\bar z})=0,\\
&G^i:=\frac14\big(v^i_z(\mathsf Fv_{\bar z})^i-v^i_{\bar z}(\mathsf Fv_z)^i\big)
+\frac{\beta\rm i}8(v^i_z+v^i_{\bar z}+a)+\frac{\rm i}8(\mathsf Fc)^i.
\end{split}
\end{gather}
The condition of infinitesimal invariance of the system~\eqref{eq:RedSystSubalgS11ComplexForm}
with respect to a vector field of the form
$Q=\xi(z,\bar z,v)\p_z+\bar\xi(z,\bar z,v)\p_{\bar z}+\eta^j(z,\bar z,v)\p_{v^j}$,
\begin{gather*}
\eta^{i,z}v^i_{z\bar z\bar z}+v^i_z\eta^{i,z\bar z\bar z}
-\eta^{i,\bar z}v^i_{zz\bar z}-v^i_{\bar z}\eta^{i,zz\bar z}
+Q_{(1)}G^i=0\quad\mbox{on solutions of~\eqref{eq:RedSystSubalgS11ComplexForm}},
\end{gather*}
is successively split with respect to parametric derivatives
when assuming $v^i_{z\bar z\bar z}$ as the leading derivatives and
taking into account the constraints derived for its components on the previous steps.
Here and in what follows,
$\bar\xi$ is the conjugate of~$\xi$,
$Q_{(1)}$ is the first prolongation of~$Q$,
$\eta^{i,\mu\nu\kappa}$~are the components of the third prolongation of~$Q$
for third-order derivatives,
${\rm D_\mu}$ denotes the total derivative operator with respect to $\mu\in\{z,\bar z\}$,
\begin{gather*}
Q_{(1)}=Q+\eta^{j,\mu}\p_{v^j_\mu}\quad\mbox{with}\quad
\eta^{i,\mu}={\rm D}_\mu (\eta^i-\xi v^i_z-\bar\xi v^i_{\bar z})
+\xi v^i_{z\mu}+\bar\xi v^i_{\bar z\mu},\\
\eta^{i,\mu\nu\kappa}={\rm D}_\mu{\rm D}_\nu{\rm D}_\kappa(\eta^i-\xi v^i_z-\bar\xi v^i_{\bar z})
+\xi v^i_{z\mu\nu\kappa}+\bar\xi v^i_{\bar z\mu\nu\kappa},\\
{\rm D_\mu}=\p_\mu+v^j_\mu\p_{v^j}+v^j_{\mu\nu}\p_{v^j_\nu}
+v^j_{\mu\nu\kappa}\p_{v^j_{\nu\kappa}}+v^j_{\mu\nu\kappa\lambda}\p_{v^j_{\nu\kappa\lambda}}+\cdots,
\end{gather*}
and the indices~$\kappa$, $\lambda$ and~$\nu$ run through the set $\{z,\bar z\}$.
Separately collecting the coefficients of the third-order jet variables~$v^i_{zzz}$, $v^i_{\bar z\bar z\bar z}$,
$v^j_{zz\bar z}$ with $j\ne i$ and $v^i_{zz\bar z}$ in the above $i$th condition,
we derive the equations
${\rm D}_z\bar\xi=0$, ${\rm D}_{\bar z}\xi=0$,
$\eta^i_{v^j}=0$, $j\ne i$, and
$v^i_{\bar z}\eta^i_z=v^i_z\eta^i_{\bar z}$,
which further splits to $\bar\xi_z=\xi_{\bar z}=\xi_{v^k}=0$ and $\eta^i_z=\eta^i_{\bar z}=0$.
Similarly continuing with the second-order jet variables~$v^i_{zz}$, $v^i_{\bar z\bar z}$ and $v^i_{z\bar z}$,
we obtain ${\rm D}_{\bar z}\eta^i_{v^i}=0$, ${\rm D}_z\eta^i_{v^i}=0$,
$v^i_{\bar z}\xi_{zz}=v^i_z\bar\xi_{\bar z\bar z}$,
which implies $\eta^i_{v^iv^i}=0$ and $\xi_{zz}=\bar\xi_{\bar z\bar z}=0$.
The next step is to collect the coefficients of the first-degree terms
with respect to the jet variables~$v^i_z$ and $v^i_{\bar z}$,
giving $\xi_z=\bar\xi_{\bar z}$ and $\eta^i_{v^i}=3\xi_z$.
Finally, analyzing the quadratic terms with respect to~$(v^k_z,v^k_{\bar z})$,
results in the equation $\xi_z=0$.
Therefore, all the components of~$Q$ are constants,
which is an obvious solution of the determining equations for~$Q$.

The normalizer ${\rm N}_{\mathfrak g}(\mathfrak s_{1.1}^{ac})$
of the subalgebra $\mathfrak s_{1.1}^{ac}$ in the algebra $\mathfrak g$
is spanned by the vector fields
$\mathcal P^t$, $\mathcal P^y$, $\mathcal P^x(1)$, $\mathcal J^1$,~\dots, $\mathcal J^m$.
Therefore, the algebra of induced symmetries of the reduced system~\eqref{eq:RedSystSubalgS11}
is the linear span $\langle \p_{z_1},\p_{z_2},\p_{v^1},\dots,\p_{v^m}\rangle$,
which coincides with the maximal Lie invariance algebra of this system.
\end{proof}

\subsubsection{Case of affine dependence}\label{sec:CaseOfAffDependence}

An important particular case of the submodels of the form~\eqref{eq:RedSystSubalgS11}
corresponds to the zero values of the parameter combinations
$f_{i,i-1}(c_{i-1}-c_i)-f_{i,i+1}(c_i-c_{i+1})+a\beta$, $i=1,\dots,m$.
The matrix form of this vanishing condition is $\mathsf Fc+a\beta\bar 1=0$.
Under the physical constraints of positivity on the essential components of the matrix~$\mathsf F$
and the parameter~$\beta$,
we have $\ker\mathsf F\perp_{\mathsf W}^{}\mathop{\rm im}\mathsf F$ and
can split the vanishing condition into $a=0$ and $\mathsf Fc=0$,
i.e., $c\in\langle\bar 1\rangle$, and thus $c=0$ since $c\perp_{\mathsf W}^{}\bar 1$ here.
This gives stationary solutions of the original system~\eqref{eq:mLaysMod}.
For each~$i$, the equation $\{v^i,w^i\}=0$ implies that $v^i$  and~$w^i$ are functionally dependent.

Each solution for which this dependence is affine satisfies
a linear system of partial differential equations of the form $w^i=b_{1i}v^i+b_{0i}$
with some real constants $b_{1i}$ and~$b_{0i}$.
Collecting the parameters $b_{1i}$ and $b_{0i}$
into the matrix $\mathsf B:=\mathop{\rm diag}(b_{11},\dots,b_{1m})$
and the column $b:=(b_{01},\dots,b_{0m})^{\mathsf T}$, respectively,
we rewrite the above system in the matrix form
\begin{gather}\label{eq:ReducitonS11LinearModons}
v_{11}+v_{22}+(\mathsf F-\mathsf B)v=b-\beta z_2\bar 1.
\end{gather}
Since the matrix $\mathsf F-\mathsf B$ is tridiagonal
with the same subdiagonal and the superdiagonal entries as those of~$\mathsf F$,
which are positive,
it is also similar to a symmetric tridiagonal matrix via the same matrix~$\mathsf D$
as the matrix~$\mathsf F$ is, see Section~\ref{sec:PropertiesOfMatrixF}.
Hence, it is diagonalizable and its eigenvalues~$\nu_i$ are real and pairwise distinct.
Thus, we can assume that $\nu_1<\dots<\nu_m$.
Consider the weighted inner product $(\cdot,\cdot)_{\mathsf W}$ with the weight matrix $\mathsf W:=\mathsf D^{-2}$
and denote by $\tilde e_{\nu_i}$ an eigenvector of $\mathsf F-\mathsf B$
with the eigenvalue $\nu_i$ and $\|\tilde e_{\nu_i}\|_{\mathsf W}^{}=1$.
Since $\mathsf F-\mathsf B$ is irreducible,
the first and the last components of each of the eigenvectors are nonzero.
According to the Perron--Frobenius theorem, 
one can choose $\tilde e_{\nu_m}$ with all components to be positive,
and all the other eigenvectors~$\tilde e_{\nu_i}$ (with $i<m$)
necessarily have components with opposite signs.
These eigenvectors constitute an orthonormal basis
in the $\mathsf W$-weighted inner product $(\cdot,\cdot)_{\mathsf W}^{}$.
The transition matrix~$\mathsf R$ to the obtained eigenbasis admits
the representation $\mathsf R=\mathsf D\tilde{\mathsf O}$,
where the matrix $\tilde{\mathsf O}$ is orthogonal with respect to the standard inner product
and the diagonal matrix~$\mathsf D$ is defined in Section~\ref{sec:PropertiesOfMatrixF},
and $\mathsf N:=\mathop{\rm diag}(\nu_1,\dots,\nu_m)=\mathsf R^{-1}(\mathsf F-\mathsf B)\mathsf R$.

Under the change of dependent variables $\tilde v=\mathsf R^{-1}(v-v^0)$,
where $v^0$ is a particular solution of the system~\eqref{eq:ReducitonS11LinearModons},
this system takes the form
\begin{gather}\label{eq:ReducitonS11LinearModonsDecoupled}
\tilde v_{11}+\tilde v_{22}+\mathsf N\tilde v=0,
\end{gather}
which is a decoupled system of homogeneous linear partial differential equations,
where the kind of the $i$th equation depends on the sign of~$\nu_i$,
the (two-dimensional) modified Helmholtz, the Laplace or the Helmholtz equation
if $\nu_i$ is negative, zero or positive, respectively.
The form of a particular solution $v^0$ of the system~\eqref{eq:ReducitonS11LinearModonsDecoupled}
depends on whether the matrix $\mathsf F-\mathsf B$ is degenerate or not.

If $\det(\mathsf F-\mathsf B)=0$,
the image of the operator $\mathsf F-\mathsf B$ is orthogonal to its kernel
with respect to the above weighted inner product, $\ker(\mathsf F-\mathsf B)=\langle \tilde e_0\rangle$ and
the restriction of the operator $\mathsf F-\mathsf B$ on its image $\langle \tilde e_{\nu_i},\nu_i\ne0\rangle$ is invertible.
Then we can take
\[
v^0
=(3b-\beta z_2\bar1,\tilde e_0)_{\mathsf W}^{}\frac{z_2^2}6\tilde e_0
+(\mathsf F-\mathsf B)^+(b-\beta z_2\bar1), 
\]
where $(\mathsf F-\mathsf B)^+$ is the Moore--Penrose inverse of $\mathsf F-\mathsf B$,
see Footnote~\ref{fnt:Moore-PenroseInverse}.
As a result, we construct a wide family of solutions of the original system~\eqref{eq:mLaysMod}
that are expressed in terms of the general solutions of the modified Helmholtz, the Laplace or the Helmholtz equations,
\begin{gather}\label{eq:GeneralSolutionHelmholtzDegenerate}
\solution
\psi=\mathsf R\tilde v(x,y)
-\beta(\bar1,\tilde e_0)_{\mathsf W}^{}\frac{y^3}6\tilde e_0
+\gamma y^2\tilde e_0
-\beta y(\mathsf F-\mathsf B)^+\bar1, 
\end{gather}
where
$\mathsf B:=\mathop{\rm diag}(b_{11},\dots,b_{1m})$,
$b_{11}$, \dots, $b_{1m}$ and~$\gamma$ are arbitrary constants,
$\tilde v_{xx}+\tilde v_{yy}+\mathsf N\tilde v=0$, $\mathsf N:=\mathop{\rm diag}(\nu_1,\dots,\nu_m)$,
$\nu_i$ are the eigenvalues of $\mathsf F-\mathsf B$,
the columns of~$\mathsf R$ are the corresponding eigenvectors  $\tilde e_{\nu_i}$ of the matrix~$\mathsf F-\mathsf B$
with $\|\tilde e_{\nu_i}\|_{\mathsf W}^{}=1$,
$(\cdot,\cdot)_{\mathsf W}$ is the weighted inner product with the weight matrix $\mathsf W:=\mathsf D^{-2}$,
and $(\mathsf F-\mathsf B)^+$ is the Moore--Penrose inverse of $\mathsf F-\mathsf B$.
We omitted the constant vector summand since it is inessential up to the $G$-equivalence
or, more specifically, the stream function tuple~$\psi$ is defined up to adding a constant tuple.

For the system~\eqref{eq:ReducitonS11LinearModons} with nondegenerate matrix $\mathsf F-\mathsf B$,
we can take a particular solution of much simpler form,
$v^0=(\mathsf F-\mathsf B)^{-1}\big(b-\beta z_2\bar1)$.
This leads to a wide family of solutions of the original system~\eqref{eq:mLaysMod}
that are expressed in terms of the general solutions of the modified Helmholtz or the Helmholtz equations,
\begin{gather}\label{eq:GeneralSolutionHelmholtzNonDegenerate}
\solution
\psi=\mathsf R\tilde v(x,y)-\beta y(\mathsf F-\mathsf B)^{-1}\bar1,
\end{gather}
where again
$\mathsf B:=\mathop{\rm diag}(b_{11},\dots,b_{1m})$,
$b_{11}$, \dots, $b_{1m}$ are arbitrary constants,
$\tilde v_{xx}+\tilde v_{yy}+\mathsf N\tilde v=0$, $\mathsf N:=\mathop{\rm diag}(\nu_1,\dots,\nu_m)$,
$\nu_i$ are the eigenvalues of $\mathsf F-\mathsf B$,
the columns of~$\mathsf R$ are the corresponding eigenvectors of the matrix~$\mathsf F-\mathsf B$,
and we omitted the constant vector summand since it is inessential up to the $G$-equivalence
or, more specifically, the stream function tuple~$\psi$ is defined up to adding a constant tuple.

While the found solutions~\eqref{eq:GeneralSolutionHelmholtzDegenerate} and~\eqref{eq:GeneralSolutionHelmholtzNonDegenerate}
of the multi-layer quasi-geostrophic system~\eqref{eq:mLaysMod} are stationary,
they can be extended to time-evolving solutions using (generalized) Galilean boosts in the $x$-direction,
which are point symmetries of~\eqref{eq:mLaysMod}.
More specifically, given a stationary solution~$\psi$ of~\eqref{eq:mLaysMod},
its pushforward $\mathscr P(\chi)_*\psi$ by the generalized Galilean transformation $\mathscr P(\chi)$
with an arbitrary sufficiently smooth function~$\chi$ of~$t$,
which has been defined in Section~\eqref{sec:mLaysPointSymGroup},
is also a solution of~\eqref{eq:mLaysMod}.
In order for $\mathscr P(\chi)_*\psi$ to represent a physically meaningful solution,
the parameter function $\chi$ should satisfy additional constraints,
in particular, on the sign of its derivative.

Many closed-form solutions of the modified Helmholtz, the Laplace and the Helmholtz equations
were presented in the literature.
In particular, wide families of such solutions were constructed
in~\cite{boye1976b} and~\cite[Sections~1.1--1.3]{mill1977A} using separation of variables based on symmetries.
See also~\cite[Chapter~7]{poly2002A}.

One can select various physically relevant solutions from the above families or use them as
building blocks for merging to more complicated solutions.
The signs of the eigenvalues of the matrix $\mathsf F-\mathsf B$ play an important role
in this selection.
A sufficient condition for this matrix to have at least one positive eigenvalue
is $f_{i,i-1}+f_{i,i+1}+b^{i1}<0$ for some $i\in\{1,\dots,m\}$,
which follows from application of the Poincar\'e separation theorem
to the symmetric matrix $\mathsf D^{-1}(\mathsf F-\mathsf B)\mathsf D$.
Sylvester's criterion implies
that all the eigenvalues of the matrix $\mathsf F-\mathsf B$ are positive
if all the leading principal minors of its symmetric counterpart $\mathsf D^{-1}(\mathsf F-\mathsf B)\mathsf D$ positive.
In view of the Gershgorin circle theorem,
a sufficient condition for the positivity of all the eigenvalues of the matrix $\mathsf F-\mathsf B$
is that $b_i<-\min\big(2|f_{i,i-1}+f_{i,i+1}|,|f_{i,i-1}+f_{i,i+1}+f_{i-1,i}+f_{i+1,i}|\big)$ for any~$i$.
Analogously, all the eigenvalues of the matrix $\mathsf F-\mathsf B$ are negative
if all the odd and even leading principal minors are negative and positive, respectively,
and a sufficient condition for this is that all $b_i$ are positive.\looseness=-1

\subsubsection{Solutions associated with entirely bounded solutions of Helmholtz equations}%
\label{sec:Reduction11BoundedSolOfHelmholtz}

Among the solutions constructed in Section~\ref{sec:CaseOfAffDependence},
we single out those corresponding to velocity fields
that are defined on the entire $(x,y)$-plane and bounded.
Denote $\alpha_i:=-\beta\big((\mathsf F-\mathsf B)^+\bar1,\tilde e_{\nu_i}\big)_{\mathsf W}^{}$.
If $\nu_i<0$, then the corresponding $\tilde e_{\nu_i}$-component
is $\tilde v^i+\alpha_iy$ for some solution~$\tilde v^i$
of the modified Helmholtz equation $\tilde v^i_{xx}+\tilde v^i_{yy}-|\nu_i|\tilde v^i=0$.
Since the nonzero solutions of this equation that are defined on the entire $(x,y)$-plane
have exponential growth at infinity, the only choice for~$\tilde v^i$ is $\tilde v^i\equiv0$.
If $\nu_i=0$, then $\alpha_i=0$ and the corresponding $\tilde e_{\nu_i}$-component
is $\tilde v^i-\frac16\beta(\bar1,\tilde e_0)_{\mathsf W}^{}y^3+\gamma y^2$
for some solution~$\tilde v^i$
of the Laplace equation $\tilde v^i_{xx}+\tilde v^i_{yy}=0$
and some constant~$\gamma$.
The first derivatives of this component are bounded if and only if $\gamma=0$, $(\bar1,\tilde e_0)_{\mathsf W}^{}=0$ and
both~$\tilde v^i_x$ and~$\tilde v^i_y$ are constant,
that is, the function~$\tilde v^i$ is affine with respect to $(x,y)$.%
\footnote{%
The derivative~$\tilde v^i_x$ itself is bounded and satisfies the Laplace equation;
hence it is constant in view of the Liouville theorem, $\tilde v^i_x=:\sigma_1=\const$.
The boundedness of the first $y$-derivative of the component implies that
the solution~$\tilde v^i_y$ of the Laplace equation grows not faster than a quadratic polynomial.
It follows from the extended Liouville theorem \cite[Theorem 5.9]{tesc2025A}
that actually $\tilde v^i_y$ is a quadratic polynomial.
Again recalling about the boundedness of the first $y$-derivative of the component,
we have that $\tilde v^i_y=\frac12\beta(\bar1,\tilde e_0)_{\mathsf W}^{}y^2-2\gamma y+\sigma_2$
for some constant~$\sigma_2$.
Then the function~$\tilde v^i$ satisfies the Laplace equation if and only if
$\gamma=0$ and $(\bar1,\tilde e_0)_{\mathsf W}^{}=0$, i.e.,
the derivative~$\tilde v^i_y$ is constant as well.
}
For $\nu_i>0$, the corresponding $\tilde e_{\nu_i}$-component
is $\tilde v^i+\alpha_iy$ for some solution~$\tilde v^i$
of the Helmholtz equation $\tilde v^i_{xx}+\tilde v^i_{yy}+\nu_i\tilde v^i=0$.
The first derivatives of this component are bounded if and only if
the first derivatives of~$\tilde v^i$ are bounded,
which is equivalent to the fact that the function~$\tilde v^i$ itself is bounded.
The latter means that we have the representation of~$\tilde v^i$ as a {\it generalized Herglotz wave function},
\[
\tilde v^i=\int_0^{2\pi}{\rm e}^{{\rm i}\sqrt{\nu_i}(x\cos\theta+y\sin\theta)}{\rm d}\mu_i(\theta)
\]
for some finite complex Radon measure~$\mu_i$ on the circle.%
\footnote{%
For any tempered solution~$u$ of the (homogeneous) Helmholtz equation $u_{xx}+u_{yy}+\kappa^2u=0$,
its Fourier transform $\hat u(\xi)$, $\xi:=(\xi_1,\xi_2)$,
satisfies the equation $(|\xi|^2-\kappa^2)\hat u(\xi)=0$,
and thus $\hat u(\xi)$ is a distribution supported on the circle $\{\xi\mid|\xi|=\kappa\}$.
If $u$ is a bounded solution, then $\hat u(\xi)$ is a finite measure supported on this circle.
The solution~$u$ of the Helmholtz equation is the Herglotz wave function with kernel $g\in L^2(S^1)$,
$u(x,y)=\int_0^{2\pi}g(\theta){\rm e}^{{\rm i}\kappa(x\cos\theta+y\sin\theta)}{\rm d}\theta$,
if and only if it has $L^2$-averaged $r^{-1/2}$-decay, i.e.,
there exists a (finite) limit of
\smash{$\kappa(4\pi r)^{-1}\int_{x^2+y^2\leqslant r^2}|u(x,y)|^2{\rm d}x{\rm d}y$}
as $r\to\infty$~\cite{hart1961a},
and it coincides with the square of the $L^2$-norm of~$g$.
Recall that every distributional solution of the homogeneous Helmholtz equation is a real analytic function.
}
More details on Herglotz wave functions can be found in~\cite{colt2019A,hart1961a} and references therein.
For the function~$\tilde v^i$ to be real-valued, the measure~$\mu_i$ should satisfy, roughly speaking, the condition
$\mu_i(\theta+\pi)=\bar\mu_i(\theta)$, where $\bar\mu_i$ denotes the complex conjugate of the measure~$\mu_i$.
The alternative way to construct real-valued solutions is
to take the real parts of generalized Herglotz wave functions,
where the corresponding measures do, in general, not satisfy the above condition.

Summing up, the stationary solution of the original system~\eqref{eq:mLaysMod}
with an affine componentwise dependence between the stream function~$\psi$ and the potential vorticity~$q$
that are defined in the entire $(x,y)$-plane and correspond to bounded velocity fields are exhausted
by the tuples
\begin{gather}\label{eq:Reduction12BoundedSolutions}
\solution
\psi=
\sum_{j\colon\nu_j>0}\mathop{\rm Re}\left(\int_0^{2\pi}{\rm e}^{{\rm i}\sqrt{\nu_j}(x\cos\theta+y\sin\theta)}
{\rm d}\mu_j(\theta)\right)\tilde e_{\nu_j}
+(\sigma_1x+\sigma_2y)\tilde e_0
-\beta y(\mathsf F-\mathsf B)^+\bar1.
\end{gather}
Here $\nu_j$ runs through the subset of $\{\nu_1,\dots,\nu_m\}$
constituted by the positive eigenvalues of $\mathsf F-\mathsf B$,
$\tilde e_{\nu_i}$ is an eigenvector of $\mathsf F-\mathsf B$
with the eigenvalue~$\nu_i$ and $\|\tilde e_{\nu_i}\|_{\mathsf W}=1$,
$\tilde e_0=0$ if the matrix $\mathsf F-\mathsf B$ is invertible,
$\sigma_1$ and~$\sigma_2$ are arbitrary constants,
and $\mu_i$ are arbitrary finite complex Radon measures on the circle.
The operator $(\mathsf F-\mathsf B)^+$ is the Moore--Penrose inverse of $\mathsf F-\mathsf B$,
see Footnote~\ref{fnt:Moore-PenroseInverse}.

Pushing the stationary solutions~\eqref{eq:Reduction12BoundedSolutions} forward
using, e.g., Galilean boosts $\mathscr P(\gamma t)$ with $\gamma<0$,
we obtain the travelling-wave solutions of the original system~\eqref{eq:mLaysMod},
\begin{gather*}
\solution
\psi=
\sum_{j\colon\nu_j>0}\mathop{\rm Re}\left(\int_0^{2\pi}
{\rm e}^{{\rm i}\sqrt{\nu_j}((x-\gamma t)\cos\theta+y\sin\theta)}
{\rm d}\mu_j(\theta)\right)\tilde e_{\nu_j}\\\qquad{}
+(\sigma_1(x-\gamma t)+\sigma_2y)\tilde e_0
-\beta y(\mathsf F-\mathsf B)^+\bar1
-\gamma y\bar 1,
\end{gather*}
which propagates in the negative direction of $x$.
In the geophysical context, this corresponds to westward wave propagation relative to the Earth's surface.

Suitably specifying the form of the (finite) measures $\mu_i$ in~\eqref{eq:Reduction12BoundedSolutions},
we construct various types of wave solutions.
To demonstrate the physical capability of the solution family~\eqref{eq:Reduction12BoundedSolutions},
we construct several distinct classes of solutions bounded on the entire plane:
baroclinic simple plane waves,
coherent baroclinic eddies,
coherent baroclinic hetons
and superpositions of solutions of the above kinds.
We also plot some simple solutions among the obtained ones
in the case of three layers with the numerical data
from the Section~\ref{sec:NumericalExample} and specific values of solution parameters,
neglecting the tuples $u_{\rm bg}=\beta(\mathsf F-\mathsf B)^+\bar1$ of constant background flows in the layers
from these solutions for a better illustration.
The multiplier $\psi_0:=15000$ m$^2/$s is chosen for solutions to be of the realistic scale.

\medskip\par\noindent
{\it Baroclinic plane waves.}
If the density of~$\mu_i$ is $\sum_{n=0}^\infty A_{in}\delta(\theta-\theta_n)$,
where~$A_{in}$ are arbitrary complex constants with $\sum_{n=0}^\infty|A_{in}|<\infty$, $0\leqslant\theta_n<2\pi$
and $\delta$ denotes the Dirac delta-function,
then the function~$\tilde v^i$ and its Galilean boost are the simple and the travelling two-dimensional plane harmonics,
\[
\tilde v^i=\mathop{\rm Re}\sum_{n=0}^{\infty}A_{in}{\rm e}^{{\rm i}\sqrt{\nu_i}(x\cos\theta_n+y\sin\theta_n)},\quad
\mathscr P^x(\gamma t)_*\tilde v^i=\mathop{\rm Re}\sum_{n=0}^{\infty}A_{in}{\rm e}^{{\rm i}\sqrt{\nu_i}((x-\gamma t)\cos\theta_n+y\sin\theta_n)},
\]
respectively,
and these solution describe stationary and travelling geophysical Rossby waves.
Note that each fixed harmonic can be derived via the multiplicative separation of variables
in the Cartesian coordinate system.

To plot a representative of this solution family,
we select $\mathsf B=\mathop{\rm diag}(-7,4,-1)\cdot10^{-10}$.
The eigenvalues and corresponding eigenvectors of the matrix $\mathsf F-\mathsf B$ are given by
\begin{gather*}
\nu_1\approx-12.9\cdot10^{-10},\quad
\nu_2\approx-0.76\cdot10^{-10},\quad
\nu_3\approx 1.67\cdot10^{-10},
\\
\tilde e_{\nu_1}\approx(0.74,-1.02,0.14)^{\mathsf T},\quad
\tilde e_{\nu_2}\approx(-0.82,-0.06,1.02)^{\mathsf T},\quad
\tilde e_{\nu_3}\approx(0.96,0.35,0.25)^{\mathsf T}.
\end{gather*}
Note that in this case the matrix $\mathsf F-\mathsf B$ is invertible, so $\tilde e_0=0$, and
the tuple of constant background flows in the layers is
$u_{\rm bg}:=\beta(\mathsf F-\mathsf B)^{-1}\bar 1\approx(0.28,0.06,-0.08)^{\mathsf T}$.
Evidently, only the third mode is relevant for constructing globally bounded solutions.
We choose the measure $\mu_3$ to be proportional to the sum
of two Dirac delta functions centered at a propagation angle $\pi/4$,
${\rm d}\mu_3(\theta)=\frac12\psi_0\big(\delta(\theta-\pi/4)+\delta(\theta-5\pi/4)\big){\rm d}\theta$.
As a result, we recover the plane wave solution
\begin{gather}\label{eq:Red11PlaneWaves}
\psi=\psi_0\cos\big(\sqrt{\nu_3/2}(x+y)\big)\tilde e_{\nu_3}-yu_{\rm bg}\quad\mbox{with $\psi_0:=15000$ m$^2/$s}.
\end{gather}
The plot of~\eqref{eq:Red11PlaneWaves} is given, up to the summand that gives constant flows in layers,
in Figure~\ref{fig:Red11PlaneWaves}.
Physically, this solution corresponds to a stationary Rossby wave pattern
where the westward phase propagation of the wave
is exactly balanced by the eastward advection of the background current.
\begin{figure}[!ht]
\centering
\includegraphics[width=\linewidth]{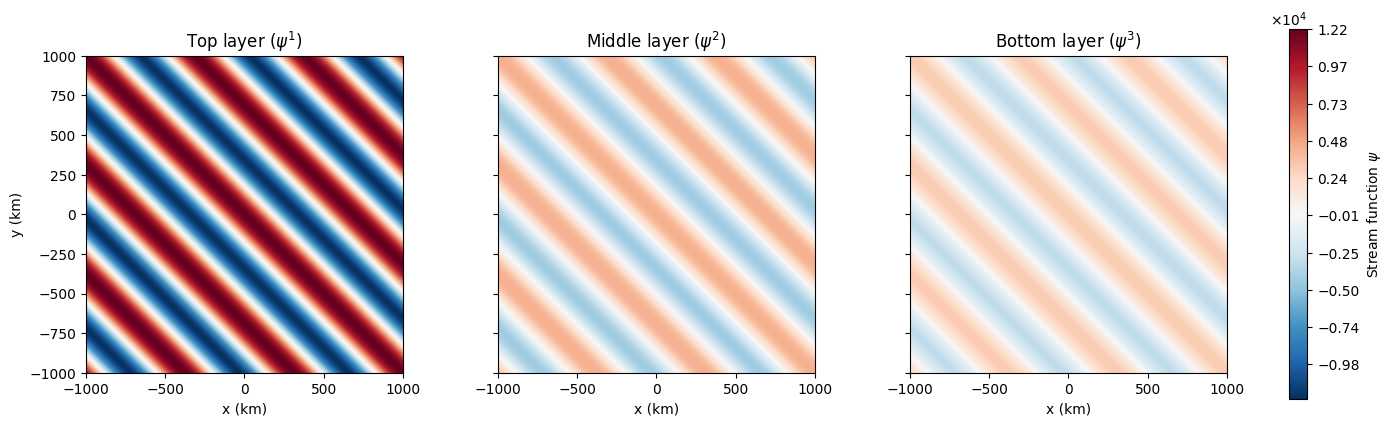}
\caption{Baroclinic simple plane waves~\eqref{eq:Red11PlaneWaves}.
}
\label{fig:Red11PlaneWaves}
\end{figure}

\medskip\par\noindent
{\it Coherent baroclinic eddies and hetons.}
The density $\sum_{n=0}^\infty A_{in}{\rm e}^{{\rm i}n\theta}$ of~$\mu_i$,
where $(A_{in},\,n\in\mathbb N_0)$ is an arbitrary tempered sequence of complex constants,
corresponds to a linear superposition of two-dimensional radial harmonics,
\begin{gather*}
\tilde v^i=2\pi\mathop{\rm Re}\sum_{n=0}^\infty{\rm i}^nA_{in}J_n(\sqrt{\nu_i}r){\rm e}^{{\rm i}n\vartheta},\quad
r:=\sqrt{x^2+y^2},\quad \vartheta:=\arccos\frac xr,
\end{gather*}
where $J_n$ denotes the $n$th order Bessel function of the first kind.
Note that each fixed harmonic can be derived via the multiplicative separation of variables
in the polar coordinate system.

For instance, using the same matrix $\mathsf B=\mathop{\rm diag}(-7,4,-1)\cdot10^{-10}$ as above
and selecting measure $\mu_3$ to be given by a constant density on the spectral circle,
${\rm d}\mu_3(\theta)=\psi_0{\rm d}\theta/2\pi$,
we recover a rotationally symmetric isolated vortex
represented in terms of the zeroth-order Bessel function of the first kind,
\begin{gather}\label{eq:Red11Eddy}
\psi=\psi_0J_0(\sqrt{\nu_3}r)\tilde e_{\nu_3}-yu_{\rm bg}\quad\mbox{with $\psi_0:=15000$ m$^2/$s}.
\end{gather}
The plot of this solution (up to the tuple $u_{\rm bg}\approx(0.28,0.06,-0.08)^{\mathsf T}$
of constant background flows in layers) is given in Figure~\ref{fig:Red11Eddy}.
Physically, it represents a standing baroclinic eddy embedded in a uniform zonal current.
\begin{figure}[!ht]
\centering
\includegraphics[width=\linewidth]{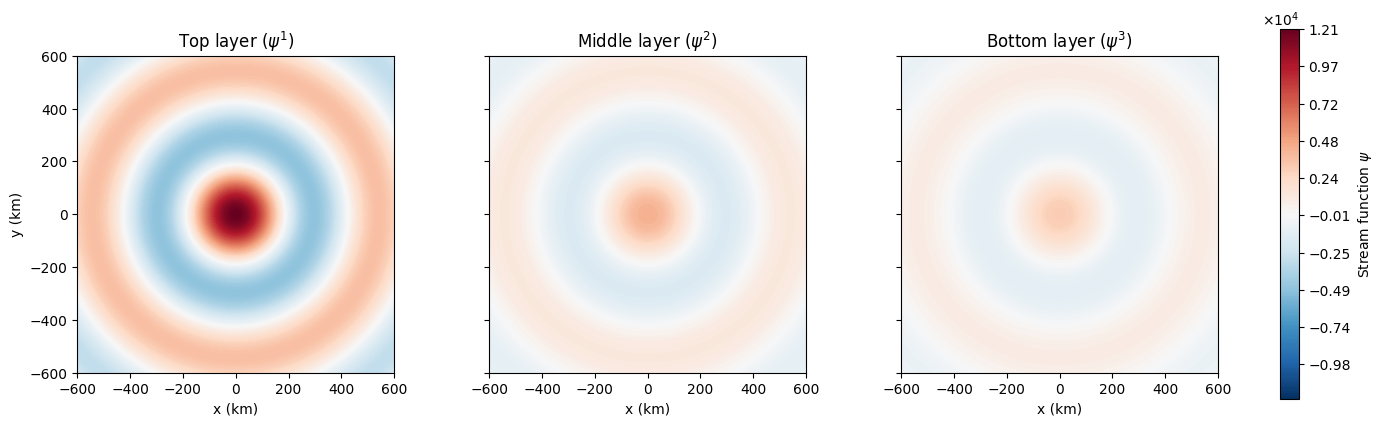}
\caption{Coherent baroclinic eddy~\eqref{eq:Red11Eddy}.}
\label{fig:Red11Eddy}
\end{figure}

To obtain a heton solution, we take the matrix $\mathsf B=\mathop{\rm diag}(-10,-8,-5)\cdot10^{-10}$
The appropriate eigenvalue and eigenvector of $\mathsf F-\mathsf B$ for constructing a heton
are then $\nu_2\approx3.13\cdot10^{-10}$ and $\tilde e_{\nu_2}\approx(0.79,0.14,-1.12)^{\mathsf T}$
since the first and the second components of~$\tilde e_{\nu_2}$ have the same order and opposite signs.
The measure $\mu_2$ is chosen by a constant density on the spectral circle as above,
${\rm d}\mu_2(\theta)=\psi_0{\rm d}\theta/2\pi$, and $\mu_3=0$ ($\nu_3>0$ as well).
While the solution arising from this setting is formally still of the form~\eqref{eq:Red11Eddy}
up to the permutation of $\nu_2$ and~$\nu_3$,
\begin{gather}\label{eq:Red11Heton}
\psi=\psi_0J_0(\sqrt{\nu_2}r)\tilde e_{\nu_2}-yu_{\rm bg}\quad\mbox{with $\psi_0:=15000$ m$^2/$s},
\end{gather}
it has a different physical interpretation, see, e.g.,~\cite{rein2018a}.
We give its plot in Figure~\ref{fig:Red11Heton},
neglecting the tuple $u_{\rm bg}\approx(0.86,1.75,3.92)^{\mathsf T}\cdot 10^{-2}$
of constant background flows in layers.
\begin{figure}[!ht]
\centering
\includegraphics[width=\linewidth]{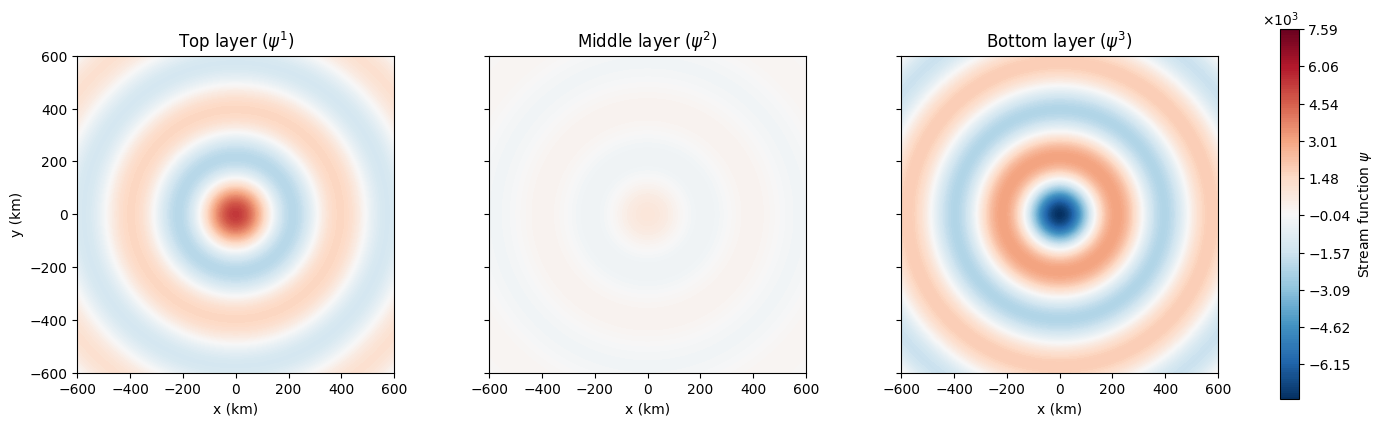}
\caption{Coherent baroclinic heton~\eqref{eq:Red11Heton}.}
\label{fig:Red11Heton}
\end{figure}

\medskip\par\noindent
{\it Superpositions of Rossby waves with different modes.}
A more complicated dynamics can be obtained via linear superposition of solutions.
Consider now the matrix $\mathsf B:=\mathop{\rm diag}(-8,3,-2)\cdot10^{-10}$.
The relevant eigenvalues and eigenvectors of the matrix $\mathsf F-\mathsf B$ are
$\nu_2\approx2.40\cdot10^{-10}$, $\nu_3\approx2.67\cdot10^{-10}$
and $\tilde e_{\nu_2}\approx(-0.83,-0.06,1.02)^{\mathsf T}$, $\tilde e_{\nu_3}\approx(0.96,0.35,0.25)^{\mathsf T}$.
The plane wave solution of the form~\eqref{eq:Reduction12BoundedSolutions} with
${\rm d}\mu_2(\theta)=\frac12\psi_0\big(\delta(\theta-\pi/3)+\delta(\theta-2\pi/3)\big){\rm d}\theta$ and
${\rm d}\mu_3(\theta)=\frac12\psi_0\big(\delta(\theta-\pi/4)+\delta(\theta-5\pi/4)\big){\rm d}\theta$,
where $\psi_0:=15000$ m$^2/$s, is
\begin{gather}\label{eq:Red11SuperpositionRossby}
\psi(x,y)=\psi_0\cos\Big(\frac{\sqrt{\nu_2}}2(\sqrt3x+y)\Big)\tilde e_{\nu_2}+
\psi_0\cos\big(\sqrt{\nu_3/2}(x+y)\big)\tilde e_{\nu_3}
-yu_{\rm bg},
\end{gather}
The plot of this solution (up to neglecting the tuple $u_{\rm bg}\approx(-0.21,0.01,0.43)^{\mathsf T}\cdot 10^{-2}$
of constant background flows in layers) is given in Figure~\ref{fig:Red11SuperpositionRossby}.
This picture represents the formation of cyclone and anticyclone in the bottom layer,
while the top and the middle layers have a Rossby wave pattern.
\begin{figure}[!ht]
\centering
\includegraphics[width=\linewidth]{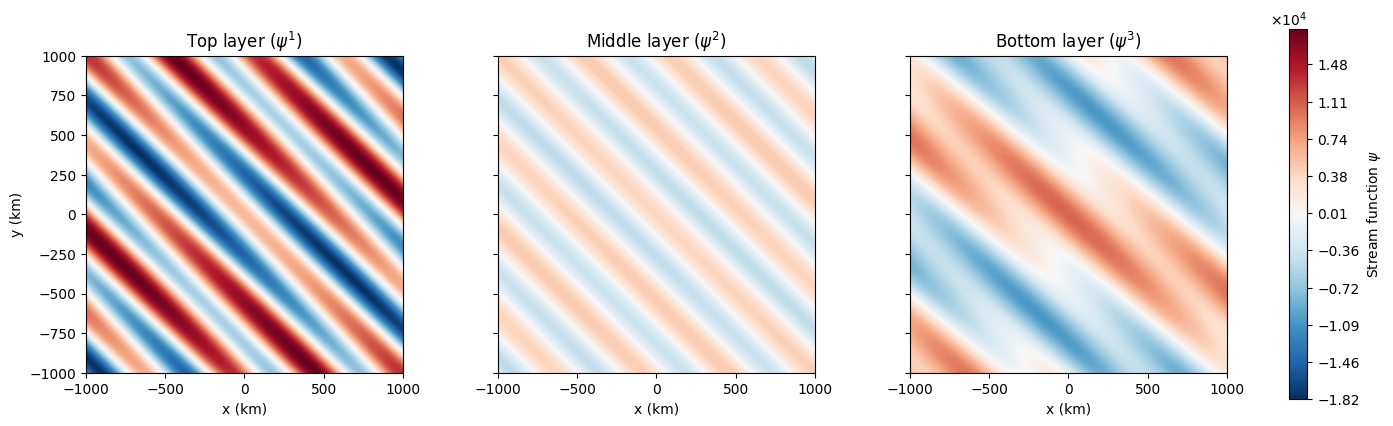}
\caption{Superposition of two baroclinic Rossby waves~\eqref{eq:Red11SuperpositionRossby}.}
\label{fig:Red11SuperpositionRossby}
\end{figure}

\medskip\par\noindent
{\it General linear superpositions.}
One can construct exact solutions of the system~\eqref{eq:mLaysMod}
with quite irregular behaviour
using merely the linear superposition of a particular solution of the system~\eqref{eq:ReducitonS11LinearModons}
with the above regular solutions patterns of the homogeneous counterpart of this system
with the same parameter matrix~$\mathsf B$.
The involved solution patterns, which are two-dimensional plane and radial harmonics,
can in addition be shifted and scaled before superposing.
An example of such an irregular solution is presented in Figure~\ref{fig:Red11GenLinSuperposition}.
It is just the linear superposition of three plane and seven radial harmonics of low order.

\begin{figure}[!ht]
\centering
\includegraphics[width=\linewidth]{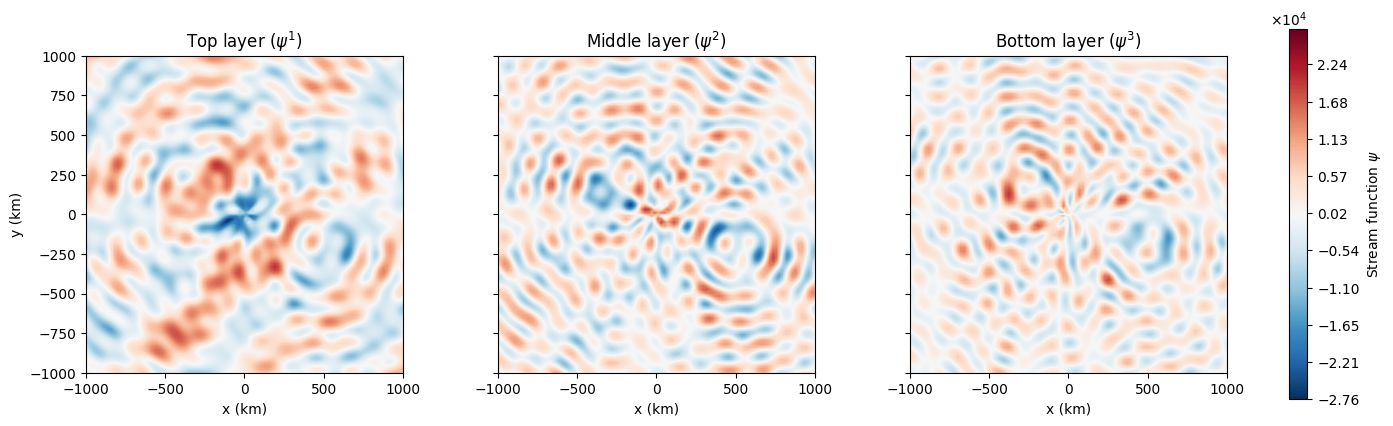}
\caption{An example of linear superposition of various solution patterns.}
\label{fig:Red11GenLinSuperposition}
\end{figure}

\subsubsection{Merging solutions and modons}\label{sec:Reduction11Modon}

Another way to construct a wide class of physically relevant (usually weak) solutions
of the multi-layer quasi-geostrophic problem~\eqref{eq:mLaysMod}
using the Lie reduction with respect to the subalgebra~$\mathfrak s_{1.1}^{00}$
is to partition the solution domain into several subdomains
and then consider different solutions of the form~\eqref{eq:GeneralSolutionHelmholtzDegenerate} or~\eqref{eq:GeneralSolutionHelmholtzNonDegenerate}
for the same or even different systems of the form~\eqref{eq:ReducitonS11LinearModons} on the subdomains,
merging these solutions by imposing interface boundary conditions between subdomains.

In this section, plotting merged solutions with the numerical data
from the Section~\ref{sec:NumericalExample} and specific values of solution parameters,
we do not neglect the tuples $u_{\rm bg}=\beta(\mathsf F-\mathsf B)^+\bar1$
of constant background flows in the layer subdomains
since otherwise the plots are discontinuous.

\medskip\par\noindent
{\it Dipolar vortices (modons).}
Consider arbitrary diagonal matrices~$\tilde{\mathsf B}$ and~$\hat{\mathsf B}$ such that
all eigenvalues~$\nu_i$ of $\mathsf F-\tilde{\mathsf B}$ are positive, ${0<\nu_1<\dots<\nu_m}$, and
all eigenvalues~$\varsigma_i$ of $\mathsf F-\hat{\mathsf B}$ are negative, $\varsigma_1<\dots<\varsigma_m<0$.
Denote by~\smash{$\tilde e_{\nu_j}$} and~\smash{$\hat e_{\varsigma_j}$} $\mathsf W$-normalized eigenvectors
of the matrices $\mathsf F-\tilde{\mathsf B}$ and $\mathsf F-\hat{\mathsf B}$
associated with the eigenvalues~$\nu_i$ and~$\varsigma_i$, respectively.
Both collections of eigenvectors
\smash{$(\tilde e_{\nu_1},\dots,\tilde e_{\nu_m})$} and \smash{$(\hat e_{\varsigma_1},\dots,\hat e_{\varsigma_m})$}
are orthonormal bases in $\mathbb R^m$  with respect to the $\mathsf W$-weighted inner product $(\cdot,\cdot)_{\mathsf W}^{}$.
We partition the $(x,y)$-plane into a disk of radius $r_0$ with center at zero and its exterior.
As the values of the stream function~$\psi$ inside and outside the disk,
we take the solutions $\tilde\psi$ and~$\hat\psi$ defined
using the first radial harmonics for the Helmholtz and modified Helmholtz equations
that are bounded on the corresponding domains, respectively,
\begin{gather*}
\tilde\psi=\sin\theta\,
\big(\tilde\alpha_jJ_1(\sqrt{\nu_j}r)\tilde e_{\nu_j}-r\tilde u_{\rm bg}\big)
\quad\mbox{with}\quad \tilde u_{\rm bg}:=\beta(\mathsf F-\tilde{\mathsf B})^{-1}\bar 1,
\quad r\leqslant r_0,
\\
\hat\psi=\sin\theta\,
\big(\hat\alpha_jK_1(\sqrt{-\varsigma_j}r)\hat e_{\varsigma_j}-r\hat u_{\rm bg}\big)
\quad\mbox{with}\quad \hat u_{\rm bg}:=\beta(\mathsf F-\hat{\mathsf B})^{-1}\bar 1,
\quad r>r_0.
\end{gather*}
Here $J_1$ and~$K_1$ denote the first-order Bessel functions of the first and the second kinds, respectively.
Note that the Bessel function~$J_1$ is slowly decaying as $r\to\infty$
and the Bessel function~$K_1$ has a singularity at $r=0$.
But the partition of the domain has accounted for these problems.
Moreover, $\tilde\psi$ and $\hat\psi$ are solutions of two different Helmholtz equations
with different eigenvalues.

The parameters $\tilde{\mathsf B}$, $\hat{\mathsf B}$, $r_0$, $\tilde\alpha_j$ and~$\hat\alpha_j$
are determined by enforcing the matching conditions $\tilde\psi=\hat\psi=0$
and $\tilde\psi_r=\hat\psi_r$ at the interface $r=r_0$,
which leads to a system for these parameters,
\begin{gather}\label{eq:ModonCondition0}
\begin{split}&
\tilde\alpha_jJ_1(\sqrt{\nu_j}r_0)\tilde e_{\nu_j}=r_0\tilde u_{\rm bg},
\quad
\hat\alpha_jK_1(\sqrt{-\varsigma_j}r_0)\hat e_{\varsigma_j}=r_0\hat u_{\rm bg}.
\\&
\tilde\alpha_j\sqrt{\nu_j}J_1'(\sqrt{\nu_j}r_0)\tilde e_{\nu_j}
-\hat\alpha_j\sqrt{-\varsigma_j}K_1'(\sqrt{-\varsigma_j}r_0)\hat e_{\varsigma_j}
=\tilde u_{\rm bg}-\hat u_{\rm bg}.
\end{split}
\end{gather}
We need the condition of vanishing both the functions $\tilde\psi$ and~$\hat\psi$ on the entire circle $r=r_0$
to construct a solution of the system~\eqref{eq:mLaysMod} at least in the weak sense
since
\begin{gather*}
\tilde\psi_{rr}=-r^{-1}\tilde\psi_r-r^{-2}\tilde\psi_{\theta\theta}
-(\mathsf F-\tilde{\mathsf B})\tilde\psi-\beta r\sin\theta\,\bar1,
\\
\hat\psi_{rr}=-r^{-1}\hat\psi_r-r^{-2}\hat\psi_{\theta\theta}
-(\mathsf F-\hat{\mathsf B})\hat\psi-\beta r\sin\theta\,\bar1,
\end{gather*}
$\hat{\mathsf B}\ne\tilde{\mathsf B}$,
and thus otherwise $\tilde\psi_{rr}\ne\hat\psi_{rr}$
on the entire circle $r=r_0$ except the two points with $\theta\in\{0,\pi\}$,
which leads to the discontinuity of $\psi_{rr}$ and the vorticity~$q$ therein.
In general, the third derivative~$\psi_{rrr}$ then has a finite discontinuity on the circle $r=r_0$,
which fits well in the weak setting, cf.~\cite{lari1976a}.
Therefore, the obtained modons are solutions of the system~\eqref{eq:mLaysMod}
in the weak sense.
The system~\eqref{eq:ModonCondition0} implies that
\begin{gather}\nonumber
\tilde\alpha_j
=\frac{r_0\beta}{J_1(\sqrt{\nu_j}r_0)}
\big((\mathsf F-\tilde{\mathsf B})^{-1}\bar 1,\tilde e_{\nu_j}\big)_{\mathsf W}^{}
=\frac{r_0\beta}{J_1(\sqrt{\nu_j}r_0)}
\big(\bar 1,(\mathsf F-\tilde{\mathsf B})^{-1}\tilde e_{\nu_j}\big)_{\mathsf W}^{}
=\frac{r_0\beta(\bar 1,\tilde e_{\nu_j})_{\mathsf W}^{}}{\nu_jJ_1(\sqrt{\nu_j}r_0)},
\\\nonumber
\hat\alpha_j
=\frac{r_0\beta}{K_1(\sqrt{-\varsigma_j}r_0)}
\big((\mathsf F-\hat{\mathsf B})^{-1}\bar 1,\hat e_{\varsigma_j}\big)_{\mathsf W}^{}
=\frac{r_0\beta}{K_1(\sqrt{-\varsigma_j}r_0)}
\big(\bar 1,(\mathsf F-\hat{\mathsf B})^{-1}\hat e_{\varsigma_j}\big)_{\mathsf W}^{}
=\frac{r_0\beta(\bar 1,\hat e_{\varsigma_j})_{\mathsf W}^{}}{\varsigma_j K_1(\sqrt{-\varsigma_j}r_0)},
\\\label{eq:ModonCondition1c}
\frac{r_0\beta}{\sqrt{\nu_j}}\frac{J_1'(\sqrt{\nu_j}r_0)}{J_1(\sqrt{\nu_j}r_0)}
(\bar 1,\tilde e_{\nu_j})_{\mathsf W}^{}\tilde e_{\nu_j}
+\frac{r_0\beta}{\sqrt{-\varsigma_j}}\frac{K_1'(\sqrt{-\varsigma_j}r_0)}{K_1(\sqrt{-\varsigma_j}r_0)}
(\bar 1,\hat e_{\varsigma_j})_{\mathsf W}^{}\hat e_{\varsigma_j}
=\tilde u_{\rm bg}-\hat u_{\rm bg}.
\end{gather}
The equality~\eqref{eq:ModonCondition1c} can be interpreted as
a system of $m$ equations with respect to the parameters
$\tilde b_1$, \dots, $\tilde b_m$, $\hat b_1$, \dots, $\hat b_m$ and~$r_0$.
Thus, in general, $m+1$ of these parameters can be arbitrarily chosen.
For practical computations,
one should consider the equality~\eqref{eq:ModonCondition1c} jointly with
the determining equations
$(\mathsf F-\tilde{\mathsf B})\tilde e_{\nu_j}=\nu_j\tilde e_{\nu_j}$,
$(\mathsf F-\hat{\mathsf B})\hat e_{\varsigma_j}=\varsigma_j\hat e_{\varsigma_j}$
for the corresponding eigenvalues and eigenvectors,
which are then also considered as additional parameters to be chosen or found.
The solutions constructed in the above way are extensions, to the case of an arbitrary number of layers,
of the Larichev--Reznik modons~\cite{lari1976a} for the single-layer model ($m=1$),
see also~\cite[Section~4.1]{crow2024a} and ~\cite{flie1980a,kizn2000a}
as well as for a precursory construction in~\cite{ster1975a}.
Despite various differences and possible involvement of both baroclinic and barotropic modes,
the shape of these solutions has common features.
The structure is vertically locked and contains matched pairs of counter-rotating vortices (dipoles).
It may exhibit baroclinic shear,
modeling coherent structures like ``mushroom'' vortices often observed in satellite imagery.

In some special cases, the computation of particular solutions of the system~\eqref{eq:ModonCondition1c}
can be simplified.
If $\hat{\mathsf B}=\tilde{\mathsf B}+\varrho\mathsf E$, then $\varsigma_i=\nu_i-\varrho$ and
$\hat e_{\varsigma_i}=\tilde e_{\nu_i}$, i.e.,
the operators $\mathsf F-\tilde{\mathsf B}$ and $\mathsf F-\hat{\mathsf B}$ share the eigenbasis,
\begin{gather*}
\tilde\alpha_j=\frac{r_0\beta(\bar 1,\tilde e_{\nu_j})_{\mathsf W}^{}}{\nu_jJ_1(\sqrt{\nu_j}r_0)},
\quad
\hat\alpha_j=\frac{-r_0\beta(\bar 1,\tilde e_{\nu_j})_{\mathsf W}^{}}{(\varrho-\nu_j)K_1(\sqrt{\varrho-\nu_j}r_0)},
\end{gather*}
and the condition~\eqref{eq:ModonCondition1c} reduces to the single equation
\begin{gather*}
\noprint{
\frac{r_0\beta}{\sqrt{\nu_j}}\frac{J_1'(\sqrt{\nu_j}r_0)}{J_1(\sqrt{\nu_j}r_0)}
(\bar 1,\tilde e_{\nu_j})_{\mathsf W}^{}\tilde e_{\nu_j}
+\frac{r_0\beta}{\sqrt{\varrho-\nu_j}}\frac{K_1'(\sqrt{\varrho-\nu_j}r_0)}{K_1(\sqrt{\varrho-\nu_j}r_0)}
(\bar 1,\tilde e_{\nu_j})_{\mathsf W}^{}\tilde e_{\nu_j}
=\tilde u_{\rm bg}-\hat u_{\rm bg}
\\\quad{}
=\beta(\nu_j^{-1}+(\varrho-\nu_j)^{-1})(\bar 1,\tilde e_{\nu_j})_{\mathsf W}^{}\tilde e_{\nu_j}.
\\[1ex]
\frac{r_0}{\sqrt{\nu_j}}\frac{J_1'(\sqrt{\nu_j}r_0)}{J_1(\sqrt{\nu_j}r_0)}
+\frac{r_0}{\sqrt{\varrho-\nu_j}}\frac{K_1'(\sqrt{\varrho-\nu_j}r_0)}{K_1(\sqrt{\varrho-\nu_j}r_0)}
=\nu_j^{-1}+(\varrho-\nu_j)^{-1}.
\\[1ex]
}
\frac{r_0}{\sqrt{\nu}}\frac{J_0(\sqrt{\nu}r_0)}{J_1(\sqrt{\nu}r_0)}
-\frac{r_0}{\sqrt{\varrho-\nu}}\frac{K_0(\sqrt{\varrho-\nu}r_0)}{K_1(\sqrt{\varrho-\nu}r_0)}
=\frac2{\nu}+\frac2{\varrho-\nu},
\end{gather*}
and $\nu_1$, \dots, $\nu_m$ are $m$ distinct positive roots of this equation
for chosen $r_0$ and~$\varrho$.
Given these roots, we find the matrix~$\tilde{\mathsf B}$ from the equality of polynomials
\smash{$\det(\mathsf F-\tilde{\mathsf B}-\nu\mathsf E)=\prod_{i=1}^m(\nu-\nu_i)$},
which gives a system of $m$ equations for the $m$ diagonal entries of~$\tilde{\mathsf B}$.
Then we find the eigenvectors~$\tilde e_{\nu_1}$,~\dots, $\tilde e_{\nu_m}$
of the matrix $\mathsf F-\tilde{\mathsf B}$.
On Figure~\ref{LRbcModons}, we plot these specific modons
for the numerical data from the Section~\ref{sec:NumericalExample}
and the following values of parameters:
\begin{description}\itemsep=0ex
\item[\rm (a) chosen:] 
$r_0=170$ km,\ \ $\varrho= 4.5\cdot10^{-9}$,
\item[\rm\phantom{(a) }found:]
$\mathsf B\approx\mathop{\rm diag}(-1.71,-4.26,-2.23)\cdot 10^{-9}$,\ \
$\nu=(1.74,3.59,9.4)\cdot 10^{-10}$,\\[.5ex]
$\tilde e_{\nu_1}\approx(-0.98, 0.12, -0.016)^{\mathsf T}$, \
$\tilde e_{\nu_2}\approx-(0.18,0.64,0.062)^{\mathsf T}$, \
$\tilde e_{\nu_3}\approx(-0.05, -0.069, 0.54)^{\mathsf T}$,\\[.5ex]
$\tilde\alpha\approx(0.94,1.87,1.2)\cdot10^4$,\ \
$\hat\alpha\approx(4.39,0.104,-1.8)\cdot 10^7$;
\item[\rm (b) chosen:] 
$r_0=150$ km,\ \ $\varrho= 8\cdot10^{-9}$,
\item[\rm\phantom{(b) }found:]
$\mathsf B\approx\mathop{\rm diag}(-3.34,-1.82,-5.25)\cdot 10^{-9}$,\ \
$\nu=(1.01, 2.7, 5.09)\cdot 10^{-9}$,\\[.5ex]
$\tilde e_{\nu_1}\approx(-0.34,0.61,-0.025)^{\mathsf T}$, \
$\tilde e_{\nu_2}\approx(-0.94,-0.22, 0.016)^{\mathsf T}$, \
$\tilde e_{\nu_3}\approx(0.011, 0.035, 0.55)^{\mathsf T}$,\\[.5ex]
$\tilde\alpha\approx(-8.12,-6.22,-7.36)\cdot10^3$,\ \
$\hat\alpha\approx(-26.46,8.97,-1.12)\cdot 10^7$.
\end{description}

\begin{figure}[!ht]
\begin{subfigure}[b]{\textwidth}
\centering
\includegraphics[width=\linewidth]{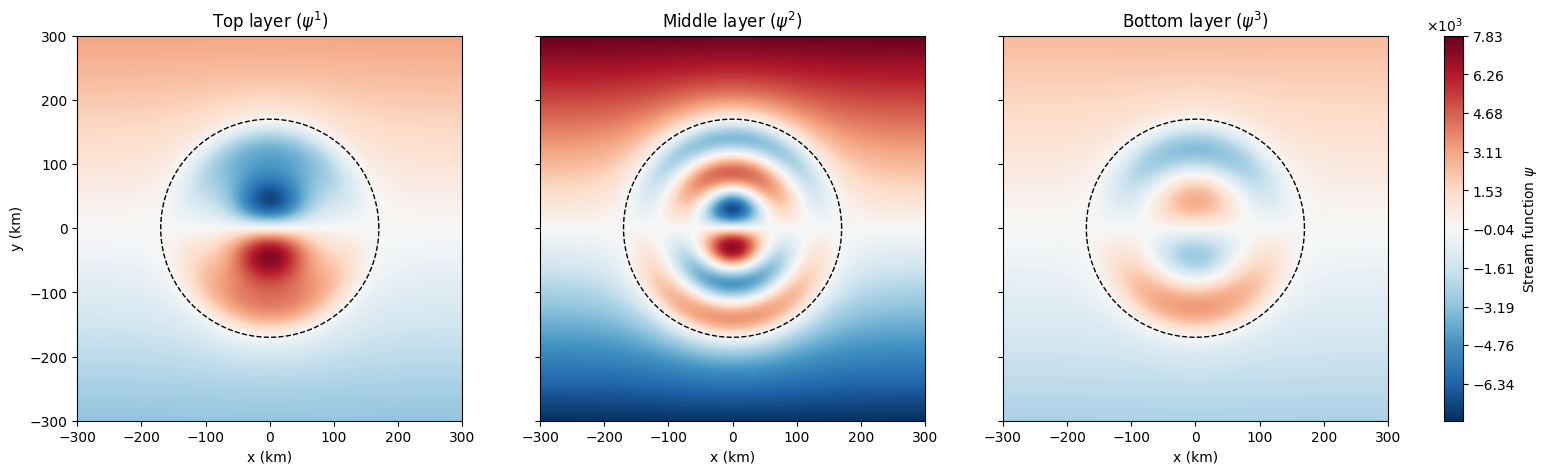}
\caption{}
\end{subfigure}
\begin{subfigure}[b]{\textwidth}
\centering
\includegraphics[width=\linewidth]{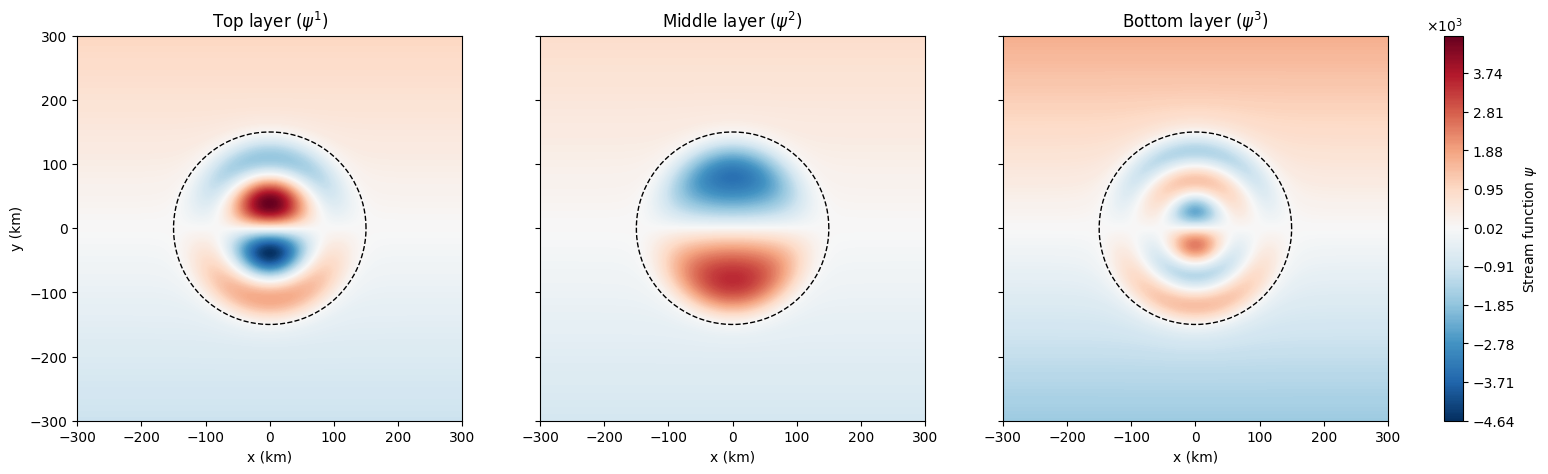}
\caption{}
\end{subfigure}
\caption{Larichev--Reznik baroclinic modon~\eqref{eq:LRbtModon}
in the case of the same eigenbasis in the inner and outer domains
for two values of the parameter tuples $(r_0,\varrho)$.}
\label{LRbcModons}
\end{figure}

The simplest dipolar vortex solution of the system~\eqref{eq:mLaysMod} is
the \textit{Larichev--Reznik barotropic modon},
which most directly extends the Larichev--Reznik modon~\cite{lari1976a} for the single-layer model
to the case of an arbitrary number of layers.
We choose $\tilde{\mathsf B}=-\tilde\varrho\mathsf E$ and $\hat{\mathsf B}=\hat\varrho\mathsf E$
with $\tilde\varrho>\rho(\mathsf F)=|\lambda_1|$ and $\hat\varrho>0$,
where $\rho(\mathsf F)$ denotes the spectral radius of~$\mathsf F$.
Then
$\nu_i=\lambda_i+\tilde\varrho>0$,
$\varsigma_i=\lambda_i-\hat\varrho<0$ and
$\tilde e_{\nu_i}=\hat e_{\varsigma_i}=e_i/\|e_i\|_{\mathsf W}^{}$;
in particular,
$\nu_m=\tilde\varrho$,
$\varsigma_m=-\hat\varrho$ and
$\tilde e_{\nu_m}=\hat e_{\varsigma_m}=\bar1/\|\bar1\|_{\mathsf W}^{}$.
Hence
$\tilde u_{\rm bg}=\beta\tilde\varrho^{-1}\bar1$,
$\hat u_{\rm bg}:=-\beta\hat\varrho^{-1}\bar1$,
$\tilde\alpha_i=\hat\alpha_i=0$, $i=1,\dots,m-1$,
\begin{gather}\nonumber
\tilde\alpha_m=\frac{r_0\beta\|\bar1\|_{\mathsf W}^{}}{\tilde\varrho J_1(\sqrt{\tilde\varrho}r_0)},
\quad
\hat\alpha_m=-\frac{r_0\beta\|\bar1\|_{\mathsf W}^{}}{\hat\varrho K_1(\sqrt{\hat\varrho}r_0)},
\\[.5ex]\label{eq:LRbtModonCondition}
\noprint{
\frac{r_0}{\sqrt{\tilde\varrho}}\frac{J_1'(\sqrt{\tilde\varrho}r_0)}{J_1(\sqrt{\tilde\varrho}r_0)}
+\frac{r_0}{\sqrt{\hat\varrho}}\frac{K_1'(\sqrt{\hat\varrho}r_0)}{K_1(\sqrt{\hat\varrho}r_0)}
=\tilde\varrho^{-1}+\hat\varrho^{-1},\quad\mbox{i.e.,}\quad
}
\frac{r_0}{\sqrt{\tilde\varrho}}\frac{J_0(\sqrt{\tilde\varrho}r_0)}{J_1(\sqrt{\tilde\varrho}r_0)}
-\frac{r_0}{\sqrt{\hat\varrho}}\frac{K_0(\sqrt{\hat\varrho}r_0)}{K_1(\sqrt{\hat\varrho}r_0)}
=\frac2{\tilde\varrho}+\frac2{\hat\varrho}.
\end{gather}
We consider the last equality as an equation
with respect to one of the parameters~$r_0$, $\tilde\varrho$ and~$\hat\varrho$
when the other two parameters are chosen.
The corresponding Larichev--Reznik barotropic modon takes the form
\begin{gather}\label{eq:LRbtModon}
\psi=\begin{cases}
\dfrac\beta{\tilde\varrho}
\left(r_0\dfrac{J_1(\sqrt{\tilde\varrho}r)}{J_1(\sqrt{\tilde\varrho}r_0)}\sin\theta-y\right)\bar1,
\ \ r\leqslant r_0,
\\[3ex]
-\dfrac\beta{\hat\varrho}
\left(r_0\dfrac{K_1(\sqrt{\hat\varrho}r)}{K_1(\sqrt{\hat\varrho}r_0)}\sin\theta-y\right)\bar1,
\ \ r>r_0
\end{cases}
\mbox{with \eqref{eq:LRbtModonCondition}.}
\end{gather}
The structure is vertically locked and purely barotropic.
On Figure~\ref{LRbtModons},
we plot the solution~\eqref{eq:LRbtModon}
for the numerical data from the Section~\ref{sec:NumericalExample} and the following values of parameters:
\begin{itemize}\itemsep=0ex
\item[(a)]
chosen:\ \ $\tilde\varrho=1.5\cdot10^{-9}$, $\hat\varrho= 3.0\cdot10^{-10}$,\\
found: \ \ $r_0\approx104.43\mbox{ km}$,\ \ $\tilde\alpha_3\approx-3.47\cdot10^4$,\ \
$\hat\alpha_3\approx-7.97\cdot10^5$;
\noprint{
\item[(b)] 
chosen:\ \ $\tilde\varrho = 1.8 \cdot10^{-9}$,\ \ $r_0= 105\mbox{ km}$, \\
found: \ \ $\hat\varrho\approx2.64\cdot 10^{-9}$,\ \ $\tilde\alpha_3\approx-1.1\cdot10^4$,\ \
$\hat\alpha_3\approx-6.3\cdot10^5$.
}
\item[(b)]
chosen:\ \ $\tilde\varrho = 1.3 \cdot10^{-9}$,\ \ $r_0= 125\mbox{ km}$, \\
found: \ \ $\hat\varrho\approx2.39\cdot 10^{-9}$,\ \ $\tilde\alpha_3\approx-1.71\cdot10^4$,\ \
$\hat\alpha_3\approx-1.81\cdot10^6$.
\end{itemize}

\begin{figure}[!ht]
\begin{subfigure}[b]{\textwidth}
\centering
\includegraphics[width=\linewidth]{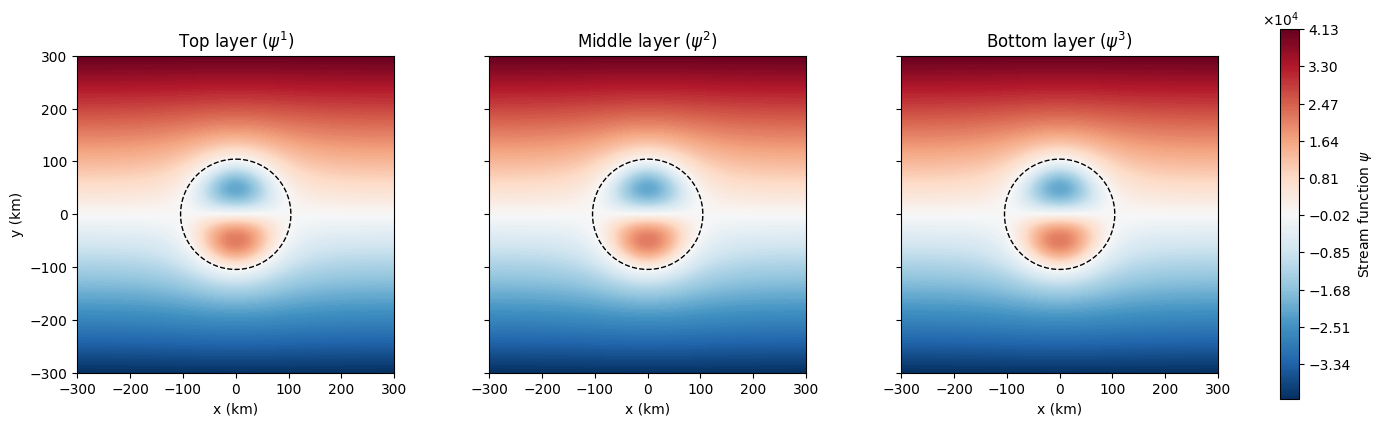}
\caption{}
\end{subfigure}
\begin{subfigure}[b]{\textwidth}
\centering
\includegraphics[width=\linewidth]{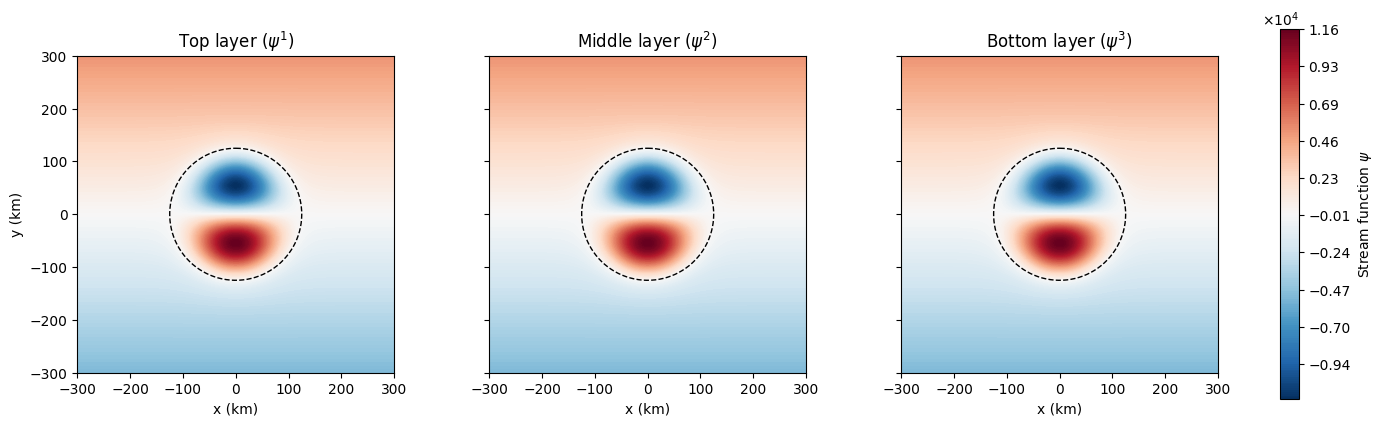}
\caption{}
\end{subfigure}
\caption{Larichev--Reznik barotropic modon~\eqref{eq:LRbtModon} 
for two values of the parameter tuples $(\tilde\varrho,\hat\varrho,r_0)$.}
\label{LRbtModons}
\end{figure}

\begin{remark}
The above construction of modons is closer to the pioneering paper~\cite{lari1976a}
on the single-layer model rather than to their recent studies, e.g., in~\cite{crow2024a}.
In particular, in the quasi-geostrophic model considered in~\cite{crow2024a},
it is allowed that the parameter~$\beta$ may essentially vary from layer to layer,
which increases the possible freedom in the course of merging solutions in the inner and outer domains.
Moreover, as far as we understand, the continuity of the stream function is chosen in~\cite{crow2024a}
as the only matching condition at the interface.
Then the stream function is in general nonzero for almost all points of the interface,
which implies a finite discontinuity of the potential vorticity at these points.
Therefore, the modons presented in~\cite{crow2024a} are not weak solutions
of the original multi-layer quasi-geostrophic model~\cite[Eq.~(2.1)--(2.4)]{crow2024a}
and, if the derivative $\psi_r$ is not continuous, even of the reduced problem~\cite[Eq.~(2.6)]{crow2024a}.
\end{remark}

\begin{remark}
In view of the enforced matching conditions $\tilde\psi=\hat\psi=0$
and $\tilde\psi_r=\hat\psi_r$ at the interface $r=r_0$,
the common values of~$\tilde\psi_r$ and~$\hat\psi_r$ at the interface $r=r_0$
for any nonzero merged solution are necessarily nonzero, except those at the points $(\pm r_0,0)$.
This explains the opposite signs of the stream function
in the inner and outer parts of neighborhoods of most of the interface points,
cf.\ Figures~\ref{LRbcModons} and~\ref{LRbtModons}.
\end{remark}

\subsection{Subalgebra family 1.2}\label{sec:Reduction12}

\subsubsection{Reduced system}\label{sec:Reduction12ReducedSystem}

For the one-dimensional subalgebra
$\mathfrak s_{1.2}^{\chi c}:=\langle\mathcal P^y+\mathcal P^x(\chi)+c_k\mathcal J^k\rangle$
of~$\mathfrak g$,
where $c:=(c_1,\dots,c_m)^{\mathsf T}$ is an arbitrary element of~$\mathop{\rm im}\mathsf F$
and $\chi$ is an arbitrary smooth function of~$t$,
an associated Lie ansatz is $\psi^i=v^i-\frac12\chi_t(t) y^2+c_iy$
or, in the vector notation,
\[
\psi=v-\frac12\chi_t(t)y^2\bar1+yc\quad\mbox{with}\quad z_1=t,\quad z_2=x-\chi(t)y.
\]
It reduces the system~\eqref{eq:mLaysMod} to the system
\begin{gather}\label{eq:RedSystSubalgS12}
\begin{split}
&\left((1+\chi^2)v^i_{22}\right)_1-c_i(1+\chi^2)v^i_{222}
-f_{i,i-1}(c_iv^{i-1}_2-c_{i-1}v^i_2-v^{i-1}_1+v^i_1)
\\
&\qquad{}+f_{i,i+1}(c_{i+1}v^i_2-c_iv^{i+1}_2-v^i_1+v^{i+1}_1)+\beta v^i_2-\chi_{11}=0.
\end{split}
\end{gather}
The condition $c\in\mathop{\rm im}\mathsf F$ modulo the $G$-equivalence
is, in fact, inessential in the course of the above Lie reduction and for its further study
and is neglected.
Recall that $\bar1$ is the all-ones $m$-column and
$\mathsf F=(f_{ij})$ is the tridiagonal matrix
that are defined in Section~\ref{sec:VectorFormOfSyst},
and $\mathsf E$ is the $m\times m$ identity matrix.
Denote $\mathsf C:=\mathop{\rm diag}c$
and $\mathsf B:=\mathop{\rm diag}\mathsf Fc-\mathsf C\mathsf F+\beta\mathsf E$,
which is a tridiagonal matrix as well.
\noprint{
Moreover, using similar arguments to those in Section~\ref{sec:mLaysMod},
we can show that, up to the $G$-equivalence, the matrix $\mathsf B$ is diagonalizable
and its eigenvalues are real and pairwise distinct.
}
Using this notation, the reduced system~\eqref{eq:RedSystSubalgS12} can be represented
in a form that is more convenient for analysis,
\begin{gather}\label{eq:RedSystSubalgS12MatrixForm}
((1+\chi^2)v_{22})_1-(1+\chi^2)\mathsf Cv_{222}+\mathsf Bv_2+\mathsf Fv_1=\chi_{11}\bar1.
\end{gather}

Depending on whether $\chi_t=0$ or $\chi_t\ne0$,
the normalizer of the subalgebra $\mathfrak s_{1.2}^{\chi c}$
in the algebra~$\mathfrak g$ is
$\langle\mathcal P^t,\mathcal P^y,\mathcal P^x(1),\mathcal J^1,\dots,\mathcal J^m,\mathcal Z(\kappa)\rangle$
or
$\langle\mathcal P^y+\mathcal P^x(\chi),\mathcal P^x(1),\mathcal J^1,\dots,\mathcal J^m,\mathcal Z(\kappa)\rangle$,
respectively,
where the parameter function~$\kappa$ runs through the set of smooth functions of~$t$.
Therefore, the algebra of induced symmetries of the system~\eqref{eq:RedSystSubalgS12MatrixHomog}
coincides with
\begin{gather*}
\big\langle\p_{z_2},\,\p_{v^1},\,\dots,\,\p_{v^m},\,\kappa(z_1)(\p_{v_1}+\cdots+\p_{v_m})\big\rangle
\quad\mbox{if}\quad\chi_t\ne0,\\
\big\langle\p_{z_1},\,\p_{z_2},\,\p_{v^1},\,\dots,\,\p_{v^m},\,\kappa(z_1)(\p_{v_1}+\cdots+\p_{v_m})\big\rangle
\quad\mbox{if}\quad\chi_t=0,
\end{gather*}
where the parameter function~$\kappa$ runs through the set of smooth functions of~$z_1$.

The reduced system~\eqref{eq:RedSystSubalgS12MatrixForm} is an inhomogeneous linear system of partial differential equations.
Therefore, it admits Lie symmetries arising from the linear superposition of solutions of
its homogeneous counterpart
\begin{gather}\label{eq:RedSystSubalgS12MatrixHomog}
((1+\chi^2)v_{22})_1-(1+\chi^2)\mathsf Cv_{222}+\mathsf Bv_2+\mathsf Fv_1=0.
\end{gather}
Most of these symmetries are genuine hidden symmetries of the original nonlinear system~\eqref{eq:mLaysMod},
which in particular shows that genuine hidden symmetries associated with the Lie reduction
with respect to any subalgebra from the family~$\{\mathfrak s_{1.2}^{\chi c}\}$ definitely exist.
To exhaustively study hidden Lie symmetries of~\eqref{eq:mLaysMod} associated with this reduction,
it is necessary to compute the maximal Lie invariance algebra of~\eqref{eq:RedSystSubalgS12MatrixHomog}.
At the same time, the system~\eqref{eq:RedSystSubalgS12MatrixHomog} is parameterized by
an arbitrary function $\chi=\chi(z_1)$ and the components of~$\theta=(\mathsf F,\beta)$ and~$c$,
so it is actually a class of systems of differential equations.
To properly describe Lie symmetries of systems from the class~\eqref{eq:RedSystSubalgS12MatrixHomog},
it is in fact necessary to solve the group classification problem for this class.

Since $\ker\mathsf F=\langle\bar 1\rangle$ and $(\mathop{\rm diag}\mathsf Fc)\bar 1=\mathsf Fc$,
we can readily verify that
\[v=\frac{\chi_{11}}\beta z_2\bar 1-\frac{\chi_1}\beta c\]
is a particular solution of~\eqref{eq:RedSystSubalgS12MatrixForm}.
This is why the system~\eqref{eq:RedSystSubalgS12MatrixForm} is reduced by
the transformation $v-\beta^{-1}\chi_{11}z_2\bar 1+\beta^{-1}\chi_1 c\mapsto v$
to the homogeneous system~\eqref{eq:RedSystSubalgS12MatrixHomog},
and the symmetry analysis of the former thus reduces to the symmetry analysis of the latter.
Another interpretation for this is to consider the modified ansatz
\[
\psi=v(z_1,z_2)-\frac{\chi_t}2y^2\bar 1+yc+\frac{\chi_{tt}}{\beta}(x-\chi y)\bar 1-\frac{\chi_t}\beta c,\quad
z_1=t,\quad z_2=x-\chi y,
\]
which straightforwardly reduces the system~\eqref{eq:mLaysMod} to the system~\eqref{eq:RedSystSubalgS12MatrixHomog}.
For the purpose of symmetry classification, we denote
the class of systems of the form~\eqref{eq:RedSystSubalgS12MatrixHomog} by~$\mathcal R$,
its arbitrary-element tuple by $\varrho:=(\mathsf F,\mathsf C,\chi,\beta)$,
and the system from the class~$\mathcal R$ with
a fixed value of the arbitrary-element tuple~$\varrho$ by~$\mathcal R_\varrho$.

\subsubsection{Group classification}\label{sec:Reduction12GroupClassification}

We begin with the most general form of Lie-symmetry vector fields of the system~$\mathcal R_\varrho$,
which constitute the maximal Lie invariance algebra~$\mathfrak g^{\rm max}_\varrho$ of this system,
\[
\tau\p_{z_1}+\xi\p_{z_2}+\eta^j\p_{v^j},
\]
where the components~$\tau$, $\xi $ and $\eta^j$ are smooth functions depending on $(z_1,z_2,v)$.
Similarly to the proof of Lemma~\ref{lem:LieSymVecFieldClassV},
we apply the Lie invariance criterion, see, e.g.,~\cite[Theorem~2.31]{olve1993A},
to this vector field and the system~$\mathcal R_\varrho$.
In the obtained equation that associated with the $i$th equation of the system~$\mathcal R_\varrho$,
we separately collect coefficients of the jet monomial
$v^k_{112}$, $v^j_2v^k_{112}$, $v^i_{12}v^j_{22}$, $v^j_{12}v^k_2$, $v^j_{12}$,
which results in the following determining equations:
\begin{gather*}
\tau_2=\tau_{v^j}=0,\quad
\xi_{v^j}=0,\quad
\eta^i_{v^jv^k}=0,\quad
\eta^i_{z_2v^i}=\xi_{22},\quad
\eta^i_{z_2v^j}=0,\ j\ne i.
\end{gather*}
These equations do not involve parameters $\chi$, $c$ and $\theta$.
Integrating them, we find that
\begin{gather}\label{eq:DetEqIntegratedRed12}
\tau=\tau(z_1),\quad
\xi=\xi(z_1,z_2),\quad
\eta^i=\eta^{ij}(z_1,z_2)v^j+\eta^{i0}(z_1,z_2)
\end{gather}
with $\eta^{ii}_2=\frac12\xi_{22}$ and $\eta^{ij}_2=0$, $j\ne i$.
Collecting summands without the unknown functions~$v^j$ and their derivatives
results in the system~$\mathcal R_\varrho$
on the tuple $(\eta^{10},\dots,\eta^{m0})^{\mathsf T}$.
Moreover, the components of this tuple do not appear in the other determining equations.

It is convenient to carry out the further computations in matrix form,
using the following representation
for Lie symmetry vector fields of the system~$\mathcal R_\varrho$
in view of~\eqref{eq:DetEqIntegratedRed12}:
\[
\tau\p_{z_1}+\xi\p_{z_2}+(\mathsf H v)^j\p_{v^j}
\]
with the matrix $\mathsf H=(\eta^{jk}(z_1,z_2))_{j,k=1}^m$, where $\tau=\tau(z_1)$, $\xi =\xi(z_1,z_2)$.
Successively collecting the coefficients of $v_{12}$ (we repeat this part of collecting for convenience),
$v_{222}$, $v_{22}$, $v_2$, $v_1$ and~$v$
in the matrix equation obtained from the Lie invariance criterion
for such a vector field and the system~$\mathcal R_\varrho$,
we derive the system
\begin{subequations}\label{eq:DetEq1.2.Part2}
\begin{gather}
\label{eq:Condition1}
2\mathsf H_2=\xi_{22}\mathsf E,
\\
\label{eq:Condition2}
\xi_1=0,\quad
[\mathsf H,\mathsf C]=0,\quad
(\tau_1-\xi_2)\mathsf C=\mathsf 0,
\\
\label{eq:Condition3}
\mathsf H_1+\big(\tau(\chi^2)_1(1+\chi^2)^{-1}\big)_1\mathsf E=\mathsf 0,
\\
\label{eq:Condition4}
[\mathsf H,\mathsf B]=\big(\tau_1+\xi_2-\tau(\chi^2)_1(1+\chi^2)^{-1}\big)\mathsf B,
\\
\label{eq:Condition5}
(1+\chi^2)\mathsf H_{22}-[\mathsf H,\mathsf F]+(2\xi_2-\tau(\chi^2)_1(1+\chi^2)^{-1})\mathsf F=\mathsf 0,
\\
\label{eq:Condition6}
\noprint{(1+\chi^2)\mathsf H_{122}+}(\chi^2)_1\mathsf H_{22}
\noprint{-(1+\chi^2)\mathsf C\mathsf H_{222}}+\mathsf F\mathsf H_1+\beta\mathsf H_2=\mathsf 0,
\end{gather}
\end{subequations}
We have in addition arranged the system in view of the equation $\xi_{22}\mathsf C=\mathsf 0$,
which is a differential consequence of the equation $(\tau_1-\xi_2)\mathsf C=\mathsf 0$.

Substituting the expressions for $\mathsf H_1$ and $\mathsf H_2$
arising from~\eqref{eq:Condition1} and~\eqref{eq:Condition3} into~\eqref{eq:Condition6},
we straightforwardly derive~$\big(\tau(\chi^2)_1(1+\chi^2)^{-1}\big)_1=0$ and thus $\mathsf H_1=\mathsf 0$
since the matrix $\mathsf F$ has nonzero non-diagonal entries
while~$\mathsf H_2$ and~$\mathsf H_{22}$ are diagonal.
Differentiating~\eqref{eq:Condition5} with respect to~$z_1$ results in the equation $(\chi^2)_1\mathsf H_{22}=\mathsf 0$.
In view of this constraint, the equation~\eqref{eq:Condition6} implies $\mathsf H_2=\mathsf 0$.
Therefore, $\mathsf H$ is a constant matrix and, in view of~\eqref{eq:Condition1}, $\xi_{22}=0$.
The equation~\eqref{eq:Condition5} takes the form
\[
[\mathsf H,\mathsf F]=\big(2\xi_2-\tau(\chi^2)_1(1+\chi^2)^{-1}\big)\mathsf F,
\]
where the coefficient $2\xi_2-\tau(\chi^2)_1(1+\chi^2)^{-1}$ is a real constant.
If this constant differs from zero,
the equation is compatible if and only if $\mathsf F$ is nilpotent,
see \cite[Lemma~4, p.~44]{jaco1962A} or \cite[Theorem~II]{roth1952a}.
The discussion in Section~\ref{sec:mLaysMod} implies that the matrix~$\mathsf F$ is not nilpotent
under the assumed physical constraint of positivity of the essential components of~$\mathsf F$.
Hence, $2\xi_2-\tau(\chi^2)_1(1+\chi^2)^{-1}=0$ and $[\mathsf H,\mathsf F]=\mathsf 0$.
Combining the former equation with~\eqref{eq:Condition4} leads to the equation
\begin{gather}\label{eq:Condition4Simplified}
[\mathsf H,\mathsf B]=(\tau_1-\xi_2)\mathsf B.
\end{gather}
If $\mathsf C=\mathsf 0$, then $\mathsf B=\beta\mathsf E$ by the construction,
and thus $[\mathsf H,\mathsf B]=0$ and $\tau_1=\xi_2$.
Otherwise, the third equation from~\eqref{eq:Condition2} directly gives $\tau_1=\xi_2$,
and hence again $[\mathsf H,\mathsf B]=0$.

The above discussion implies that the system~\eqref{eq:DetEq1.2.Part2} is equivalent to the system
\begin{gather*}
\tau_{11}=\tau_2=0,\quad
\xi_1=\xi_{22}=0,\quad
\tau_1=\xi_2,\quad
\tau(\chi^2)_1(1+\chi^2)^{-1}=2\xi_2,
\\
\mathsf H_1=\mathsf H_2=0,\quad
[\mathsf H,\mathsf F]=[\mathsf H,\mathsf C]=[\mathsf H,\mathsf B]=0,
\end{gather*}
and thus $[\mathsf H,\mathop{\rm diag}\mathsf Fc]=0$.
Analyzing this system along with~\eqref{eq:DetEqIntegratedRed12},
we solve the group classification problem for the class $\mathcal R$.

Denote by $\mathfrak H$ the space of matrices~$\mathsf H$
that commute with the matrices~$\mathsf F$, $\mathsf C$ and $\mathop{\rm diag}\mathsf Fc$.
It is obvious that $\mathfrak H\supseteq\langle\mathsf E\rangle$.
Since the matrix~$\mathsf F$ has $m$ distinct eigenvalues, the condition $[\mathsf H,\mathsf F]=0$
implies that $\mathfrak H$ is contained in the space $\mathbb R_{m-1}[\mathsf F]$
of polynomials of~$\mathsf F$ whose degrees are not greater than $m-1$.
Hence $\dim\mathfrak H\leqslant\dim\mathbb R_{m-1}[\mathsf F]=m$.

\begin{lemma}\label{lem:Reduction12Hmax}
$\mathfrak H=\mathbb R_{m-1}[\mathsf F]$ if and only if
$c\in\langle\bar1\rangle$, and thus $c=0$ if in addition $c\perp_{\mathsf W}^{}\bar1$.
\end{lemma}

\begin{proof}
If $c\in\langle\bar1\rangle$, then $\mathsf C\in\langle\mathsf E\rangle$, $\mathsf Fc=0$.
Hence the commutation conditions $[\mathsf H,\mathsf C]=[\mathsf H,\mathsf B]=0$
imply no constraints for~$\mathsf H$.
Conversely, if $\mathfrak H=\mathbb R_{m-1}[\mathsf F]$,
then in particular $[\mathsf H,\mathsf C]=0$ for $\mathsf H=\mathsf F$,
which implies the equations $f_{i,i+1}(c_{i+1}-c_i)=0$, $i=1,\dots,m-1$.
In view of the inequalities $f_{i,i+1}>0$, $i=1,\dots,m-1$,
this means that $c_1=\dots=c_m$, i.e., $c\in\langle\bar1\rangle$.
\end{proof}

\begin{lemma}\label{lem:Reduction12Hmin}
$\mathfrak H=\langle\mathsf E\rangle$ if all $c_i$ are pairwise distinct.
\end{lemma}

\begin{proof}
The matrix $\tilde{\mathsf F}:=-\mathsf\Omega\mathsf F\mathsf\Omega$
with $\mathsf\Omega:=\mathop{\rm diag}((-1)^k,\,k=1,\dots,m)$
is tridiagonal with strictly positive subdiagonal, diagonal and superdiagonal entries.
Its $k$th power $\tilde{\mathsf F}^k$ is a $(2k+1)$-diagonal matrix with strictly positive entries
on these diagonals.
Then, $\mathsf F^k=(-1)^k\mathsf\Omega\tilde{\mathsf F}^k\mathsf\Omega$,
i.e., $(\mathsf F^k)_{ij}=(-1)^{j-i+k}(\tilde{\mathsf F}^k)_{ij}$.
Therefore, $\mathsf F^k$ is a $(2k+1)$-diagonal matrix with nonzero entries on these diagonals
and chessboard sign pattern on them.

Take an arbitrary matrix $H\in\mathfrak H$.
Suppose that the degree~$l$ of~$H$ as a polynomial of~$\mathsf F$ is greater than zero
and $a_l$ is the leading coefficient of this polynomial.
Then the condition $[\mathsf H,\mathsf C]=0$ implies that
$h_{i,i+l}(c_{i+l}-c_i)=0$.
At the same time, $c_{i+l}\ne c_i$ and $h_{i,i+l}=a_l(\mathsf F^l)_{i,i+l}\ne0$
since $a_l\ne0$ and $(\mathsf F^l)_{i,i+l}\ne0$,
which gives a contradiction.
Therefore $l=0$, i.e., $\mathsf H\in\langle\mathsf E\rangle$.
\end{proof}

It follows from Lemmas~\ref{lem:Reduction12Hmax} and~\ref{lem:Reduction12Hmin}
that there are two outermost situations,
all $c_i$ are the same, where $\mathfrak H=\mathbb R_1[\mathsf F]$,
and all $c_i$ are pairwise distinct, where $\mathfrak H=\langle\mathsf E\rangle$.
For $m=2$, only these situations are possible.
If $m>2$ and $c_i$ are not all the same and are not all pairwise distinct,
then  $\langle\mathsf E\rangle\subseteq\mathfrak H\subsetneq\mathbb R_{m-1}[\mathsf F]$,
where in the first inclusion we can have equality,
but, according to Lemma~\ref{lem:Reduction12Hmax}, not in the second one.

We exhaustively describe the space~$\mathfrak H$ depending on~$\mathsf F$ and $\mathsf C$
for low values of~$m$,
\begin{gather*}
\hspace*{-\mathindent}\mathfrak H\mbox{ for }m=2\colon\quad
\circ\ \mathfrak H=\mathbb R_1[\mathsf F]\quad\mbox{if}\quad c\in\langle\bar1\rangle\quad\mbox{and}\quad
\circ\ \mathfrak H=\langle\mathsf E\rangle\quad\mbox{otherwise}.
\\
\hspace*{-\mathindent}\mathfrak H\mbox{ for }m=3\colon\\
\circ\ \mathbb R_2[\mathsf F]\quad\mbox{if}\quad c\in\langle\bar1\rangle,\\
\circ\ \big\langle\mathsf E,\mathsf F^2+(f_{12}+f_{21}+f_{23})\mathsf F\big\rangle
\quad\mbox{if}\quad f_{12}=f_{32},\quad c_1=c_3\ne c_2,\\
\circ\ \langle\mathsf E\rangle\quad\mbox{otherwise}.
\\
\hspace*{-\mathindent}\mathfrak H\mbox{ for }m=4\colon\\
\circ\ \mathbb R_2[\mathsf F]\quad\mbox{if}\quad c\in\langle\bar1\rangle,
\\
\circ\ \langle\mathsf E,
\mathsf F^3+(f_{12}+2f_{21}+f_{23}+f_{32})\mathsf F^2+(f_{21}+f_{23}+f_{32})(f_{12}+f_{21})\mathsf F\rangle
\\\hphantom{\circ\ }\mbox{if}\quad f_{43}=f_{12},\ f_{34}=f_{21},\ \ c_1=c_4\ne c_2=c_3,\\
\circ\ \big\langle\mathsf E,\mathsf F^2+(f_{12}+f_{21}+f_{23})\mathsf F\big\rangle
\quad\mbox{if}\quad f_{34}=f_{12}-f_{32},\ \ f_{43}=f_{21}-f_{23},\ \ c_1=c_3\ne c_2=c_4,\\
\circ\ \langle\mathsf E\rangle\quad\mbox{otherwise}.
\end{gather*}
Note that for $m=2$ and $m=3$,
the conditions $[\mathsf H,\mathsf F]=0$ and $[\mathsf H,\mathsf C]=0$ imply
$[\mathsf H,\mathop{\rm diag}\mathsf Fc]=0$,
but this is not the case for $m=4$.

\begin{remark}\label{rem:SpaceHParams}
If $c\notin\langle\bar1\rangle$,
the space~$\mathfrak H$ is larger than the minimal case $\mathfrak H=\langle\mathsf E\rangle$
only under specific additional constraints
not only on the subalgebra parameter tuple~$c$, which can be chosen arbitrarily,
but also on the essential components of the matrix~$\mathsf F$, which define the model.
Although for $m\in\{3,4\}$ the constraints are linear in $(c,\mathsf F)$
and decoupled with respect to~$c$ and~$\mathsf F$,
after analysing the commutation relations $[\mathsf H,\mathsf F]=[\mathsf H,\mathsf C]=[\mathsf H,\mathsf B]=0$,
we conjecture that in higher dimensions~$m$, there exist nonlinear constrains coupling $c$ and~$\mathsf F$.
\end{remark}

\begin{remark}\label{rem:MinMaxDimOfH}
While we have exhaustively described the space~$\mathfrak H$ for low values of~$m$,
the same problem for an arbitrary value of~$m$ remains open.
It is not even clear what are justified conjectures about deeper properties of this space.
Furthermore, as noted in Remark~\ref{rem:SpaceHParams},
non-maximal extensions of~$\mathfrak H$ definitely require
that the corresponding model parameters satisfy a system of algebraic equations
and thus their set is of measure zero in the space of admitted model parameters.
This is why for the remainder of Section~\ref{sec:CodimOneLieReds},
we consider the entire space $\mathfrak H=\mathbb R_{m-1}[\mathsf F]$ if $c\in\langle\bar1\rangle$
and its subspace~$\langle\mathsf E\rangle$ otherwise,
which correspond to the \textit{completely decoupled} and \textit{generically coupled}
reduced systems~\eqref{eq:RedSystSubalgS12MatrixHomog}, respectively.
This remark is also relevant for the analogous reduced systems~\eqref{eq:RedSys13Homogen}.
At the same time, the complete description of the nontrivial extensions of the space~$\mathfrak H$
may be of interest for finding bifurcation values of the model parameters.
\end{remark}

The kernel Lie invariance algebra $\mathfrak g_{\mathcal R}^\cap$ of the class $\mathcal R$
is spanned by the vector fields
\[
\p_{z_2},\quad
v^j\p_{v^j},\quad
\p_{v^1},\quad\dots,\quad
\p_{v^m},\quad
\kappa(\p_{v^1}+\cdots+\p_{v^m}),
\]
where $\kappa$ is an arbitrary smooth function of~$z_1$.
Among these vector fields, only $v^j\p_{v^j}$ is not induced
by a Lie-symmetry vector field of the original system~$\mathcal M_\theta$.
Any system~$\mathcal R_\varrho$ from the class~$\mathcal R$
is invariant with respect to the algebra~$\mathfrak g^{\rm gen}_\varrho$ spanned by
\begin{gather*}
\p_{z_2},\quad
(\mathsf Hv)^j\p_{v^j},\quad
\zeta^j(p,q)\p_{v^j},
\end{gather*}
where the tuple $(\zeta^1,\dots,\zeta^m)$
runs through the solution set of the system~$\mathcal R_\varrho$,
and $\mathsf H$ runs through a basis of the space $\mathfrak H$ of (constant) matrices
commuting with the matrices~$\mathsf F$, $\mathsf C$ and $\mathsf B$.
The subalgebra of~$\mathfrak g^{\rm gen}_\varrho$ constituted by induced symmetries
is spanned by the vector fields~$\p_{z_2}$ and $\zeta^j(z_1,z_2)\p_{v^j}$
with $\zeta^j_2=0$ and $\zeta^j_1=\zeta^k_1$.
All the other elements of~$\mathfrak g^{\rm gen}_\varrho$ are hidden symmetries
for the original system~$\mathcal M_\theta$.

The classification of extensions of~$\mathfrak g^{\rm gen}_\varrho$,
i.e., the cases where $\mathfrak g^{\rm max}_\varrho\ne\mathfrak g^{\rm gen}_\varrho$,
is carried out modulo the action of the equivalence transformations in the class $\mathcal R$
that are induced by elements of the point symmetry pseudogroup~$G$ of the system~$\mathcal M_\theta$,
which is the same for all systems from the class~$\mathcal M$.
The only essential among these equivalence transformations are the translations with respect to~$z_1$.
Since $\tau_{11}=\tau_2=0$, $\xi_1=\xi_{22}=0$ and $\tau_1=\xi_2$,
we have $\tau=at+b$ and $\xi_2=a$ for some constants~$a$ and~$b$.
Then the classifying equation $\tau(\chi^2)_1(1+\chi^2)^{-1}=2\xi_2$ takes the form
$(at+b)(1+\chi^2)_1=2a(1+\chi^2)$.
Up to translations with respect to~$z_1$ and multiplying the classifying equation by a nonzero constant,
there are two inequivalent cases of this equation,
in addition to the identity corresponding to the absence of extension,
with $(a,b)=(1,0)$ and with $(a,b)=(0,1)$.
In view of this argument, the inequivalent extensions of the algebra $\mathfrak g^{\rm gen}_\varrho$
are exhausted by two cases,
\begin{align*}
\chi_1=0\colon&\quad
\mathfrak g^{\rm max}_\varrho=\mathfrak g^{\rm gen}_\varrho+\langle\p_{z_1}\rangle,
\\
\chi=\pm\sqrt{\alpha z_1^2-1}\colon&\quad
\mathfrak g^{\rm max}_\varrho=\mathfrak g^{\rm gen}_\varrho+\langle z_1\p_{z_1}+z_2\p_{z_2}\rangle,
\end{align*}
where $\alpha$ is a positive constant.
Recall that the extension~$\p_{z_1}$ is induced by the Lie-symmetry vector field~$\p_t$ of~$\mathcal M_\theta$,
whereas the extension~$z_1\p_{z_1}+z_2\p_{z_2}$ is a hidden Lie symmetry of~$\mathcal M_\theta$.

We split further consideration into cases depending on two criteria,
the dimension of the space~$\mathfrak H$, see Remark~\ref{rem:MinMaxDimOfH},
and the presence of Lie-symmetry vector fields with nonzero $z_1$-components.
According to the first criterion, we have the \emph{decoupled} and the \emph{coupled} cases,
where $\dim\mathfrak H$ is minimal and maximal, respectively.
Each of these cases additionally splits into
the \emph{general}, the \emph{scale-invariant} and the \emph{shift-invariant} cases
using the second criterion and meaning shifts and scalings with respect to $z_1$.

\subsubsection{Completely decoupled case}\label{sec:ReductionA12DecoupledCase}

In view of Lemma~\ref{lem:Reduction12Hmax}, the most singular case of the group classification
of the class $\mathcal R$ of systems~\eqref{eq:RedSystSubalgS12MatrixHomog}
takes place if and only if $c\in\langle\bar1\rangle$, i.e., $c=0$ modulo the $G$-equivalence,
which corresponds to the maximal possible dimension of the general algebra $\mathfrak g^{\rm gen}_\varrho$.
Consequently, we have $\mathsf C=0$, which implies $\mathsf B=\beta\mathsf E$.
Under these constraints, the system~$\mathcal R_\varrho$ takes the form
\begin{gather}\label{eq:RedSystSubalgS12Decoupled}
((1+\chi^2)v_{22})_1+\beta v_2+\mathsf Fv_1=0,
\end{gather}
and thus it is decoupled.
In this case, the algebra $\mathfrak g^{\rm gen}_\varrho$ coincides with the span
\[
\langle\p_{z_2},\
(\mathsf F^{i-1}v)^j\p_{v^j},\,i=1,\dots,m,\
\zeta^j(p,q)\p_{v^j}\rangle,
\]
where the tuple $(\zeta^1,\dots,\zeta^m)$
runs through the solution set of the system~$\mathcal R_\varrho$.

Using the change of the dependent variables $\tilde v:=\mathsf P^{-1}v$,
we diagonalize the matrix~$\mathsf F$ to its eigenvalue matrix
$\mathsf\Lambda=\mathop{\rm diag}(\lambda_1,\dots,\lambda_m)$, see the end of Section~\ref{sec:PropertiesOfMatrixF}
and thus we map the system~\eqref{eq:RedSystSubalgS12Decoupled} to the decoupled system
$
((1+\chi^2)\tilde v_{22})_1+\beta\tilde v_2+\mathsf\Lambda\tilde v_1=0.
$
or, componentwise,
\begin{gather}\label{eq:RedSys12DecoupledModIthEq}
((1+\chi^2)\tilde v^i_{22})_1+\beta\tilde v^i_2+\lambda_i\tilde v^i_1=0, \quad i=1,\dots,m.
\end{gather}
Since $\lambda_m=0$, the $m$th equation of the system~\eqref{eq:RedSys12DecoupledModIthEq} takes the form
$((1+\chi^2)\tilde v^m_{22})_1+\beta\tilde v^m_2=0$.
This equation obviously integrates to \[\big((1+\chi^2)\tilde v^m_2\big)_1+\beta\tilde v^m=\varsigma,\]
where $\varsigma=\varsigma(z_1)$ is an arbitrary smooth function.
Changing the variables in the last equation according to
\begin{gather}\label{eq:RedSystSubalgS12DecoupledChangeVariables}
\hat z_1=\beta\int\frac{{\rm d}z_1}{1+\chi^2},\quad
\hat z_2=z_2,\quad
\hat v^m=\tilde v^m-\frac \varsigma\beta,
\end{gather}
we derive the famous Klein--Gordon equation in the light-cone variables
\begin{gather}\label{eq:KG}
\hat v^m_{\hat z_1\hat z_2}+\hat v^m=0.
\end{gather}
There are large families of known exact solutions for this equation,
including invariant ones, see, e.g.,
\cite[Sections~1.1--1.3]{mill1977A} for separation of variables based on symmetries
for the counterpart of~\eqref{eq:KG} in the standard spacetime coordinates or
the collection of results for this counterpart in~\cite[Section~4.1.3]{poly2002A}.
The generalized symmetries and the local conservation laws of the equation~\eqref{eq:KG},
which can be interpreted as the corresponding hidden objects of the original system~\eqref{eq:mLaysMod},
were exhaustively described in~\cite{opan2020e}.

As a result, the collection of the $\mathfrak s_{1.2}^{\chi}$-invariant solutions
of the original system~\eqref{eq:mLaysMod} can be represented
in terms of the general solutions of the equations~\eqref{eq:RedSys12DecoupledModIthEq},
\begin{gather}\label{eq:GeneralSolutionDecoupledRedS12}
\solution
\psi=\sum_{i=1}^{m-1}\tilde v^i(z_1,z_2)e_i+\left(\hat v^m(\hat z_1,\hat z_2)
-\frac{\chi_t}2 y^2+\frac{\chi_{tt}}\beta z_2+\frac\varsigma\beta\right)\bar 1,
\end{gather}
where $z_1:=t$, $\hat z_1:=\beta\int(1+\chi^2)^{-1}{\rm d}t$,
$\hat z_2=z_2:=x-\chi(t)y$,
$e_i$ is an eigenvector of the matrix~$\mathsf F$ corresponding to its eigenvalue~$\lambda_i$ and,
moreover, $\lambda_m=0$ and $e_m=\bar1$.
The function $\tilde v^i$ is an arbitrary solution
of the $i$th equation in~\eqref{eq:RedSys12DecoupledModIthEq}, $i=1,\dots,m-1$,
the function $\hat v^m$ is an arbitrary solution of the Klein--Gordon equation~\eqref{eq:KG},
$\tilde v^m(z_1,z_2)=\hat v^m(\hat z_1,\hat z_2)$,
and $\chi$ and $\varsigma$ are arbitrary functions of $t$.
The function~$\varsigma$ can be neglected modulo the $G$-equivalence.

Since each of the equations~\eqref{eq:RedSys12DecoupledModIthEq} and~\eqref{eq:KG} is linear and homogeneous,
any linear combination of its solutions is its solution as well.
This is why the representation~\eqref{eq:GeneralSolutionDecoupledRedS12} gives
wide families of exact solutions of the original system~\eqref{eq:mLaysMod}.
The remainder of this subsection is devoted to finding
closed-form solutions of the equations~\eqref{eq:RedSys12DecoupledModIthEq}
using its Lie reductions (even with respect to induced symmetries)
or the other methods like Shapovalov--Shirokov noncommutative integration~\cite{shap1995a},
especially, its simplest variant, which can be described in terms of complexification.

Consider now the $i$th equation in~\eqref{eq:RedSys12DecoupledModIthEq}.
Using the results of Section~\ref{sec:Reduction12GroupClassification},
we obtain that any equation of the form~\eqref{eq:RedSys12DecoupledModIthEq}
is invariant with respect to the algebra~$\mathfrak a_{\chi\lambda_i}$
spanned~by the vector fields
\begin{gather*}
\p_{z_2},\quad
\tilde v^i\p_{\tilde v^i},\quad
\zeta^i(z_1,z_2)\p_{\tilde v^i},
\end{gather*}
where the function~$\zeta^i$ runs through the solution set of~\eqref{eq:RedSys12DecoupledModIthEq} with the fixed $\lambda_i$, $\beta$ and~$\chi$.
The inequivalent extensions of the algebra~$\mathfrak a_{\chi\lambda_i}$ are exhausted by
$\langle\p_{z_1}\rangle$ for $\chi_1=0$ and $\langle z_1\p_{z_1}+z_2\p_{z_2}\rangle$ for $\chi=\pm\sqrt{\alpha z_1^2-1}$,
where $\alpha$ is a positive constant.

\paragraph{General case.}
For general $\chi$,
a complete list of inequivalent one-dimensional subalgebras of~$\mathfrak a_{\chi\lambda_i}$
that are appropriate for Lie reduction consists of the subalgebras of the form
$\langle\p_{z_2}+b_i\tilde v^i\p_{\tilde v^i}\rangle$,
where $b_i$ is an arbitrary real constant.
An ansatz corresponding to such a subalgebra is given by
$\tilde v^i={\rm e}^{b_iz_2}\varphi^i(\omega)$ with $\omega:=z_1$,
and it reduces the equation~\eqref{eq:RedSys12DecoupledModIthEq} with the fixed~$i$
to the homogeneous linear first-order ordinary differential equations
\begin{gather*}
b_i^2\big((1+\chi^2)\varphi^i)_\omega+\lambda_i\varphi^i_\omega+\beta b_i\varphi^i=0,
\end{gather*}
which can be easily integrated.
\noprint{
\[
\varphi^i(\omega):=A_i\exp\int\frac{-(2b_i^2\chi\chi_\omega+\beta b_i)}{b_i^2(1+\chi^2)+\lambda_i}{\rm d}\omega
\quad\mbox{if}\quad b_i\ne0\ \mbox{or}\ i\ne m,
\]
and $\varphi^m$ is an arbitrary function of~$\omega$ if $b_m=0$.
}
Pulling its general solution back with respect to the above ansatz,
we obtain an explicit solution of the equation~\eqref{eq:RedSys12DecoupledModIthEq},
\begin{gather*}
\solutionRedEq
\tilde v^i(z_1,z_2)=
\frac{A_i}{b_i^2(1+\chi^2)+\lambda_i}
\exp\left(b_iz_2-\int\frac{\beta b_i{\rm d}z_1}{b_i^2(1+\chi^2)+\lambda_i}\right)
\quad\mbox{if}\quad b_i\ne0\ \mbox{or}\ i\ne m,\\
\tilde v^m(z_1,z_2)=\varphi^m(z_1)\quad\mbox{if}\quad b_m=0,
\end{gather*}
where $A_i$ and~$b_i$ are arbitrary constants,
$\varphi^m$ is an arbitrary function of~$z_1$, and this function can be neglected modulo the $G$-equivalence.
Applying the complexification trick, we can assume that the constants~$A_i$ and~$b_i$ are complex.
We choose the constant $b_i$ to be imaginary, $b_i=\gamma_i{\rm i}$ with $\gamma_i\ne0$,
and obtain wave-like solutions of the form
\[
\solutionRedEq
\tilde v^i(z_1,z_2)=\frac{\hat A_i}{\gamma_i^2(1+\chi^2)-\lambda_i}
\cos\left(\gamma_iz_2+\int\frac{\beta\gamma_i{\rm d}z_1}{\gamma_i^2(1+\chi^2)-\lambda_i}+\alpha_i\right),
\]
where $\hat A_i$, $\alpha_i$ and~$\gamma_i$ are arbitrary real constants with $\gamma_i\ne0$.
Since $\lambda_i\leqslant0$, all such solutions are bounded.

\paragraph{Scale-invariant case.}
This case take place when $\chi=\pm\sqrt{\alpha z_1^2-1}$ with $\alpha>0$,
and thus we have one more family of inequivalent one-dimensional subalgebras of the algebra~$\mathfrak a_{\chi\lambda_i}$,
\[
\langle
z_1\p_{z_1}+z_2\p_{z_2}+b_i\tilde v^i\p_{\tilde v^i}
\rangle,
\]
again parameterized by an arbitrary real constant~$b_i$.
For any subalgebra from this family, an ansatz constructed for~$\tilde v$ using it
can be chosen in the form $\tilde v^i=|z_1|^{b_i}\varphi^i(\omega)$ with $\omega:=z_2/z_1$,
which reduces the equation~\eqref{eq:RedSys12DecoupledModIthEq} to
\begin{gather}\label{eq:RedSystSubalgS12b}
\alpha\omega\varphi^i_{\omega\omega\omega}-\alpha b_i\varphi^i_{\omega\omega}+
(\lambda_i\omega-\beta)\varphi^i_\omega-\lambda_i b_i\varphi^i=0.
\end{gather}
Since $\lambda_m=0$, the $m$th equation of the system~\eqref{eq:RedSystSubalgS12b}
is $\big(\omega\varphi^m_{\omega\omega}-(b_m+1)\varphi^m_\omega-\beta\alpha^{-1}\varphi^m\big)_\omega=0$
and can be integrated in terms of the Bessel functions $J_\nu(z)$ and $Y_\nu(z)$ for $\omega<0$
and in terms of the modified Bessel functions $I_\nu(z)$ and $K_\nu(z)$ for $\omega>0$,
\begin{gather*}
\varphi^m=|\omega|^{\frac12b_m+1}\big(
 A_{1m}J_{b_m+2}(\sqrt{\beta'|\omega|})
+A_{2m}Y_{b_m+2}(\sqrt{\beta'|\omega|})\big)+A_{3m} \quad\mbox{if}\quad \omega<0,
\\
\varphi^m=\omega^{\frac12b_m+1}\big(
 A_{1m}I_{b_m+2}(\sqrt{\beta'\omega})
+A_{2m}K_{b_m+2}(\sqrt{\beta'\omega})\big)+A_{3m} \quad\mbox{if}\quad \omega>0,
\end{gather*}
where $A_{1m}$, $A_{2m}$ and $A_{3m}$ are arbitrary constants, and $\beta':=4\beta\alpha^{-1}$.
When $b_i=0$, $i=1,\dots,m-1$, the general solution of the $i$th equation
of the form~\eqref{eq:RedSystSubalgS12b} can be written as
\begin{gather*}
\varphi^i=A_{1i}{\rm Mi}_{\mu_i,\frac12}(\nu_i\omega)
+A_{2i}{\rm Wi}_{\mu_i,\frac12}(\nu_i\omega){\rm d}\omega+A_{3i},
\end{gather*}
where $\mu_i:=\frac\beta2(-\alpha\lambda_i)^{-1/2}$ and $\nu_i:=2(-\alpha^{-1}\lambda_i)^{1/2}$,
and by ${\rm Mi}_{\mu,\kappa}(z)$ and ${\rm Wi}_{\mu,\kappa}(z)$ we denote antiderivatives of Whittaker functions,
\[
{\rm Mi}_{\mu,\kappa}(z):=\int M_{\mu,\kappa}(z){\rm d}z,\quad
{\rm Wi}_{\mu,\kappa}(z):=\int W_{\mu,\kappa}(z){\rm d}z.
\]
Since $\alpha\lambda_i<0$ for all $i=1,\dots,m-1$, the constants $\mu_i$ and $\nu_i$ are real,
which implies that the derived solutions are real-valued.
Pulling the obtained solution back with respect to the used ansatz,
we obtain the following solution of the system~\eqref{eq:RedSys12DecoupledModIthEq}:
\begin{gather*}
\solutionRedEq\chi=\pm\sqrt{\alpha z_1^2-1}:
\\
\tilde v^i=A_{1i}{\rm Mi}_{\mu_i,\frac12}(\nu_i\omega)+A_{2i}{\rm Wi}_{\mu_i,\frac12}(\nu_i\omega),\quad i=1,\dots,m-1,
\\
\tilde v^m=|z_1|^{b_m}|\omega|^{\frac12b_m+1}\big(
 A_{1m}J_{b_m+2}(\sqrt{\beta'|\omega|})
+A_{2m}Y_{b_m+2}(\sqrt{\beta'|\omega|})\big)+A_{3m}|z_1|^{b_m} \quad\mbox{if}\quad \omega<0,
\\
\tilde v^m=|z_1|^{b_m}\omega^{\frac12b_m+1}\big(
 A_{1m}I_{b_m+2}(\sqrt{\beta'\omega})
+A_{2m}K_{b_m+2}(\sqrt{\beta'\omega})\big)+A_{3m}|z_1|^{b_m} \quad\mbox{if}\quad \omega>0,
\end{gather*}
where $\omega:=z_2/z_1$,
$A_{1i}$, $A_{2i}$ and $A_{3m}$ are arbitrary real constants, and $\beta':=4\beta\alpha^{-1}$.
Moreover, up to the $G$-equivalence, the constant $A_{3m}$ can be gauged to zero.

\paragraph{Shift-invariant case.}
When $\chi_1=0$, each equation~\eqref{eq:RedSys12DecoupledModIthEq} with $i<m$
is a linearized Benjamin--Bona--Mahony (BBM) equation,
\begin{gather}\label{eq:LinBBMSyst}
\mathcal B_i\colon\quad
(1+\chi^2)\tilde v^i_{122}+\beta\tilde v^i_2+\mathsf\lambda_i\tilde v^i_1=0.
\end{gather}
The dispersion relation for this equation is \[\omega^i(k)=-\frac{\beta k}{(1+\chi^2)k^2-\lambda_i}.\]
Since $\lambda_i<0$ if $i<m$, the phase velocity $\omega^i(k)/k$ is bounded on the entire $\mathbb R_k$.
Therefore, given initial condition $\tilde v^i(0,z_2):=g^i(z_2)$,
the general solution of~\eqref{eq:LinBBMSyst} can be written using Fourier transform as
\[
\solutionRedEq\chi_t=0\colon\quad
\tilde v^i(z_1,z_2)=\frac1{\sqrt{2\pi}}\int_{-\infty}^\infty \hat g^i(k){\rm e}^{{\rm i}(kz_2-\omega^i(k) z_1)}{\rm d}k,
\]
where $\hat g^i(k)$ is the Fourier transform of the initial value $g^i(z_2)$.
If the initial value is sufficiently regular, then
the corresponding solution is bounded in the entire space $\mathbb R^2_{z_1,z_2}$.
More specifically, if $g^i(z_2)$ belongs to the Sobolev space $H^s(\mathbb R)$ for $s>1/2$,
the solution $\tilde v^i(z_1,\cdot)$ remains in that space for all~$z_1$.
The (one-dimensional) Sobolev embedding theorem, which states that $H^s(\mathbb R)\subset L^\infty(\mathbb R)$ for $s>1/2$,
implies that the solution $\tilde v$ is bounded such that
$|\tilde v(z_1,z_2)|\leqslant C\|\tilde v(z_1,\cdot)\|_{H^s(\mathbb R)}$ for some constant $C$
and for all $z_1$.

The maximal Lie invariance algebra $\mathfrak g_i$ of the linearized BBM equation $\mathcal B_i$
is spanned by the vector fields
\[
\p_{z_1},\quad
\p_{z_2},\quad
\tilde v^i\p_{\tilde v^i},\quad
\zeta^i(z_1,z_2)\p_{\tilde v^i},
\]
where the parameter-function $\zeta$ runs through the solution set of this equation.
The algebra~$\mathfrak g_i$ splits over
the infinite-dimensional ideal $\mathfrak g_i^{\rm lin}:=\{\zeta^i(z_1,z_2)\p_{\tilde v^i}\}$,
which is associated with the linear superposition of solutions of~\eqref{eq:LinBBMSyst}.
More specifically, the algebra $\mathfrak g_i$ is the semidirect sum
of the ideal~$\mathfrak g_i^{\rm lin}$ with the complementing subalgebra $\mathfrak g_i^{\rm ess}$,
which is three-dimensional, abelian
and spanned by the vector fields~$\p_{z_1}$, $\p_{z_2}$ and $\tilde v^i\p_{\tilde v^i}$.

The structure of the algebra $\mathfrak g_i^{\rm ess}$ suggests that
most Lie and generalized reductions of~$\mathcal B_i$ result in solutions
that have at least polynomial growth as $z_1$ or $z_2$ tends to infinity
and thus their physical relevance is under question.
At the same time, one can easily construct more interesting solutions of~$\mathcal B_i$.
Using a subalgebra of the form $\langle \p_{z_1}+\alpha_i\p_{z_2}+\gamma_i\tilde v^i\p_{\tilde v^i}\rangle$,
where $\alpha_i$ and $\gamma_i$ are arbitrary real constants with $\alpha_i\geqslant0$,
we construct modulated wave ansatz $\tilde v^i={\rm e}^{\gamma_i z_1}\varphi^i(\omega)$,
where $\omega=z_2-\alpha_i z_1$, and reduce~\eqref{eq:LinBBMSyst}
to the constant-coefficient homogeneous linear third-order ordinary differential equation
\begin{gather}\label{eq:ReductionOfLinBBM}
\alpha_i(1+\chi^2)\varphi^i_{\omega\omega\omega}-\gamma_i(1+\chi^2)\varphi^i_{\omega\omega}
+(\alpha_i\lambda_i-\beta)\varphi^i_\omega-\gamma_i\lambda_i\varphi^i=0.
\end{gather}
Since the parameter~$\beta$ is positive,
the real parts of each root of the characteristic polynomial of the equation~\eqref{eq:ReductionOfLinBBM} are nonzero.
Hence, each solution of~\eqref{eq:ReductionOfLinBBM} is unbounded on the entire space $\mathbb R_\omega$.
Nevertheless, we can apply the complexification trick and assume that $\gamma_i$ is purely imaginary,
that is, $\gamma_i={\rm i}\delta_i$ for some real parameter $\delta_i$.
Substituting the ansatz $\varphi^i=\exp({\rm i}r\omega)$ into the equation~\eqref{eq:ReductionOfLinBBM}
gives that $r$ is a root of the cubic polynomial $a_3r^3+a_2r^2+a_1r+a_0$ with
$a_0:=\delta_i\lambda_i$,
$a_1:=\beta-\alpha_i\lambda_i$,
$a_2:=-\delta_i(1+\chi^2)$ and
$a_3:=\alpha_i(1+\chi^2)$.
Real roots of this polynomial definitely exist
and each of them leads to a solution family of~$\mathcal B_i$.
In total, we have the following collection of simple-wave periodic solutions
of the equation~$\mathcal B_i$, which can be, of course, linearly superposed:
\begin{gather}\label{eq:SolutionLinBBM}
\solutionRedEq
\tilde v^i(z_1,z_2)=C\cos\big(\delta_iz_1+r_i(z_1-\alpha_iz_2)+\theta\big),
\end{gather}
where $\delta_i$, $\alpha_i$, $C$ and $\theta$ are arbitrary real constants,
and $r_i$ is a real root of the above polynomial,
whose coefficients are parameterized by~$\delta_i$ and~$\alpha_i$.
In particular, this polynomial has three distinct real roots if (and only if)
its discriminant $\Delta=18a_0a_1a_2a_3-4a_0a_2^3+a_1^2a_2^2-4a_1^3a_3-27a_0^2a_3^2$
is positive.

We illustrate the solution~\eqref{eq:SolutionLinBBM} using the numerical data from Section~\ref{sec:NumericalExample}.
Assuming that $\chi$ is constant, substituting the solution~\eqref{eq:SolutionLinBBM} into~\eqref{eq:GeneralSolutionDecoupledRedS12}
leads to
\begin{gather}\label{eq:SolutionLinBBMForPlotting}
\psi=\sum_{i=1}^{3}\psi_{i0}\cos\big(\delta_iz_1+r_i(z_1-\alpha_iz_2)+\theta\big)e_i,
\end{gather}
where $\psi_{i0}$ are scaling constants.
The parameters for this simulation are defined as follows:
\begin{gather*}
\chi=0.7,
\\
\alpha_1=1.2\cdot10^{-4},\quad
\delta_1=4\cdot10^{-6},\quad
r_1=3.34\cdot10^{-2},\quad
\theta_1=\pi/3,\quad
\psi_{10}=15000,
\\
\alpha_2=1.5\cdot10^{-4},\quad
\delta_2=3\cdot10^{-6},\quad
r_2=2.00\cdot10^{-2},\quad
\theta_2=\pi/4,\quad
\psi_{20}=12000,
\\
\alpha_3=1.4\cdot10^{-4},\quad
\delta_3=3.5\cdot10^{-6},\quad
r_1=2.50\cdot10^{-2},\quad
\theta_3=\pi/6,\quad
\psi_{30}=10000.
\end{gather*}
The resulting solution is plotted on Figure~\ref{fig:Red112TravelingRossbyWaves} at time snapshots $t=0$~s, $t=3600$~s, and $t=7200$~s.
Physically, these profiles represent the evolution of traveling baroclinic Rossby waves.
\begin{figure}[!ht]
\centering
\includegraphics[width=\linewidth]{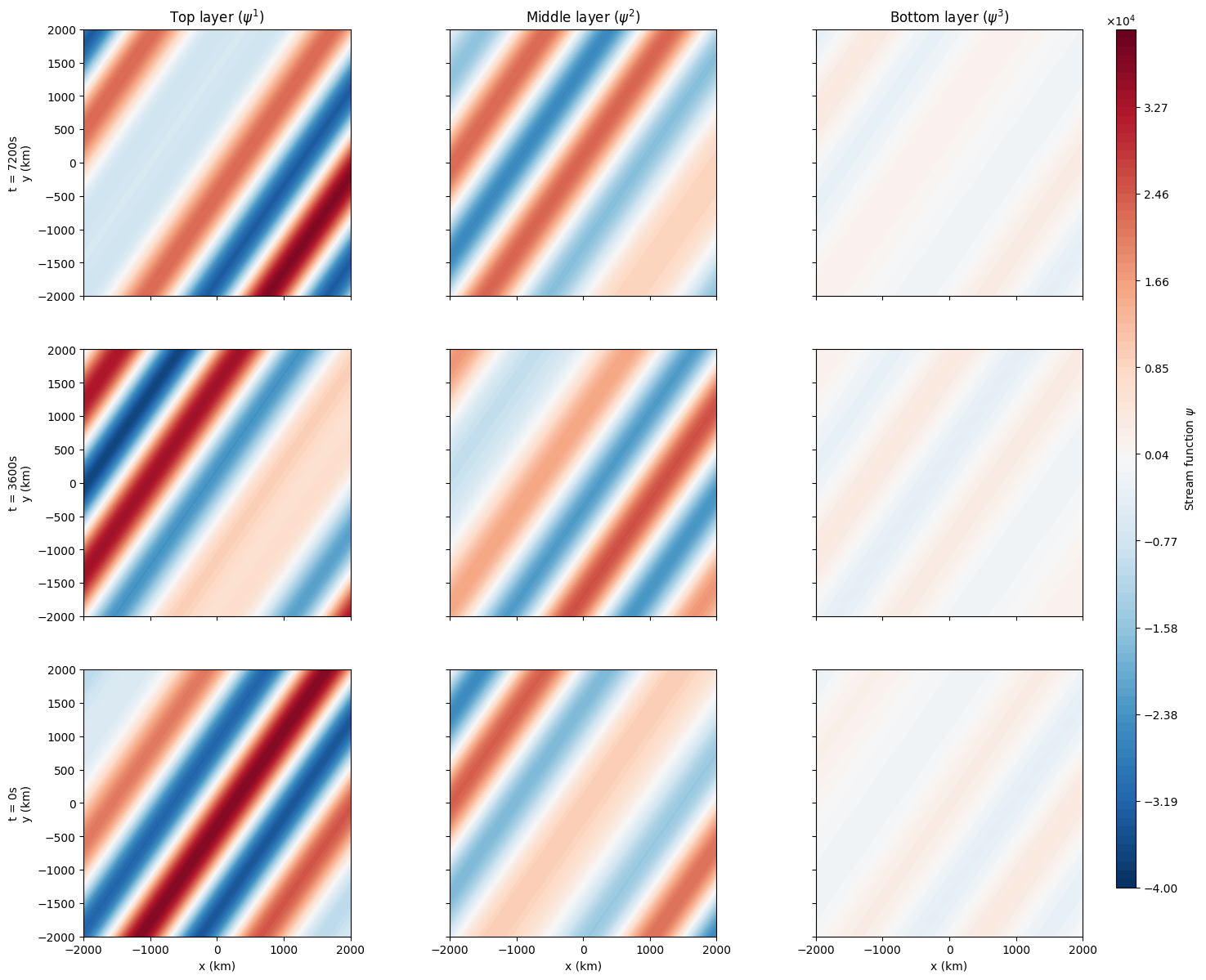}
\caption{Traveling Rossby wave~\eqref{eq:SolutionLinBBMForPlotting} in three layer model
at the time snapshots $t=0$~s, $t=3600$~s and $t=7200$~s from the last row to the first row.}
\label{fig:Red112TravelingRossbyWaves}
\end{figure}

\noprint{
{\it Illustration}
$z_1=t,\quad z_2=x-\chi y$

\begin{gather*}
\solutionRedEq
\tilde v^i(z_1,z_2)=C\cos\big(\delta_it+r(t-\alpha_ix+\alpha_i\chi y)+\theta\big),
\\
\solution
\psi=\sum_{i=1}^{m-1}\tilde v^i(z_1,z_2)e_i+\hat v^m(\hat z_1,\hat z_2)\bar 1,
\end{gather*}

Values of parameters:
\begin{gather*}
\chi=0.7\quad\mbox{tilting parameter}
\\
\alpha_1=1.2\cdot10^{-4},\quad
\delta_1=4\cdot10^{-6},\quad
\Delta=-4.34\cdot10^{-30},\\
\qquad
r_1=3.34\cdot10^{-2},\quad
r_2=(1.34+29.35{\rm i})\cdot10^{-6},\quad
r_3=(1.34-29.35{\rm i})\cdot10^{-6}
\\
\alpha_2=1.5\cdot10^{-4},\quad
\delta_2=3\cdot10^{-6},\quad
\Delta=-3.25\cdot10^{-31},\\
\qquad
r_1=2.00\cdot10^{-2},\quad
r_2=(1.79+14.27{\rm i})\cdot10^{-6},\quad
r_3=(1.79-14.27{\rm i})\cdot10^{-6}
\\
\alpha_3=1.4\cdot10^{-4},\quad
\delta_3=3.5\cdot10^{-6},\quad
\Delta=6.96\cdot10^{-33},\\
\qquad
r_1=2.50\cdot10^{-2},\quad
r_2=3.06\cdot10^{-6},\quad
r_3=0
\end{gather*}
\begin{gather}
\solution
\psi=\sum_{i=1}^{m-1}\tilde v^i(z_1,z_2)e_i+\left(\hat v^m(\hat z_1,\hat z_2)
-\frac{\chi_t}2 y^2+\frac{\chi_{tt}}\beta z_2+\frac\varsigma\beta\right)\bar 1,
\end{gather}
}

\subsubsection{Generically coupled case}\label{sec:Reduction12CoupledCase}

Recall that the completely decoupled case, which corresponds to the zero parameter tuple~$c$,
represents the most singular case in the group classification of the class~$\mathcal R$,
see Lemma~\ref{lem:Reduction12Hmax}.
The opposite, generic case, when the components of~$c$ are pairwise distinct,
is described in Lemma~\ref{lem:Reduction12Hmin}, which states
that then $\mathfrak H=\langle\mathsf E\rangle$
and thus the algebra $\mathfrak g^{\rm gen}_\varrho$ coincides with the span
\[
\langle\p_{z_2},\ v^j\p_{v^j},\ \zeta^j(p,q)\p_{v^j}\rangle,
\]
where the tuple $(\zeta^1,\dots,\zeta^m)$ runs through the solution set
of the corresponding system~$\mathcal R_\varrho$ of the form~\eqref{eq:RedSystSubalgS12MatrixHomog}.
According to Remark~\ref{rem:MinMaxDimOfH},
for further consideration with an arbitrary $c\notin\langle\bar1\rangle$,
we use, instead of the entire spaces~$\mathfrak H$ depending on~$(c,\mathsf F)$,
the common subspace~$\langle\mathsf E\rangle$ of all such spaces.

\paragraph{General case.}
A complete list of inequivalent one-dimensional subalgebras of~$\mathfrak g^{\rm max}_{\varrho}$
that are appropriate for Lie reduction of the system~\eqref{eq:RedSystSubalgS12MatrixHomog}
contains the subalgebras $\langle\p_{z_2}+\mu v^j\p_{v^j}\rangle$, $\mu\in\mathbb R$.
An ansatz associated with such a subalgebra reads $v={\rm e}^{\mu z_2}\varphi(\omega)$ with $\omega=z_1$.
Substituting this ansatz into the system~\eqref{eq:RedSystSubalgS12MatrixHomog},
expanding~$\mathsf B$ according its definition,
$\mathsf B:=\mathop{\rm diag}\mathsf Fc-\mathsf C\mathsf F+\beta\mathsf E$,
and rearranging summands, we derive the reduced system
\begin{gather}\label{eq:Reduction12CoupledCase}
(\p_\omega-\mu\mathsf C)\tilde{\mathsf F}\varphi+\mu\tilde{\mathsf B}\varphi=0,
\end{gather}
where
$\mathsf C:=\mathop{\rm diag} c$,
$\tilde{\mathsf F}:=\mathsf F+\mu^2(1+\chi^2)\mathsf E$ and
$\tilde{\mathsf B}:=\mathop{\rm diag}\mathsf Fc+\beta\mathsf E$.
The modified matrix~$\tilde{\mathsf F}$ commutes with~$\mathsf F$, and thus it is diagonalizable,
its eigenvectors and corresponding eigenvalues are $e_i$ and $\lambda_i+\mu^2(1+\chi^2)$, respectively,
the transition matrix to this eigenbasis coincides with $\mathsf P$,
see Section~\ref{sec:PropertiesOfMatrixF} for the notation.
Moreover, if $\chi_\omega\ne0$, the matrix $\tilde{\mathsf F}$ is invertible as a matrix-valued function.
Substituting $\tilde\varphi=\tilde{\mathsf F}\varphi$ into~\eqref{eq:Reduction12CoupledCase}
we derive a variable-coefficient linear homogeneous system of first-order ordinary differential equations
in the canonical form,
\begin{gather}\label{eq:Reduction12CoupledCaseGeneralFinvert}
\tilde\varphi_\omega-\mu(\mathsf C-\tilde{\mathsf B}\tilde{\mathsf F}^{-1})\tilde\varphi=0.
\end{gather}
The corresponding solutions of the system~\eqref{eq:mLaysMod} can be represented in the form
\[
\solution
\psi={\rm e}^{\mu(x-\chi y)}\tilde{\mathsf F}^{-1}\tilde\varphi
-\frac12\chi_ty^2\bar1+yc+\frac{\chi_{tt}}{\beta}(x-\chi y)\bar 1-\frac{\chi_t}\beta c,
\]
where $\tilde\varphi=\tilde\varphi(\omega)$ with $\omega=t$ is an arbitrary solution
of~\eqref{eq:Reduction12CoupledCaseGeneralFinvert},
$c$ is an arbitrary constant tuple that does not belong to $\langle\bar1\rangle$,
$\mu$ is an arbitrary constant,
$\chi$ is an arbitrary smooth function of~$t$, and
$\tilde{\mathsf F}:=\mathsf F+\mu^2(1+\chi^2)\mathsf E$.
Complexifying these solutions, setting the parameter~$\mu$ to be imaginary
and splitting the result into the real and the imaginary parts,
we construct solutions that are bounded and periodic with respect to $x-\chi y$.

\paragraph{Scale-invariant case.}
When $\chi=\pm\sqrt{\alpha z_1^2-1}$,
a complete list of inequivalent subalgebras of~$\mathfrak g^{\rm max}_\varrho$
that are appropriate for Lie reduction of the system~\eqref{eq:RedSystSubalgS12MatrixHomog}
includes one more subalgebra family $\langle z_1\p_{z_1}+z_2\p_{z_2}+\mu v^j\p_{v^j}\rangle$, $\mu\in\mathbb R$.
For a fixed~$\mu$, an associated ansatz is $v=|z_1|^\mu\varphi(\omega)$ with $\omega:=z_2/z_1$,
which leads to the reduced system
\begin{gather}\label{eq:Reduction12CoupledCaseScaleInv}
\alpha(\omega^3\mathsf E+\mathsf C)\varphi_{\omega\omega\omega}
-\alpha(\mu-2)\omega^2\varphi_{\omega\omega}
+(\omega\mathsf F-\mathsf B)\varphi_\omega
-\mu\mathsf F\varphi=0.
\end{gather}
The integration of this system does not look feasible
since even the analogous system~\eqref{eq:RedSystSubalgS12b} in the decoupled case
was solved in terms of Bessel functions and antiderivatives of Whittaker functions.
For the sake completeness, we write a representation
for the corresponding solutions of the system~\eqref{eq:mLaysMod}:
\[
\solution
\psi=|t|^\mu\varphi-\frac{\varepsilon\alpha t}{\sqrt{\alpha t^2-1}}(y^2\bar1+2\beta^{-1}c)+yc
-\frac{\varepsilon\alpha t\omega}{\beta(\alpha t^2-1)^{3/2}}\bar 1,
\]
where $\varepsilon=\pm1$,
$\varphi=\varphi(\omega)$ with $\omega=(x+\varepsilon y\sqrt{\alpha t^2-1})/t$ is an arbitrary solution of~\eqref{eq:Reduction12CoupledCaseScaleInv},
$c$ is an arbitrary constant $m$-tuple with $c\notin\langle\bar1\rangle$,
$\alpha$ and $\mu$ are arbitrary constants with $\alpha>0$.

\paragraph{Shift-invariant case.}
When $\chi_1=0$,
a complete list of inequivalent one-dimensional subalgebras
of the maximal Lie invariance algebra $\mathfrak g^{\rm max}_{\varrho}$
of the system~\eqref{eq:RedSystSubalgS12MatrixHomog}
that are appropriate for Lie reductions includes at least the subalgebra families,
$\langle \p_{z_1}+\nu\p_{z_2}+\mu v^j\p_{v^j}\rangle$
and $\langle\p_{z_2}+\mu v^j\p_{v^j}\rangle$,
where $\nu,\mu\in\mathbb R$.

Using a subalgebra from the first family,
we construct the ansatz $v={\rm e}^{\mu z_1}\varphi(\omega)$ with $\omega=z_2-\nu z_1$,
which reduces the system~\eqref{eq:RedSystSubalgS12MatrixHomog} to
a constant-coefficient homogeneous linear system of ordinary differential equations that are at most of order three,
\begin{gather}\label{eq:Reduction12CoupledCaseShiftInvA}
(1+\chi^2)\hat{\mathsf C}\varphi_{\omega\omega\omega}
-\mu(1+\chi^2)\varphi_{\omega\omega}
-\hat{\mathsf B}\varphi_\omega
-\mu\mathsf F\varphi=0,
\end{gather}
where
$\hat{\mathsf C}:=\mathsf C+\nu\mathsf E$ and
$\hat{\mathsf B}:=\mathop{\rm diag}\mathsf Fc-\hat{\mathsf C}\mathsf F+\beta\mathsf E$.
The general solution of the reduced system~\eqref{eq:Reduction12CoupledCaseShiftInvA}
can be obtained via solving the eigenvalue problem for a certain $3m$ by $3m$ matrix
constructed from the matrices $(1+\chi^2)\hat{\mathsf C}$, $\mu(1+\chi^2)\mathsf E$,
$\hat{\mathsf B}$ and $\mu\mathsf F$.
While finding this solution for specific values of parameters causes no issues,
writing it down in a closed form in terms of the involved parameters remains a challenge
even for a low number of layers~$m$
since it requires further investigation of properties of the above collection of matrices.
We do not consider this problem here, but we will study an analogous inhomogeneous system
for the specific case $\mu=\nu=0$ in Section~\ref{sec:ReductionS21}.
Nevertheless, the corresponding solution of~\eqref{eq:mLaysMod} is given~by
\begin{gather}\label{eq:Reduction12CoupledCaseShiftInvASol}
\solution
\psi={\rm e}^{\mu t}\varphi+yc,
\end{gather}
where $\varphi=\varphi(\omega)$ with $\omega=x-\chi y-\nu t$ is an arbitrary solution of~\eqref{eq:Reduction12CoupledCaseShiftInvA},
$c$ is an arbitrary constant $m$-tuple with $c\notin\langle\bar1\rangle$ and
$\mu$, $\nu$ and~$\chi$ are arbitrary constants.
We can complexify this solution, assuming that~$\mu$ is an imaginary number.
If in this framework the eigenvalue problem associated with~\eqref{eq:Reduction12CoupledCaseShiftInvA}
has imaginary solutions,
the system~\eqref{eq:Reduction12CoupledCaseShiftInvA} in its turn
has bounded periodic solutions, which results in simple-wave velocity fields.

The reduction using any subalgebra from the family $\langle\p_{z_2}+\mu v^j\p_{v^j}\rangle$
leads to the system~\eqref{eq:Reduction12CoupledCase} with constant~$\chi$,
which can be easily integrated in terms of matrix exponentials.
The precise form of the general solution of~\eqref{eq:Reduction12CoupledCaseGeneralFinvert}
depends on whether the modified matrix $\tilde{\mathsf F}=\mathsf F+\mu^2(1+\chi^2)\mathsf E$
is invertible or not.

It is invertible if and only if $\lambda_i\ne-\mu^2(1+\chi^2)$ for all $i=1,\dots,m$.
Recall that $\lambda_i$ are the eigenvalues of the matrix~$\mathsf F$,
see the end of Section~\ref{sec:VectorFormOfSyst}.
In this case, the system~\eqref{eq:Reduction12CoupledCase} is equivalent to
the system~\eqref{eq:Reduction12CoupledCaseGeneralFinvert},
which is a linear homogeneous system of first-order ordinary differential equations with constant coefficients.
Hence the corresponding solution of the original system~\eqref{eq:mLaysMod} is given by
\begin{gather}\label{eq:Reduction12CoupledCaseSolutionFromZ2Invertible}
\solution
\psi={\rm e}^{\mu(x-\chi y)}\tilde{\mathsf F}^{-1}
\exp\big(\mu t(\mathsf C-\tilde{\mathsf B}\tilde{\mathsf F}^{-1})\big)A
+yc,
\end{gather}
where $\tilde{\mathsf F}:=\mathsf F+\mu^2(1+\chi^2)\mathsf E$,
$\tilde{\mathsf B}:=\mathop{\rm diag}\mathsf Fc+\beta\mathsf E$,
$\mathsf C:=\mathop{\rm diag}c$,
$c$ and $A$ are $m$-tuples of arbitrary constants with $c\notin\langle 1\rangle$,
and $\mu$ and $\chi$ are arbitrary constants.

Since the eigenvalues of $\tilde{\mathsf F}$ are pairwise different,
the matrix $\tilde{\mathsf F}$ is degenerate if and only if
there exists~$i$ such that the eigenvector $e_i$ of~$\tilde{\mathsf F}$ spans the kernel of $\tilde{\mathsf F}$,
$\ker\tilde{\mathsf F}=\langle e_i\rangle$.
This allows us to write down the corresponding solution of the original system~\eqref{eq:mLaysMod} as follows:
\begin{gather*}
\solution
\psi={\rm e}^{\mu(x-\chi y)}\tilde{\mathsf F}^+
\exp\big(\mu t(\mathsf C-\tilde{\mathsf B}\tilde{\mathsf F}^+)\big)A
+g(t)e_i+yc,
\end{gather*}
where $\tilde{\mathsf F}:=\mathsf F+\mu^2(1+\chi^2)\mathsf E$,
$\tilde{\mathsf B}:=\mathop{\rm diag}\mathsf Fc+\beta\mathsf E$,
$\mathsf C:=\mathop{\rm diag}c$,
$c$ and $A$ are tuples of arbitrary constants with $c\notin\langle\bar1\rangle$,
$\mu$ and $\chi$ are arbitrary constants,
$g(t)$ is an arbitrary smooth function of~$t$ if $\tilde{\mathsf B}e_i=0$ and
$g=0$ otherwise,
and the Moore--Penrose inverse~$\tilde{\mathsf F}^+$ of $\tilde{\mathsf F}$
is defined in the same way as that for~$\mathsf F$ in~\eqref{eq:MNInverse}.

Note that in the derived solutions, we can complexify the parameters $\mu$ and $A$ to obtain wave-like solutions.
When $\mu={\rm i}\nu$ with $\nu\in\mathbb R_{\ne0}$,
the modified matrix $\tilde{\mathsf F}:=\mathsf F-\nu^2(1+\chi^2)\mathsf E$
is invertible, and thus the only relevant solution family for the complexification is~\eqref{eq:Reduction12CoupledCaseSolutionFromZ2Invertible}.
This yields bounded periodic solutions of~\eqref{eq:mLaysMod}
if at least some eigenvalues of the matrix $\mathsf C-\tilde{\mathsf B}\tilde{\mathsf F}^{-1}$ are real.

\subsection{Subalgebra family 1.3}\label{sec:Reduction13}

In many aspects, the consideration of Lie reductions with respect to the subalgebras  $\mathfrak s_{1.3}^{\chi c\kappa}$
is analogous to that for the subalgebras $\mathfrak s_{1.2}^{\chi c}$.
This includes the group classification of reduced systems,
partitioning into cases for two-step reductions and their naming.
In the spirit of Remark~\ref{rem:MinMaxDimOfH},
special attention is paid to the {\it completely decoupled} and {\it generically coupled} cases of the reduced systems
and the further splitting of the latter one into the \emph{general}, the \emph{scale-invariant} and the \emph{shift-invariant} cases
when relevant.

\subsubsection{Reduced system}

The one-dimensional subalgebra
$\mathfrak s_{1.3}^{\chi c\kappa}:=\langle \mathcal P^x(\chi)+c_k\mathcal J^k+\mathcal Z(\kappa)\rangle$
of~$\mathfrak g$,
where $c:=(c_1,\dots,c_m)^{\mathsf T}$ is an arbitrary element of~$\mathop{\rm im}\mathsf F$
and $\chi$ and~$\kappa$ are arbitrary smooth functions of~$t$,
satisfies the transversality condition and is thus appropriate for Lie reduction
if and only if $\chi\ne0$.
Using this subalgebra, we construct a suitable Lie ansatz,
\[
\psi^i=v^i-\frac{\chi_t(t)y-\kappa(t)-c_i}{\chi(t)}x-\frac\beta6 y^3\quad\mbox{with}\quad
z_1=t,\quad
z_2=\chi y-\int\kappa(t)\,{\rm d}t,
\]
to reduce the system~\eqref{eq:mLaysMod} to the linear system
\begin{gather*}
c_i\chi^2v^i_{222}+(\chi^2 v^i_{22})_1
\\
\quad{}
+f_{i,i-1}\left(c_iv^{i-1}_2-c_iv^i_2+v^{i-1}_1-v^i_1-c_{i-1}v^i_2+c_iv^i_2+\frac\beta{2\chi^3}(z_2+\smallint\kappa)^2(c_{i-1}-c_i)\right)
\\
\quad{}
-f_{i,i+1}\left(c_iv^i_2-c_iv^{i+1}_2+v^i_1-v^{i+1}_1-c_iv^i_2+c_{i+1}v^i_2+\frac\beta{2\chi^3}(z_2+\smallint\kappa)^2(c_i-c_{i+1})\right)=0,
\end{gather*}
where $\smallint\kappa:=\int\kappa(t)\,{\rm d}t$ denotes an antiderivative of~$\kappa$ with respect to $z_1:=t$.
In the matrix notation, the ansatz and the reduced system take the form
\begin{gather}
\nonumber
\psi=v-\left(\frac{\chi_t(t)y-\kappa(t)}{\chi(t)}x-\frac\beta6 y^3\right)\bar 1+\frac x{\chi(t)}c
\quad\mbox{with}\quad z_1=t,\quad z_2=\chi y-\int\kappa(t)\,{\rm d}t,
\\
\label{eq:RedSys13}
\chi^2\mathsf Cv_{222}+(\chi^2v_{22})_1+\mathsf Bv_2+\mathsf Fv_1
+\frac{\beta}{2\chi^3}(z_2+\smallint\kappa)^2\mathsf Fc=0,
\end{gather}
respectively.
As above, $\mathsf C:=\mathop{\rm diag}c$,
and we also denote $\mathsf B:=\mathsf C\mathsf F-\mathop{\rm diag}\mathsf Fc$.
The facts that
$(\mathop{\rm diag}\mathsf Fc)\bar 1=\mathsf Fc$,
$\mathop{\rm im}\mathsf F\perp_{\mathsf W}^{}\langle\bar 1\rangle$,
$\mathop{\rm im}\mathsf B\perp_{\mathsf W}^{}\langle\bar 1\rangle$ and
$\dim\mathop{\rm im}\mathsf F=m-1$
implies $\mathop{\rm im}\mathsf B\subseteq\mathop{\rm im}\mathsf F$.
We also have $c\in\mathop{\rm im}\mathsf F$.
This is why a particular solution of the system~\eqref{eq:RedSys13} takes the form
\noprint{
\begin{gather*}
\begin{split}
\hat v=&-\beta\left(
\frac{z_2^2}2\int\frac{{\rm d}z_1}{\chi^3}+z_2\int\frac{\smallint\kappa}{\chi^3}{\rm d}z_1 +\int\frac{(\smallint\kappa)^2}{2\chi^3}{\rm d}z_1\right)c
\\&{}
+\beta\chi^2\int\frac{{\rm d}z_1}{\chi^3}\mathsf F^+c
+\beta\left(z_2\iint\frac{{\rm d}z_1}{\chi^3}{\rm d}z_1
+\iint\frac{\smallint\kappa}{\chi^3}{\rm d}z_1{\rm d}z_1\right)\mathsf F^+\mathsf Bc
\\&{}
-\beta\iiint\frac{{\rm d}z_1}{\chi^3}{\rm d}z_1{\rm d}z_1(\mathsf F^+\mathsf B)^2c
\end{split}
\end{gather*}
}
\begin{gather*}
\hat v=-\beta\left(\tfrac12z_2^2\mathcal X_{11}+\mathcal K_1z_2 +\mathcal N\right)c
+\beta\chi^2\mathcal X_{11}\mathsf F^+c
+\beta\left(z_2\mathcal X_1+\mathcal K\right)\mathsf F^+\mathsf Bc
-\beta\mathcal X(\mathsf F^+\mathsf B)^2c,
\end{gather*}
where $\mathcal X$, $\mathcal K$ and $\mathcal N$ denote
third, the second and the first antiderivatives of
$\chi^{-3}$, $\chi^{-3}\smallint\kappa$ and $\chi^{-3}(\smallint\kappa)^2$
with respect to~$z_1$, respectively.
Recall that the Moore--Penrose inverse~$\mathsf F^+$ of~$\mathsf F$ is given by~\eqref{eq:MNInverse}.
By the transformation of the dependent variables $v-\hat v\to v$,
we reduce the system~\eqref{eq:RedSys13} to its homogeneous counterpart
\begin{gather}\label{eq:RedSys13Homogen}
\chi^2\mathsf Cv_{222}+(\chi^2v_{22})_1+\mathsf Bv_2+\mathsf Fv_1=0,
\end{gather}
Therefore, we have the following representation
for the $\mathfrak s_{1.3}^{\chi c\kappa}$-invariant solutions of the system~\eqref{eq:mLaysMod}:
\begin{gather}\label{eq:Ansatz13Mod}
\begin{split}\solution
\psi=&{}v-\left(\frac{\chi_t(t)y-\kappa(t)}{\chi(t)}x-\frac\beta6 y^3\right)\bar 1+\frac x{\chi(t)}c
-\beta\left(\tfrac12z_2^2\mathcal X_{11}+\mathcal K_1z_2 +\mathcal N\right)c
\\[.5ex]&{}
+\beta\chi^2\mathcal X_{11}\mathsf F^+c
+\beta\left(z_2\mathcal X_1+\mathcal K\right)\mathsf F^+\mathsf Bc
-\beta\mathcal X(\mathsf F^+\mathsf B)^2c,
\end{split}	
\end{gather}
where $v=v(z_1,z_2)$ is an arbitrary solution of the homogeneous linear system
of partial differential equations~\eqref{eq:RedSys13Homogen},
which can be considered as a reduced system for $\mathfrak s_{1.3}^{\chi c\kappa}$-invariant solutions
instead of~\eqref{eq:RedSys13},
$z_1:=t$, $z_2=\chi y-\int\kappa(t)\,{\rm d}t$,
$\chi$ and~$\kappa$ are arbitrary functions of~$t$,
$\mathcal X$, $\mathcal K$ and $\mathcal N$ denote
third, the second and the first antiderivatives of
$\chi^{-3}$, $\chi^{-3}\smallint\kappa$ and $\chi^{-3}(\smallint\kappa)^2$
with respect to~$t$, respectively,
$c:=(c_1,\dots,c_m)^{\mathsf T}$ is a tuple of arbitrary constants with $c\perp_{\mathsf W}^{}\bar 1$,
$\mathsf C:=\mathop{\rm diag}c$ and $\mathsf B:=\mathsf C\mathsf F-\mathop{\rm diag}\mathsf Fc$.

\subsubsection{Group classification}\label{sec:Reduction13GroupClassification}

The study of induced and hidden symmetries of the system~\eqref{eq:mLaysMod}
associated with the reduction with respect to the subalgebra~$\mathfrak s_{1.3}^{\chi c\kappa}$
is analogous to that carried out in Section~\ref{sec:Reduction12}.
We begin with finding the induced symmetries of the reduced system~\eqref{eq:RedSys13Homogen}.
Note that the subalgebra $\mathfrak g_1=\langle\mathcal P^x(\chi),\mathcal J^1,\dots,\mathcal J^m,\mathcal Z(\kappa)\rangle\subset\mathfrak g$
is an abelian Lie algebra.
The normalizer ${\rm N}_{\mathfrak g}(\mathfrak s_{1.3}^{\chi c\kappa})$ of $\mathfrak s_{1.3}^{\chi c\kappa}$
coincides with the subalgebra~$\mathfrak g_1$ for generic values of $\chi$, $c$ and $\kappa$,
providing $\chi_t\ne0$ and $\kappa_t\ne0$.
We have the only following extensions of the generic case:
\begin{enumerate}\itemsep=.5ex
\item
${\rm N}_{\mathfrak g}(\mathfrak s_{1.3}^{\chi c\kappa})=\langle\mathcal P^t\rangle+\mathfrak g_1$
if and only if $\chi_t=\alpha\chi$, $\kappa_t=\alpha\kappa$ for some $\alpha\ne0$ and $c=0$,
	
\item
${\rm N}_{\mathfrak g}(\mathfrak s_{1.3}^{\chi c\kappa})
=\langle\mathcal P^y\rangle+\mathfrak g_1$
if and only if $\chi_t=0$ and $\kappa_t\ne0$,
	
\item
${\rm N}_{\mathfrak g}(\mathfrak s_{1.3}^{\chi c\kappa})=\mathfrak g$
if and only if $\chi_t=0$ and $\kappa_t=0$.	
\end{enumerate}
The algebras of induced symmetries corresponding to the above cases are the following:
\begin{enumerate}\itemsep=.5ex\setcounter{enumi}{-1}
\item
$\langle
z_2\rho(z_1)(\p_{v^1}+\cdots+\p_{v^m}),\,
\p_{v^1},\,\dots,\,\p_{v^m},\,
\kappa(z_1)(\p_{v^1}+\cdots+\p_{v^m})
\rangle$,

\item
$\langle
\p_{z_1},\,
z_2\rho(z_1)(\p_{v^1}+\cdots+\p_{v^m}),\,
\p_{v^1},\,\dots,\,\p_{v^m},\,
\kappa(z_1)(\p_{v^1}+\cdots+\p_{v^m})
\rangle$,

\item
$\langle
\p_{z_2},\,z_2\rho(z_1)(\p_{v^1}+\cdots+\p_{v^m}),\,
\p_{v^1},\,\dots,\,\p_{v^m},\,
\kappa(z_1)(\p_{v^1}+\cdots+\p_{v^m})
\rangle$,
	
\item
$\langle
\p_{z_1},\,\p_{z_2},\,z_2\rho(z_1)(\p_{v^1}+\cdots+\p_{v^m}),\,
\p_{v^1},\,\dots,\,\p_{v^m},\,
\kappa(z_1)(\p_{v^1}+\cdots+\p_{v^m})
\rangle$.
\end{enumerate}
Here $\rho=\rho(z_1)$ and $\kappa=\kappa(z_1)$ run through the set of arbitrary smooth functions of~$z_1$.

The homogeneous reduced system~\eqref{eq:RedSys13Homogen} is of the same form as~\eqref{eq:RedSystSubalgS12MatrixHomog},
which allows us to straightforwardly apply the results of Section~\ref{sec:Reduction12GroupClassification}
for computing the maximal Lie invariance algebras of~\eqref{eq:RedSys13Homogen}
depending on values of involved parameters.
For this purpose, we regard the set of systems of the form~\eqref{eq:RedSys13Homogen}
as a class of systems of differential equations
parameterized by an arbitrary function $\chi=\chi(z_1)$ and the components of~$\theta=(\mathsf F,\beta)$ and~$c$.
The most general form of a Lie symmetry vector field of~\eqref{eq:RedSys13Homogen} is
$\tau\p_{z_1}+\xi \p_{z_2}+\eta^j\p_{v^j}$,
where the components~$\tau$, $\xi$ and $\eta^j$ are smooth functions depending on $(z_1,z_2,v)$.
The maximal Lie invariance algebras of any system from the class~\eqref{eq:RedSys13Homogen}
with the general value of the arbitrary-element tuple is infinite-dimensional and spanned by the vector fields
\begin{gather*}
\p_{z_2},\quad
(\mathsf H v)^j\p_{ v^j},\quad
\zeta^j(z_1,z_2)\p_{v^j}.
\end{gather*}
Here the tuple $(\zeta^1,\dots,\zeta^m)$
runs through the solution set of the system~\eqref{eq:RedSys13Homogen},
and $\mathsf H$ runs through a basis of the space $\mathfrak H$ of (constant) matrices
commuting with the matrices~$\mathsf F$, $\mathsf C$ and $\mathsf C\mathsf F-\mathop{\rm diag}\mathsf Fc$,
where the last matrix can be replaced by $\mathop{\rm diag}\mathsf Fc$,
\[
[\mathsf H,\mathsf F]=[\mathsf H,\mathsf C]=[\mathsf H,\mathop{\rm diag}\mathsf Fc]=0.
\]
Some general properties of the space $\mathfrak H$ have been described in Lemmas~\ref{lem:Reduction12Hmax} and~\ref{lem:Reduction12Hmin},
and the exhaustive description of $\mathfrak H$ for $m=2$, $3$ and $4$ has been given in Section~\ref{sec:Reduction12GroupClassification}.

For arbitrary fixed values of the arbitrary elements~$\theta=(\mathsf F,\beta)$ and~$c$,
the corresponding general Lie invariance algebra can be further extended
for specific values of the parameter function~$\chi$.
There are two inequivalent cases of such extensions,
\begin{align*}
\chi_1=0\colon&\quad
\langle\p_{z_2},
(\mathsf Hv)^j\p_{v^j},
\zeta^j(z_1,z_2)\p_{v^j}\rangle
+\langle\p_{z_1}\rangle,
\\[.5ex]
\chi=\alpha z_1\colon&\quad
\langle\p_{z_2},
(\mathsf Hv)^j\p_{v^j},
\zeta^j(z_1,z_2)\p_{v^j}\rangle
+\langle z_1\p_{z_1}+z_2\p_{z_2}\rangle,
\end{align*}
where $\alpha$ is an arbitrary nonzero constant.
The extension~$\p_{z_1}$ is induced by the Lie-symmetry vector field~$\p_t$ of the original system~$\mathcal M_\theta$,
whereas the extension~$z_1\p_{z_1}+z_1\p_{z_2}$ is a genuine hidden Lie symmetry of~$\mathcal M_\theta$.

\subsubsection{Completely decoupled case}\label{sec:Reduction13DecoupledCase}

If $c=0$, the system~\eqref{eq:RedSys13Homogen} takes the form
$(\chi^2v_{22}+\mathsf Fv)_1=0$,
which trivially integrates with respect to~$z_1$~to
\begin{gather}\label{eq:SemidecoupledReduction13}
\chi^2v_{22}+\mathsf Fv=\varphi,
\end{gather}
where $\varphi:=(\varphi^1(z_2),\dots,\varphi^m(z_2))^{\mathsf T}$ is an $m$-tuple of arbitrary smooth functions of $z_2$.
The system~\eqref{eq:SemidecoupledReduction13} is an inhomogeneous linear system
of second-order ordinary differential equations with the independent variable $z_2$,
where $z_1$ plays the role of a parameter.
It can be decoupled via diagonalizing the matrix~$\mathsf F$
and then integrated by the method of variation of parameters.
More specifically, we first make the change of the dependent variables $v:=\mathsf P\tilde v$
and denote $\tilde\varphi:=\mathsf P^{-1}\varphi$,
where $\mathsf P$ is the transition matrix to the eigenbasis of~$\mathsf F$,
see Section~\ref{sec:PropertiesOfMatrixF}.
Then the system~\eqref{eq:SemidecoupledReduction13} reduces to the decoupled system
\begin{gather}\label{eq:SemidecoupledReduction13JNF}
\chi^2\tilde v_{22}+\lambda_i\tilde v=\tilde\varphi,\quad
i=1,\dots,m.
\end{gather}
Recall that $\lambda_i$ are the eigenvalues of $\mathsf F$ with $\lambda_1<\dots<\lambda_m=0$.
For $i\ne m$ and $i=m$, the general solution of the equations~\eqref{eq:SemidecoupledReduction13JNF}
is given by
\begin{gather*}
\solutionRedEq
\tilde v^i=
\frac{{\rm e}^{\nu^i z_2}}{2\chi^2\nu^i}\int{\rm e}^{-\nu^i z_2}\tilde\varphi^i(z_2){\rm d}z_2
-\frac{{\rm e}^{-\nu^i z_2}}{2\chi^2\nu^i}\int{\rm e}^{\nu^i z_2}\tilde\varphi^i(z_2){\rm d}z_2
+h^{1i}(z_1){\rm e}^{\nu^i z_2}+h^{2i}(z_1){\rm e}^{-\nu^i z_2},
\\[.5ex]
\tilde v^m=\chi^{-2}g(z_2)+h^{1m}(z_1)z_2+h^{2m}(z_1),
\end{gather*}
respectively,
where $\nu^i:=\sqrt{-\lambda_i}/\chi$, $h^{1i}$ and $h^{2i}$ are arbitrary smooth functions of $z_1$,
and $g=g(z_2)$ is a second antiderivative of~$\tilde\varphi^m$ with respect to~$z_2$.
Pulling this solution back with respect to the ansatz,
we derive the solution of the original system~\eqref{eq:mLaysMod},
\begin{gather*}
\solution
\psi=\sum_{i=1}^{m-1}\left(
\frac{{\rm e}^{\nu^i z_2}}{2\chi^2\nu^i}\int{\rm e}^{-\nu^i z_2}\tilde\varphi^i\,{\rm d}z_2
-\frac{{\rm e}^{-\nu^i z_2}}{2\chi^2\nu^i}\int{\rm e}^{\nu^i z_2}\tilde\varphi^i\,{\rm d}z_2
+h^{1i}{\rm e}^{\nu^i z_2}+h^{2i}{\rm e}^{-\nu^i z_2}
\right)e_i
\\
\qquad
+\Big(\chi^{-2}g+h^{1m}z_2+h^{2m}
-\frac{\chi_t y-\kappa}\chi x-\frac\beta6 y^3\Big)\bar 1,
\end{gather*}
where $z_2=\chi y-\smallint\kappa{\rm d}t$,
$e_i$ is an eigenvector of the matrix~$\mathsf F$ corresponding to its eigenvalue~$\lambda_i$,
$\lambda_m=0$, $e_m=\bar1$,
$\nu^i:=\sqrt{-\lambda_i}/\chi(t)$,
$\chi=\chi(t)$, $\kappa=\kappa(t)$, $h^{1i}=h^{1i}(t)$, $h^{2i}=h^{2i}(t)$, $\tilde\varphi^i=\tilde\varphi^i(z_2)$
and $g=g(z_2)$ are arbitrary smooth function of their arguments with $\chi\ne0$.
Note that we can gauge $h^{2m}(t)$ to zero modulo the $G$-equivalence.

In view of Corollary~\ref{cor:EigenvaluesAreNegaive},
which states that nonzero eigenvalues $\lambda_i$ are real negative,
and the positivity of the parameter $\beta$,
only particular simple solutions from the above family are physically relevant for the entire plane.

\subsubsection{Generically coupled case}\label{sec:Reduction13CoupledCase}

Analogously to Section~\ref{sec:Reduction12CoupledCase},
we study the case of generic coupling of the reduced system~\eqref{eq:RedSys13Homogen}
with values the subalgebra parameter $c$ from $\im\mathsf F\setminus\{0\}$.
The consideration is divided into three subcases
according to the presence of Lie-symmetry vector fields with nonzero $z_1$-components
and their kind, that is, the general case (with no such fields) and the scale- and the shift-invariant cases.

\paragraph{General case.}
A complete list of inequivalent one-dimensional subalgebras of the
maximal Lie invariance algebra of the system~\eqref{eq:RedSys13Homogen}
that are appropriate for Lie reductions necessarily includes the subalgebra family
$\langle\p_{z_2}+\mu v^j\p_{v^j}\rangle$,
where $\mu\in\mathbb R$.
A suitable ansatz corresponding to a subalgebra from this family with fixed value of~$\mu$
is $v={\rm e}^{\mu z_2}\varphi(\omega)$ with $\omega=z_1$.
It reduces the system~\eqref{eq:RedSys13Homogen}~to
\begin{gather}\label{eq:Reduction13CoupledCase}
(\p_\omega+\mu\mathsf C)\tilde{\mathsf F}\varphi-\mu\tilde{\mathsf B}\varphi=0,
\end{gather}
where $\tilde{\mathsf F}:=\mathsf F+\mu^2\chi^2\mathsf E$ and $\tilde{\mathsf B}:=\mathop{\rm diag}\mathsf Fc$.
The modified matrix $\tilde{\mathsf F}$ commutes with $\mathsf F$,
it is diagonalizable, its eigenvalues and eigenvectors are
$\lambda_i+\mu^2\chi^2$ and $e_i$, respectively,
where $\lambda_i$ and $e_i$ are eigenvalues and eigenvectors of the matrix $\mathsf F$,
see Section~\ref{sec:PropertiesOfMatrixF} for details.
In the general case, we can assume $\chi_\omega\ne0$, and then
the matrix $\tilde{\mathsf F}$ is invertible as a matrix-valued function.
Substituting $\varphi=\tilde{\mathsf F}^{-1}\tilde\varphi$ into~\eqref{eq:Reduction13CoupledCase},
we obtain a system of homogeneous linear first-order ordinary differential equations with variable coefficients
in the canonical form
\begin{gather}\label{eq:Reduction13CoupledCaseGeneralFinvert}
\tilde\varphi_\omega+\mu(\mathsf C-\tilde{\mathsf B}\tilde{\mathsf F}^{-1})\tilde\varphi=0.
\end{gather}
The corresponding solutions of the reduced system~\eqref{eq:RedSys13Homogen} can be represented as
\begin{gather}\label{eq:Reduction13CoupledGeneralCaseSolution}
\solutionRedEq
v={\rm e}^{\mu z_2}\tilde{\mathsf F}^{-1}\tilde\varphi,
\end{gather}
where $\tilde\varphi=\tilde\varphi(z_1)$ is an arbitrary solution
of~\eqref{eq:Reduction13CoupledCaseGeneralFinvert},
$c$ is an arbitrary constant $m$-tuple with $c\in\im\mathsf F\setminus\{0\}$,
$\mu$ is an arbitrary constant,
$\chi$ and $\kappa$ are arbitrary smooth functions of~$t$, and
$\tilde{\mathsf F}:=\mathsf F+\mu^2\chi^2\mathsf E$.
Substituting~\eqref{eq:Reduction13CoupledGeneralCaseSolution} into~\eqref{eq:Ansatz13Mod},
we obtain a solution family of~\eqref{eq:mLaysMod}.

\paragraph{Scale-invariant case.}
For $\chi=\alpha z_1$,
a complete list of inequivalent one-dimensional subalgebras
of the maximal Lie invariance algebra of~\eqref{eq:RedSys13Homogen}
in addition includes the subalgebra family $\langle z_1\p_{z_1}+z_2\p_{z_2}+\mu v^j\p_{v^j}\rangle$,
where $\mu\in\mathbb R$.
Using such a subalgebra with fixed~$\mu$, we construct an ansatz $v=|z_1|^\mu\varphi(\omega)$ with $\omega:=z_2/z_1$,
which leads to the reduced system
\begin{gather}\label{eq:RedSys13CoupledScaleInv}
\alpha^2\hat{\mathsf C}\varphi_{\omega\omega\omega}+\mu\alpha^2\varphi_{\omega\omega}+\hat{\mathsf B}\varphi_\omega+\mu\mathsf F\varphi=0,
\end{gather}
where $\hat{\mathsf C}=\mathsf C-\omega\mathsf E$ and $\hat{\mathsf B}:=\hat{\mathsf C}\mathsf F-\mathop{\rm diag}\mathsf Fc$.
The corresponding solutions of the reduced system~\eqref{eq:RedSys13Homogen} can be represented as
\begin{gather}\label{eq:Reduction13CoupledScaleCaseSolution}
\solutionRedEq
v=|z_1|^\mu\varphi,
\end{gather}
where $\varphi=\varphi(z_2/z_1)$ is an arbitrary solution
of~\eqref{eq:RedSys13CoupledScaleInv},
$c$ is an arbitrary constant $m$-tuple with $c\in\im\mathsf F\setminus\{0\}$,
$\mu$ is an arbitrary constant,
$\chi$ and $\kappa$ are arbitrary smooth functions of~$t$.
Substituting~\eqref{eq:Reduction13CoupledScaleCaseSolution} into~\eqref{eq:Ansatz13Mod},
we obtain the associated representation of solutions of~\eqref{eq:mLaysMod}.

\paragraph{Shift-invariant case.}
When $\chi$ is an arbitrary nonzero constant,
a complete list of inequivalent one-dimensional subalgebras of the maximal Lie invariance algebra
of the system~\eqref{eq:RedSys13Homogen}
that are appropriate for Lie reductions contains at least the subalgebra families
$\langle \p_{z_1}+\nu\p_{z_2}+\mu v^j\p_{v^j}\rangle$
and $\langle\p_{z_2}+\mu v^j\p_{v^j}\rangle$,
where $\nu,\mu\in\mathbb R$.

Any subalgebra from the first family $\langle \p_{z_1}+\nu\p_{z_2}+\mu v^j\p_{v^j}\rangle$
gives rise to the ansatz $v={\rm e}^{\mu z_1}\varphi(\omega)$ with $\omega=z_2-\nu z_1$,
which reduces the system~\eqref{eq:RedSys13Homogen} to
\begin{gather}\label{eq:RedSys13CoupledShiftInv}
\chi^2\hat{\mathsf C}\varphi_{\omega\omega\omega}+\chi^2\mu\varphi_{\omega\omega}
+\hat{\mathsf B}\varphi_\omega+\mu\mathsf F\varphi=0,
\end{gather}
where $\hat{\mathsf C}=\mathsf C-\nu\mathsf E$ and $\hat{\mathsf B}:=\hat{\mathsf C}\mathsf F-\mathop{\rm diag}\mathsf Fc$.
The derived system~\eqref{eq:RedSys13CoupledShiftInv} is a homogeneous linear system
of third-order ordinary differential equations.
Since finding its general solution depends on solving an eigenvalue problem for a certain $3m$ by $3m$
matrix constituted by components of $\chi^2\hat{\mathsf C}$, $\hat{\mathsf B}$, $\chi^2\mu\mathsf E$ and $\mu\mathsf F$,
we will not further consider this system.
However, the corresponding solution of~\eqref{eq:RedSys13Homogen} is given by
\begin{gather}\label{eq:Reduction13CoupledCaseShiftInvASol}
\solutionRedEq
v={\rm e}^{\mu z_1}\varphi,
\end{gather}
where $\varphi=\varphi(z_2-\nu z_1)$ is an arbitrary solution of~\eqref{eq:RedSys13CoupledShiftInv},
$c$ is an arbitrary constant $m$-tuple with $c\in\im\mathsf F\setminus\{0\}$,
$\mu$, $\chi$ and $\nu$ are arbitrary constants with $\chi\ne0$.

Reduction using the subalgebra family $\langle\p_{z_2}+\mu v^j\p_{v^j}\rangle$
leads to~\eqref{eq:Reduction13CoupledCase},
which is now a homogeneous linear system of first-order ordinary differential equations with constant coefficients.
Thus, its general solution can be presented in terms of the matrix exponential.
Its explicit representation depends on whether the modified matrix $\tilde{\mathsf F}=\mathsf F+\mu^2\chi^2\mathsf E$
is invertible or not.

The matrix $\tilde{\mathsf F}$ is invertible if and only if $\lambda_i\ne-\mu^2\chi^2$ for all $i=1,\dots,m$.
In this case, the solution~\eqref{eq:Reduction13CoupledGeneralCaseSolution} takes the form
\begin{gather}\label{eq:Reduction13CoupledCaseShiftInvBSol}
\solutionRedEq
v={\rm e}^{\mu z_1}\tilde{\mathsf F}^{-1}\exp\big(-\mu(z_1-\nu z_2)(\mathsf C-\tilde{\mathsf B}\tilde{\mathsf F}^{-1})\big)A,
\end{gather}
where $\tilde{\mathsf F}:=\mathsf F+\mu^2\chi^2\mathsf E$,
$\tilde{\mathsf B}:=\mathop{\rm diag}\mathsf Fc$,
$\mathsf C:=\mathop{\rm diag}c$,
$c$ and $A$ are tuples of arbitrary constants with $c\in\im\mathsf F\setminus\{0\}$,
$\mu$ and $\chi$ are arbitrary constants.

Since eigenvalues of $\tilde{\mathsf F}$ are pairwise different,
the matrix $\tilde{\mathsf F}$ is degenerate if and only if
there exists~$i$ such that the eigenvector $e_i$ of~$\tilde{\mathsf F}$ spans the kernel of $\tilde{\mathsf F}$,
$\ker\tilde{\mathsf F}=\langle e_i\rangle$.
This allows us to write down the corresponding solutions of the system~\eqref{eq:RedSys13Homogen} as follows:
\begin{gather}\label{eq:Reduction13CoupledCaseShiftInvCSol}
\solutionRedEq
v={\rm e}^{\mu z_1}\tilde{\mathsf F}^+
\exp\big(\mu(z_2-\nu z_1)(\mathsf C-\tilde{\mathsf B}\tilde{\mathsf F}^+)\big)A+g(t)e_i,
\end{gather}
where $\tilde{\mathsf F}:=\mathsf F+\mu^2\chi^2\mathsf E$,
$\tilde{\mathsf B}:=\mathop{\rm diag}\mathsf Fc$,
$\mathsf C:=\mathop{\rm diag}c$,
$c$ and $A$ are tuples of arbitrary constants with $c\in\im\mathsf F\setminus\{0\}$,
$\mu$ and $\chi$ are arbitrary constants,
$g(t)$ is an arbitrary smooth function of~$t$ if $\tilde{\mathsf B}e_i=0$ and
$g=0$ otherwise,
and the Moore--Penrose inverse~$\tilde{\mathsf F}^+$ of $\tilde{\mathsf F}$
is defined in the same way as that for~$\mathsf F$ in~\eqref{eq:MNInverse}.
Substituting the solutions~\eqref{eq:Reduction13CoupledCaseShiftInvASol}, \eqref{eq:Reduction13CoupledCaseShiftInvBSol}
and \eqref{eq:Reduction13CoupledCaseShiftInvCSol} into~\eqref{eq:Ansatz13Mod},
we obtain the associated representation of solutions of the original system~\eqref{eq:mLaysMod}.

\section{Codimension-two Lie reductions}\label{sec:CodimTwoLieReds}

Among the subalgebras listed in Theorem~\ref{thm:2D-subalgebras},
only the first, the second, the fourth and the sixth families
are appropriate for the Lie reduction procedure.
The corresponding reductions remain important to consider
despite the fact that they can be interpreted as two-step Lie reductions of the system~\eqref{eq:mLaysMod}
with involving induced Lie symmetries,
where the intermediate codimension-one reduced systems are
linear systems of partial differential equations
studied in Sections~\ref{sec:Reduction12} and~\ref{sec:Reduction13}.\looseness=-1

More specifically, the two-dimensional abelian subalgebras
$\mathfrak s_{2.1}^{abc}$,
$\mathfrak s_{2.2}^{abc}$ and
$\mathfrak s_{2.6}^{\chi bc\kappa}$ of $\mathfrak g$
contain one-dimensional subalgebras
$\mathfrak s_{1.2}^{ac0}$,
$\mathfrak s_{1.3}^{1c0}$ and
$\mathfrak s_{1.3}^{1 c\kappa}$.
This means that Lie reductions with respect to these two-dimensional subalgebras
are two-step reductions of the system~\eqref{eq:mLaysMod}
with the intermediate codimension-one reduced systems~\eqref{eq:RedSystSubalgS12MatrixHomog},
\eqref{eq:RedSys13Homogen} and~\eqref{eq:RedSys13Homogen}
and the final reduced systems~%
\eqref{eq:Reduction12CoupledCaseShiftInvA} with $\chi=a$, $\mu=\nu=0$,
\eqref{eq:RedSys13CoupledShiftInv} with $\chi=1$, $\mu=0$, $\nu=a$ and~%
\eqref{eq:Reduction13CoupledCase} with $\mu=0$, respectively.
However, since we have exhaustively integrated not all codimension-one reduced systems,
the consideration of the above codimension-two reductions is relevant,
leading to the derivation of new closed-form exact solutions of~\eqref{eq:mLaysMod}
supplementing those presented in Section~\ref{sec:CodimOneLieReds}.

The Lie reduction with respect to the subalgebra~$\mathfrak s_{2.4}^{ab\sigma}$
is of a different nature.
Although this subalgebra is nonabelian,
it contains the one-dimensional ideal $\mathfrak s_{1.3}^{\chi0\kappa}$
with $\chi={\rm e}^{\sigma t}$ and $\kappa=at{\rm e}^{\sigma t}$.
Therefore, the associated Lie reduction is just a two-step reduction of the system~\eqref{eq:mLaysMod},
where the intermediate codimension-one reduced system is the completely decoupled system~\eqref{eq:RedSys13Homogen},
which has been exhaustively integrated in Section~\ref{sec:Reduction13DecoupledCase}.
Nevertheless, for the sake of completeness of the presentation,
we examine this codimension-two reduction of~\eqref{eq:mLaysMod} as well
and construct the corresponding closed-form solutions of the system~\eqref{eq:mLaysMod}.

In this section,~$\varphi^1$, \dots, $\varphi^m$ denote the new dependent variables
of the new independent variable $\omega$, $\varphi:=(\varphi^1,\dots,\varphi^m)^{\mathsf T}$.

\subsection{Subalgebra family 2.1}\label{sec:ReductionS21}

An ansatz associated with the two-dimensional subalgebra
$\mathfrak s_{2.1}^{abc}:=\langle\mathcal P^t+b_k\mathcal J^k,\,\mathcal P^y+\mathcal P^x(a)+c_k\mathcal J^k\rangle$
of~$\mathfrak g$, where $a$, $b:=(b_1,\dots,b_m)^{\mathsf T}$, $c:=(c_1,\dots,c_m)^{\mathsf T}$
are arbitrary constants with $b\in\mathop{\rm im}\mathsf F$,
is $\psi^i=\varphi^i(\omega)+c_iy+b_it$
or, in the vector form,
\[
\psi=\varphi(\omega)+yc+tb\quad\mbox{with}\quad\omega:=x-ay.
\]
This ansatz reduces the system~\eqref{eq:mLaysMod} to
\begin{gather*}
\begin{split}
&c_i(1+a^2)\varphi^i_{\omega\omega\omega}
-f_{i,i-1}(c_{i-1}\varphi^i_\omega-c_i\varphi^{i-1}_\omega-b_i+b_{i-1})
\\
&\quad{}-f_{i,i+1}(c_{i+1}\varphi^i_\omega-c_i\varphi^{i+1}_\omega-b_i+b_{i+1})
-\beta\varphi^i_\omega=0,
\end{split}
\end{gather*}
which can be represented using the matrix notation
$\mathsf C:=\mathop{\rm diag}c$
and $\hat{\mathsf C}:=\mathsf C\mathsf F-\mathop{\rm diag}(\mathsf Fc+\beta\bar1)$ as
\begin{gather}\label{eq:RedSyst21Matrix}
(1+a^2)\mathsf C\varphi_{\omega\omega\omega}+\hat{\mathsf C}\varphi_\omega-\mathsf Fb=0.
\end{gather}

Let $K$ and $I$ denote the kernel and the image of the matrix $\mathsf C$ viewed as a linear operator
on the space $\mathbb R^m$, respectively,
$K:=\ker\mathsf C$ and $I:=\mathop{\rm im}\mathsf C$.
Since $\mathsf C$ is a diagonal matrix, the description of subspaces $K$ and $I$ is evident
and we have the decomposition $\mathbb R^m=K\oplus I$.
Denote by~$\pi_K$ and~$\pi_I$ the natural projections onto subspaces $K$ and $I$
relative to that decomposition, respectively.
We have that $\hat{\mathsf C}I\subset I$, i.e., $I$ is an invariant subspace of $\hat{\mathsf C}$.
At the same time, $\hat{\mathsf C}K\cap I\ne\{0\}$ if $K\ne\{0\}$.
By the definition of $\hat{\mathsf C}$ we have
$\pi_K\hat{\mathsf C}=-\pi_K\mathop{\rm diag}(\mathsf Fc+\beta\bar1)$.
This allows us to split the system~\eqref{eq:RedSyst21Matrix}
into system of equations on components $\varphi^0:=\pi_K(\varphi)$ and $\varphi^1:=\pi_I(\varphi)$
of $\varphi$,
\begin{gather}\label{eq:Red21Kernel}
\mathop{\rm diag}(\mathsf Fc+\beta\bar1)\varphi^0_\omega=-\pi_K\mathsf Fb,
\\\label{eq:Red21Image}
(1+a^2)\mathsf C\varphi^1_{\omega\omega\omega}+\hat{\mathsf C}\varphi^1_\omega
=\pi_I(\mathsf Fb-\hat{\mathsf C}\varphi^0_\omega).
\end{gather}
By the Rouch\'e--Capelli theorem, the system~\eqref{eq:Red21Kernel} is consistent
if and only if the rank of its coefficient matrix $\mathop{\rm diag}(\mathsf Fc+\beta\bar1)$
is equal to the rank of its augmented matrix
${\big(\mathop{\rm diag}(\mathsf Fc+\beta\bar1)\mid\pi_K\mathsf Fb\big)}$.
In this case, the system~\eqref{eq:Red21Kernel} easily integrates,
\begin{gather}\label{eq:Red21SolutionKernel}
\solutionRedEq
\varphi^0(\omega)=-\omega\big(\mathop{\rm diag}(\mathsf Fc+\beta\bar1)\big)^+\pi_K\mathsf Fb+g(\omega)+B,
\end{gather}
where $g(\omega)$ is a tuple of arbitrary functions that takes values in
$K\cap\ker\mathop{\rm diag}(\mathsf Fc+\beta\bar1)$,
$B$ is a tuple of integration constants that can be assumed to belong to
$K\cap\mathop{\rm im}\mathop{\rm diag}(\mathsf Fc+\beta\bar1)$
and \smash{$\big(\mathop{\rm diag}(\mathsf Fc+\beta\bar1)\big)^+$}
is the Moore--Penrose inverse of $\mathop{\rm diag}(\mathsf Fc+\beta\bar1)$.

Let $\mathsf C^+$ be the Moore--Penrose inverse of $\mathsf C$.
The matrix of the operator $(1+a^2)^{-1}\mathsf C^+\hat{\mathsf C}|_I$
in the restriction of the standard basis $\mathbb R^m$ to~$I$ is represented by a block-diagonal matrix,
and each of its blocks is a tridiagonal matrix with positive off-diagonal entries.
Therefore, this matrix is diagonalizable.
Denote by $\nu_1$, \dots, $\nu_l$ its eigenvalues, by $\tilde e_1$, \dots,~$\tilde e_l$ the corresponding eigenvectors, where $l:=\dim I$,
and by~$\mathsf R$ the transition matrix from the above standard basis of the subspace~$I$
to the eigenbasis $(\tilde e_1,\dots,\tilde e_l)$.
Consider the operators
\[
\mathsf H:=(1+a^2)^{-1}\mathsf R^{-1}\mathsf C^+\pi_I\quad\mbox{and}\quad
\mathsf B:=\mathsf F+\hat{\mathsf C}\big(\mathop{\rm diag}(\mathsf Fc+\beta\bar1)\big)^+\pi_K\mathsf F
\]
and define $\tilde h:=\mathsf H(\mathsf Bb+\hat{\mathsf C}g_\omega)$.
Using the change of dependent variables $\tilde\varphi^1=\mathsf R^{-1}\varphi^1$,
we decouple the system~\eqref{eq:Red21Image} as follows:
\[
\tilde\varphi^{1s}_{\omega\omega\omega}+\nu_s\tilde\varphi^{1s}_\omega=\tilde h^s,\quad s\in\{1,\dots,l\}.
\]
The general solution of each equation from this system
depends on the sign of $\tilde\nu_i$ and it can be written as follows:
\begin{gather}\label{eq:Red21SolutionImage}
\begin{split}
\solutionRedEq\nu_s=0\colon\quad\tilde\varphi^{1s}(\omega)
={}&A_{0s}+A_{1s}\omega+A_{2s}\omega^2+\frac{\omega^3}6\mathsf H\mathsf Bb+\mathsf H\hat{\mathsf C}G,
\\[1.5ex]
\nu_s<0\colon\quad\tilde\varphi^{1s}(\omega)
={}&A_{0s}+A_{1s}{\rm e}^{\sqrt{-\nu_s}\omega}+A_{2s}{\rm e}^{-\sqrt{-\nu_s}\omega}\\
&{}+\frac\omega{\nu_s}\mathsf H\mathsf Bb
+\frac1{\nu_s}\mathsf H\hat{\mathsf C}
\int_{\omega_0}^\omega\big(1-\cosh(\sqrt{-\nu_s}(\omega-\varsigma))\big)g_\varsigma(\varsigma)\,{\rm d}\varsigma,
\\[1.8ex]
\nu_s>0\colon\quad\tilde\varphi^{1s}(\omega)
={}&A_{0s}+A_{1s}\sin(\sqrt{\nu_s}\omega)+A_{2s}\cos(\sqrt{\nu_s}\omega)\\
&{}+\frac\omega{\nu_s}\mathsf H\mathsf Bb
+\frac1{\nu_s}\mathsf H\hat{\mathsf C}
\int_{\omega_0}^\omega\big(1-\cos(\sqrt{\nu_s}(\omega-\varsigma))\big)g_\varsigma(\varsigma)\,{\rm d}\varsigma,
\end{split}
\end{gather}
where $A_{0s}$, $A_{1s}$ and $A_{2s}$ are arbitrary constants,
and $G$ is a tuple of second antiderivatives of the tuple~$g$, $G_{\omega\omega}=g$.
Pulling the obtained solution back, we obtain the solution of the original system~\eqref{eq:mLaysMod}
\begin{gather*}
\solution
\psi=\sum_{s=1}^l\tilde\varphi^{1s}(x-ay)\tilde e_s+\varphi^0(x-ay)+tb+yc,
\end{gather*}
where
the functions  $\varphi^0$ and $\varphi^{1s}$ are defined in~\eqref{eq:Red21SolutionKernel} and~\eqref{eq:Red21SolutionImage}, respectively,
$\tilde e_1$, \dots,~$\tilde e_l$ are eigenvectors of the matrix $(1+a^2)^{-1}\mathsf C^+\hat{\mathsf C}|_I$
corresponding to its eigenvalues $\nu_1$, \dots, $\nu_l$,
$\mathsf C^+$ is the Moore--Penrose inverse of $\mathsf C:=\mathop{\rm diag}c$, $I:=\im\mathsf C$,
$\hat{\mathsf C}:=\mathsf C\mathsf F-\mathop{\rm diag}(\mathsf Fc+\beta\bar1)$,
$b=(b_1,\dots,b_m)^{\mathsf T}$ and $c=(c_1,\dots,c_m)^{\mathsf T}$
are arbitrary constant tuples with $b\in\im\mathsf F$.

\subsection{Subalgebra family 2.2}

As associated with the two-dimensional subalgebra
$\mathfrak s_{2.2}^{abc}:=\langle
\mathcal P^t+a\mathcal P^y+b_k\mathcal J^k,\,
\mathcal P^x(1)+c_k\mathcal J^k
\rangle$ of~$\mathfrak g$, where
$b:=(b_1,\dots,b_m)^{\mathsf T}$, $c:=(c_1,\dots,c_m)^{\mathsf T}$
and~$a$ are arbitrary constants with $\sigma\ne0$ and $b\in\mathop{\rm im}\mathsf F$,
we choose the ansatz $\psi^i=\varphi^i(\omega)+c_ix+b_it$ with $\omega=y-at$.
The corresponding reduced system then is
\begin{gather*}
(c_i-a)\varphi^i_{\omega\omega\omega}
+f_{i,i-1}((c_i-a)\varphi^{i-1}_\omega-(c_{i-1}-a)\varphi^i_\omega+b_{i-1}-b_i)
\\
\qquad{}-f_{i,i+1}((c_{i+1}-a)\varphi^i_\omega-(c_i-a)\varphi^{i+1}_\omega+b_i-b_{i+1})
+c_i\beta=0.
\end{gather*}
In the vector notation, the ansatz and the reduced system respectively take the form
\begin{gather}\nonumber
\psi=\varphi(\omega)+xc+tb\quad\mbox{with}\quad\omega:=y-at,
\\\label{eq:RedSystSubalgS22MatrixForm}
\mathsf C\varphi_{\omega\omega\omega}+\hat{\mathsf C}\varphi_\omega+\mathsf Fb+\beta c=0,
\end{gather}
where $\mathsf C:=\mathop{\rm diag}(c-a\bar 1)$ and $\hat{\mathsf C}:=\mathsf C\mathsf F-\mathop{\rm diag}\mathsf Fc$.
The form of the reduced system is similar to that of the system~\eqref{eq:RedSyst21Matrix},
therefore, to find the general solution of~\eqref{eq:RedSystSubalgS22MatrixForm} we use the approach from Section~\ref{sec:ReductionS21}.

Since the matrix $\mathsf C$ is diagonal,
we have the direct sum decomposition $\mathbb R^m=K\oplus I$,
where $K:=\ker\mathsf C$ and $I:=\im\mathsf C$.
Denote by $\pi_K$ and $\pi_I$ the natural projections from~$\mathbb R^m$ onto~$K$ and~$I$
according to this decomposition, respectively.
It follows from the definition of the matrix $\hat{\mathsf C}$ that
$\hat{\mathsf C}I\subset I$, $\hat{\mathsf C}K\cap I\ne\{0\}$ if $K\ne\{0\}$,
and $\pi_K\hat{\mathsf C}=-\pi_K\mathop{\rm diag}\mathsf Fc$,
which allows us to decompose the reduced system~\eqref{eq:RedSystSubalgS22MatrixForm}
according to the above decomposition of~$\mathbb R^m$,
\begin{gather}
\label{eq:Red22Kernel}
\mathop{\rm diag}\mathsf Fc\varphi^0_\omega=\pi_K(\mathsf Fb+\beta c),
\\\label{eq:Red22Image}
\mathsf C\varphi^1_{\omega\omega\omega}+\hat{\mathsf C}\varphi^1_\omega=-\pi_I(\mathsf Fb+\beta c+\hat{\mathsf C}\varphi^0_\omega),
\end{gather}
where $\varphi^0:=\pi_K(\varphi)$ and $\varphi^1:=\pi_I(\varphi)$.
The system~\eqref{eq:Red22Kernel} is compatible if and only if
\[\rank\mathop{\rm diag}\mathsf Fc=\rank\big(\mathop{\rm diag}\mathsf Fc\mid\pi_K(\mathsf Fb+\beta c)\big),\]
where $\big(\mathop{\rm diag}\mathsf Fc\mid\pi_K(\mathsf Fb+\beta c)\big)$
is the augmented matrix of the system~\eqref{eq:Red22Kernel}.
In this case, the general solution of~\eqref{eq:Red22Kernel} is given by
\begin{gather}\label{eq:Red22SolutionKernel}
\solutionRedEq
\varphi^0=\omega\big(\mathop{\rm diag}\mathsf Fc\big)^+\pi_K(\mathsf Fb+\beta c)+B+g(\omega),
\end{gather}
where $g(\omega)$ is a tuple of arbitrary functions that takes values in
$K\cap\ker\mathop{\rm diag}\mathsf Fc$,
$B$ is a tuple of integration constants that can be assumed to belong to
$K\cap\mathop{\rm im}\mathop{\rm diag}\mathsf Fc$
and \smash{$\big(\mathop{\rm diag}\mathsf Fc\big)^+$} is the Moore--Penrose inverse of $\mathop{\rm diag}\mathsf Fc$.

It can be shown using the arguments as in Section~\ref{sec:ReductionS21}
that the matrix $\mathsf C^+\hat{\mathsf C}|_I$, where $\mathsf C^+$ is the Moore--Penrose inverse of $\mathsf C$,
is diagonalizable.
Denote by $\nu_1$, \dots, $\nu_l$ its eigenvalues, where $l:=\dim I$,
by $\tilde e_1$,~\dots,~$\tilde e_l$ the corresponding eigenvectors
and by~$\mathsf R$ the transition matrix from the standard basis of the subspace~$I$
to the eigenbasis $(\tilde e_1,\dots,\tilde e_l)$.
Consider the operators
\[
\mathsf H:=\mathsf C^+\pi_I\quad\mbox{and}\quad
\mathsf B:=\mathsf F+\hat{\mathsf C}\big(\mathop{\rm diag}\mathsf Fc\big)^+\pi_K\mathsf F
\]
and define $\tilde h:=-\mathsf H(\mathsf Bb+\beta c+\hat{\mathsf C}g_\omega)$.
The substitution $\varphi^1:=\mathsf R\tilde\varphi^1$ enables decoupling the system~\eqref{eq:Red22Image}
to the system
\[
\tilde\varphi^{1s}_{\omega\omega\omega}+\nu_s\tilde\varphi^{1s}_\omega=\tilde h^s,\quad s\in\{1,\dots,l\}.
\]
The general solution of this system, depending on the sign of $\nu_s$,
is given by
\begin{gather}\label{eq:Red22SolutionImage}
\begin{split}
\solutionRedEq\nu_s=0\colon\quad\tilde\varphi^{1s}(\omega)
={}&A_{0s}+A_{1s}\omega+A_{2s}\omega^2-\frac{\omega^3}6\mathsf H(\mathsf Bb+\beta c)-\mathsf H\hat{\mathsf C}G,
\\[1.5ex]
\nu_s<0\colon\quad\tilde\varphi^{1s}(\omega)
={}&A_{0s}+A_{1s}{\rm e}^{\sqrt{-\nu_s}\omega}+A_{2s}{\rm e}^{-\sqrt{-\nu_s}\omega}\\
&{}-\frac\omega{\nu_s}\mathsf H(\mathsf Bb+\beta c)
-\frac1{\nu_s}\mathsf H\hat{\mathsf C}
\int_{\omega_0}^\omega\big(1-\cosh(\sqrt{-\nu_s}(\omega-\varsigma))\big)g_\varsigma(\varsigma)\,{\rm d}\varsigma,
\\[1.8ex]
\nu_s>0\colon\quad\tilde\varphi^{1s}(\omega)
={}&A_{0s}+A_{1s}\sin(\sqrt{\nu_s}\omega)+A_{2s}\cos(\sqrt{\nu_s}\omega)\\
&{}-\frac\omega{\nu_s}\mathsf H(\mathsf Bb+\beta c)
-\frac1{\nu_s}\mathsf H\hat{\mathsf C}
\int_{\omega_0}^\omega\big(1-\cos(\sqrt{\nu_s}(\omega-\varsigma))\big)g_\varsigma(\varsigma)\,{\rm d}\varsigma,
\end{split}
\end{gather}
where $A_{0s}$, $A_{1s}$ and $A_{2s}$ are arbitrary constants,
and $G$ is a tuple of second antiderivatives of the tuple~$g$, $G_{\omega\omega}=g$.
Pulling the obtained solution back, we construct the solution of the original system~\eqref{eq:mLaysMod}
\begin{gather*}
\solution
\psi=\sum_{s=1}^l\tilde\varphi^{1s}(y-at)\tilde e_s+\varphi^0(y-at)+xc+tb,
\end{gather*}
where
the functions  $\varphi^0$ and $\varphi^{1s}$ are defined in~\eqref{eq:Red22SolutionKernel} and~\eqref{eq:Red22SolutionImage}, respectively,
$\tilde e_1$, \dots,~$\tilde e_l$ are eigenvectors of the matrix $\mathsf C^+\hat{\mathsf C}|_I$
corresponding to its eigenvalues $\nu_1$,~\dots,~$\nu_l$,
$\mathsf C^+$ is the Moore--Penrose inverse of $\mathsf C:=\mathop{\rm diag}(c-a\bar 1)$, $I:=\im\mathsf C$,
$\hat{\mathsf C}:=\mathsf C\mathsf F-\mathop{\rm diag}\mathsf Fc$,
$b=(b_1,\dots,b_m)^{\mathsf T}$ and $c=(c_1,\dots,c_m)^{\mathsf T}$
are arbitrary constant tuples with $b\in\im\mathsf F$.

\subsection{Subalgebra family 2.4}

Consider the two-dimensional subalgebra
$\mathfrak s_{2.4}^{ab\sigma}:=\langle
\mathcal P^t+a\mathcal P^y+b_k\mathcal J^k,\,
\mathcal P^x({\rm e}^{\sigma t})+\mathcal Z(a\sigma t{\rm e}^{\sigma t})
\rangle$  of~$\mathfrak g$,
where $b:=(b_1,\dots,b_m)^{\mathsf T}$, $a$ and~$\sigma$ are arbitrary constants
with $b\in\mathop{\rm im}\mathsf F$ and $\sigma\ne0$.
With this subalgebra, we construct
the ansatz $\psi^i=\varphi^i(\omega)+b_it-\sigma(y-at)x$
or, in the vector form,
\[
\psi=\varphi(\omega)+tb-\sigma(y-at)x\bar1,\quad\mbox{where}\quad \omega:=y-at,
\]
and reduce the original system~\eqref{eq:mLaysMod} to
\begin{gather*} (\sigma\omega+a)\varphi^i_{\omega\omega\omega}
+f_{i,i-1}((\sigma\omega+a)\varphi^{i-1}_\omega-(\sigma\omega+a)\varphi^i_\omega-b_{i-1}+b_i)
\\
\qquad{}-f_{i,i+1}((\sigma\omega+a)\varphi^i_\omega-(\sigma\omega+a)
\varphi^{i+1}_\omega-b_i+b_{i+1})
+\beta\sigma\omega=0,
\end{gather*}
which can also be conveniently represented in the matrix form
\begin{gather}\label{eq:ReductionS24matrix}
\varphi_{\omega\omega\omega}+\mathsf F\varphi_\omega=(\sigma\omega+a)^{-1}(-\sigma\beta\omega\bar1+\mathsf Fb).
\end{gather}
This is an inhomogeneous linear system of constant-coefficient third-order differential equations,
which can be completely decoupled.
More specifically, denote
\begin{gather}\label{eq:RedSys24PartOfPartSol}
\check\varphi(\omega):=\beta\left(-\frac{\omega^3}6
+\frac12\frac a\sigma\left(\omega+\frac a\sigma\right)^2\left(\ln\left(\omega+\frac a\sigma\right)-\frac32\right)
\right)\bar1.
\end{gather}
By the change of the dependent variables $\tilde\varphi=\mathsf P^{-1}(\varphi-\check\varphi)$,
see the end of Section~\ref{sec:PropertiesOfMatrixF} for the notation,
we reduce the system~\eqref{eq:ReductionS24matrix}
to the decoupled system of linear inhomogeneous constant-coefficient third-order differential equations
\[
\tilde\varphi_{\omega\omega\omega}+\mathsf\Lambda\tilde\varphi_\omega=(\sigma\omega+a)^{-1}\mathsf\Lambda\tilde b,
\]
where $\tilde b:=\mathsf P^{-1}b$.
Recall that $\mathsf\Lambda:=\mathsf P^{-1}\mathsf F\mathsf P=\mathop{\rm diag}(\lambda_1,\dots,\lambda_m)$,
$\lambda_1$, \dots, $\lambda_m$ are the eigenvalues of~$\mathsf F$ with $\lambda_1<\dots<\lambda_m$,
$\lambda_m=0$ and thus all other~$\lambda_i$ are negative.
The general solution of the decoupled system is given by
\begin{gather}\label{eq:RedSys24DecoupledSolA}
\begin{split}\solutionRedEq
\tilde\varphi^i(\omega)=&
 \frac{\lambda_i\tilde b_i}2{\rm e}^{-\kappa_i\omega}
\int\frac{{\rm e}^{\kappa_i\omega}}{\sigma\omega+a}{\rm d}\omega
-\frac{\lambda_i\tilde b_i}2{\rm e}^{\kappa_i\omega}
\int\frac{{\rm e}^{-\kappa_i\omega}}{\sigma\omega+a}{\rm d}\omega
+\frac{\lambda_i\tilde b_i}\sigma\ln|\sigma\omega+a|
\\&{}
+A_i{\rm e}^{-\kappa_i\omega}+B_i{\rm e}^{\kappa_i\omega}+C_i,
\qquad \kappa_i:=\sqrt{-\lambda_i},\qquad i=1,\dots,m-1,
\end{split}
\\\label{eq:RedSys24DecoupledSolB}
\tilde\varphi^m(\omega)=A_m\omega^2+B_m\omega+C_m.
\end{gather}
The corresponding solutions of the original system~\eqref{eq:mLaysMod} are of the form
\[
\solution
\psi=\sum_{i=1}^{m-1}\tilde\varphi^i(y-at)e_i+\tilde\varphi^m(y-at)\bar1+\check\varphi(y-at)+tb-\sigma(y-at)x\bar1,
\]
where $e_1$, \dots, $e_{m-1}$, $e_m:=\bar1$ are eigenvectors of the matrix~$\mathsf F$
for its eigenvalues~$\lambda_1$, \dots, $\lambda_{m-1}$, $\lambda_m=0$,
which are the columns of the matrix~$\mathsf P$,
the functions~$\check\varphi$, $\tilde\varphi^i$, $i=1,\dots,m-1$, and~$\tilde\varphi^m$
are defined in~\eqref{eq:RedSys24PartOfPartSol},
\eqref{eq:RedSys24DecoupledSolA} and~\eqref{eq:RedSys24DecoupledSolB}, respectively,
$b=(b_1,\dots,b_m)$, $\sigma$ and~$a$ are arbitrary constants with $\sigma\ne0$,
and $\tilde b:=\mathsf P^{-1}b$.

\subsection{Subalgebra family 2.6}\label{sec:ReductionS26}

An ansatz constructed with the subalgebra
$\mathfrak s_{2.6}^{\chi bc\kappa}:=\langle
\mathcal P^y+\mathcal P^x(\chi)+b_k\mathcal J^k,\,
\mathcal P^x(1)+c_k\mathcal J^k+\mathcal Z(\kappa)\rangle$
with arbitrary constant tuples
$b:=(b_1,\dots,b_m)^{\mathsf T}$ and $c:=(c_1,\dots,c_m)^{\mathsf T}$ from $\mathop{\rm im}\mathsf F$
and arbitrary smooth functions~$\chi$ and~$\kappa$ of $t$,
is $\psi^i=\varphi^i(\omega)+yb_i+(x-\chi y)c_i+\kappa(x-\chi y)-\frac12\chi_ty^2$
or, in the vector notation,
\[
\psi=\varphi(\omega)+yb+(x-\chi y)c+\big(\kappa(x-\chi y)-\tfrac12\chi_ty^2\big)\bar 1,
\quad\mbox{where}\quad\omega:=t.
\]
The corresponding reduced system is
\begin{gather}\label{eq:RedSyst26}
\begin{split}
&f_{i,i-1}(\varphi^{i-1}_\omega-\varphi^i_\omega+\kappa(b_{i-1}-b_i)+b_{i-1}c_i-b_ic_{i+1})
\\
&-f_{i,i+1}(\varphi^i_\omega-\varphi^{i+1}_\omega+\kappa(b_i-b_{i+1})+b_ic_{i+1}-b_{i+1}c_i)
+\beta\kappa+\beta c_i-\chi_{tt}=0.
\end{split}
\end{gather}
It can also be more conveniently represented in the matrix form
\begin{gather*}
\mathsf F\varphi_\omega+(\mathsf C+\kappa\mathsf E)\mathsf Fb-\mathsf B\mathsf Fc+\beta c
+(\beta\kappa-\chi_{\omega\omega})\bar 1=0,
\end{gather*}
where $\mathsf B:=\mathop{\rm diag}b$ and $\mathsf C:=\mathop{\rm diag}c$.
This is a system of linear first-order ordinary differential equations
with the degenerate coefficient matrix~$\mathsf F$ of the first derivative tuple~$\varphi_\omega$.
Taking the $\mathsf W$-weighted inner product of the system with~$\bar 1$,
we derive $(\beta\kappa-\chi_{\omega\omega})(\bar1,\bar1)_{\mathsf W}^{}=0$
since $\im\mathsf F\perp_{\mathsf W}^{}\langle\bar 1\rangle$ and
$(\bar1,\mathsf C\mathsf Fb-\mathsf B\mathsf Fc+\beta c)_{\mathsf W}^{}=0$.
Hence the reduced system is consistent if and only if $\beta\kappa-\chi_{\omega\omega}=0$, i.e.,
$\kappa=\beta^{-1}\chi_{\omega\omega}$, and its general solution is
\[
\varphi=-\beta^{-1}\chi_\omega b-\omega\mathsf F^+(\mathsf C\mathsf Fb-\mathsf B\mathsf Fc+\beta c)+\varsigma+g\bar1,
\]
where
$\varsigma$ is an arbitrary constant tuple from $\mathop{\rm im}\mathsf F$,
$g$ is an arbitrary smooth function of~$\omega$ and
the Moore--Penrose inverse~$\mathsf F^+$ of~$\mathsf F$ is given by~\eqref{eq:MNInverse}.
Pulling this solution back to a solution of the original system~\eqref{eq:mLaysMod} by the above ansatz,
we obtain
\begin{gather}\label{eq:SolObtainedFromRed26}
\begin{split}
\solution
\psi={}&(x-\chi y)c+(y-\beta^{-1}\chi_t)b-t\mathsf F^+(\mathsf C\mathsf Fb-\mathsf B\mathsf Fc+\beta c)+\varsigma
\\&
+\left(\frac{\chi_{tt}}{\beta}(x-\chi y)-\frac{\chi_t}2y^2+g\right)\bar 1,
\end{split}
\end{gather}
where
$b$, $c$ and~$\varsigma$ are arbitrary constant tuples from $\im\mathsf F$,
$\mathsf B:=\mathop{\rm diag}b$ and $\mathsf C:=\mathop{\rm diag}c$,
$\chi$ and~$g$ are arbitrary smooth functions of $t$
and the Moore--Penrose inverse~$\mathsf F^+$ of~$\mathsf F$ is given by~\eqref{eq:MNInverse}.
The presence of the arbitrary function~$g$ and the arbitrary constant tuple~$\varsigma$ from $\im\mathsf F$
is explained by the invariance of the system~\eqref{eq:mLaysMod}
with respect to the point transformations~$\mathscr Z(g)$ and $\mathscr J^i(\varsigma_i)$,
and thus these function and constant tuple can be set to zero up to the $G$-equivalence.

\begin{remark}
A solution of the multi-layer problem~\eqref{eq:mLaysMod}
gives rise to the velocity fields that depend only on~$t$ in all layers
if and only if it is of the form
\begin{gather}\label{eq:MultilayerLinInXY}
\psi=(x-\zeta^2)c+(y+\zeta^1)b+(\zeta^1_tx+\zeta^2_ty+g)\bar1
-t\mathsf F^+(\mathsf C\mathsf Fb-\mathsf B\mathsf Fc+\beta c),
\end{gather}
where the constant tuples $b$ and $c$ are arbitrary elements of $\im\mathsf F$,
$\mathsf B:=\mathop{\rm diag}b$ and $\mathsf C:=\mathop{\rm diag}c$,
the Moore--Penrose inverse~$\mathsf F^+$ of~$\mathsf F$ is given by~\eqref{eq:MNInverse},
and $\zeta^1$, $\zeta^2$ and $g$ are arbitrary functions of~$t$.
The functions $\zeta^2$ and $g$ can be set to be equal to zero up to the $G$-equivalence, i.e.,
more specifically, using the pushforwards by the transformations~$\mathscr P^x(\zeta^2)$ and~$\mathscr Z(-g)$.
The vector fields associated with the baroclinic components of each solution
of the form~\eqref{eq:MultilayerLinInXY} are constant in all layers.
The intersection of the families~\eqref{eq:SolObtainedFromRed26} and~\eqref{eq:MultilayerLinInXY}
is singled out from them by the constraints $\chi_t=0$ and $\zeta^1_t=\zeta^2_t=0$, respectively,
and it consists of the solutions of the multi-layer problem~\eqref{eq:mLaysMod}
with constant velocity fields in all layers.
\end{remark}

\section{Conclusion}\label{sec:Conclusion}

In the present paper, we have systematically carried out symmetry analysis
of the multi-layer quasi-geostrophic problem~\eqref{eq:mLaysMod},
significantly extending and generalizing results of~\cite{bihl2011a,moro2023a}.
In the course of this study, we have successfully addressed two primary challenges.
Firstly, the system~\eqref{eq:mLaysMod} consists of an arbitrary number $m\in\mathbb N$ of equations,
which makes it impossible to apply specialized computer algebra packages
to find the maximal Lie invariance algebra~$\mathfrak g$ of this system.
Secondly, the system~\eqref{eq:mLaysMod} is nontrivially coupled,
and this coupling is described by the tridiagonal vertical coupling matrix~$\mathsf F$.
This in turn complicates studying Lie reductions
and constructing closed-form exact solutions of~\eqref{eq:mLaysMod}.
Below, we summarize the results of this paper as well as our approaches to resolving the above challenges.

The multi-layer quasi-geostrophic model can be represented in
the componentwise form~\eqref{eq:mLaysMod} and the vector form~\eqref{eq:mLaysModMatrixForm}.
Each representation has its own application scope:
the componentwise form is more convenient for computing the algebra~$\mathfrak g$
and the point-symmetry pseudogroup~$G$ of this model,
while the vector form facilitates the study of conservation laws, Hamiltonian structure and Lie reductions
of~\eqref{eq:mLaysMod}.
Spectral properties of the coupling matrix $\mathsf F$ play an important role
in the course of this study.
It turns out that the matrix $\mathsf F$ is diagonalizable and is of rank $m-1$,
its nonzero eigenvalues are negative real
and the eigenbasis of $\mathsf F$  allows us to distinguish barotropic and baroclinic modes.

In Section~\ref{sec:HamiltonianStructure}
we have extended the well-known results on local conservation laws and a Hamiltonian structure
of the $\beta$-plane vorticity equation~\cite{shep1990a} to the general case of the system~\eqref{eq:mLaysMod},
which results in the first correct presentation of families of its local conservation laws
and its Hamiltonian structure, including a Hamiltonian, a Hamiltonian operator
and the Casimir (distinguished) functionals for this operator.
All the presented conserved functionals have a natural physical or mathematical interpretation.
Details on the construction of the above object will be presented in another paper jointly with the proof
that the listed conservation-law characteristics span the entire space
of conservation-law characteristics of~\eqref{eq:mLaysMod} up to order two.
Note that this construction is essentially based on using the weighted inner product
with the weight matrix~$\mathsf W$ related to the vertical coupling matrix~$\mathsf F$;
in this inner product, the matrix~$\mathsf F$ is symmetric.
We also conjecture that the system~\eqref{eq:mLaysMod}
admits no other independent local conservation laws.

It was evident that deriving the determining equations for the Lie-symmetry vector fields
of the system~\eqref{eq:mLaysMod} via the direct application of the Lie invariance criteria
would involve tedious and cumbersome computations.
This is why we have used two tricks to overcome this challenge.
The first trick consists in the formal replacement of the spatial independent variables $(x,y)$
by the complex conjugate variables $z=x+{\rm i}y$ and $\bar z=x-{\rm i}y$,
which maps the Laplacian $\psi^i_{xx}+\psi^i_{yy}$ to $4\psi^i_{z\bar z}$
and the original system~\eqref{eq:mLaysMod} to the system~\eqref{eq:mLaysModCompl}.
It has allowed us to significantly reduce the size of expressions in the course of the computation and
to canonically distinguish the jet variables $\psi^i_{tz\bar z}$ as the leading ones in the latter system.
The second trick is to introduce the superclasses~$\mathcal V$ and $\hat{\mathcal V}$
of systems of the form~\eqref{eq:GenVorClass} and~\eqref{eq:HatGenVorClass}, respectively,
$\mathcal V\supset\hat{\mathcal V}\supset\mathcal M^{\rm c}$,
where $\mathcal M^{\rm c}$ is the class of systems of the form~\eqref{eq:mLaysModCompl}.
Applying the Lie invariance criterion to systems from the widest class $\mathcal V$,
we have derived the general form of Lie-symmetry vector fields
for systems from this class in Lemma~\ref{lem:LieSymVecFieldClassV}.
We have successively specified this form for the systems
from the classes~$\hat{\mathcal V}$ and~$\mathcal M^{\rm c}$,
which has led to Lemmas~\ref{lem:LieSymVecFieldClassHatV} and~\ref{lem:LieInvAlgMLaysModelCompl},
respectively.
Note that the maximal Lie invariance algebra of any system from the class~$\mathcal M^{\rm c}$
coincides with the algebra~$\mathfrak g^{\rm c}$ obtained in the latter lemma.
Finally, by pulling back the algebra~$\mathfrak g^{\rm c}$
with respect to the transformation $(x,y)\mapsto(z,\bar z)$,
we have found the Lie invariance algebra~$\mathfrak g$
of the multi-layer quasi-geostrophic problem~\eqref{eq:mLaysMod} in Theorem~\ref{thm:mLaysLieInvAlg}.
In addition, we have singled out the subalgebra of Hamiltonian Lie symmetries in~$\mathfrak g$.

After studying the structure of the Lie invariance algebra $\mathfrak g$ and constructing its megaideals,
we have applied the megaideal-based version of the algebraic method
to construct the complete point-symmetry pseudogroup $G$ of~\eqref{eq:mLaysMod}
in Theorem~\ref{thm:mLaysPointSymGroup} and singled out its canonical discrete elements.
Based on this result,
we have computed the equivalence groupoid, the usual equivalence group and the equivalence algebra
of the class~$\mathcal M$ of systems of the form~\eqref{eq:mLaysMod}
and carried out its group classification
in Theorems~\ref{eq:ClassMEquivGroupoid} and \ref{thm:ClassMEquivGroup}, Corollary~\ref{cor:ClassMEquivAlgebra}
and Theorem~\ref{thm:ClassMGroupClassification}, respectively.
The performed analysis has revealed that the kernel Lie invariance algebra of the class~$\mathcal M$
coincides with the algebra~$\mathfrak g$
and that there is no Lie symmetry extension within this class.
We have also described the generalized equivalence group of~$\mathcal M$
and its effective counterpart in Remark~\ref{rem:GeneralizedEquivGroup}.

In Theorems~\ref{thm:1D-subalgebras} and~\ref{thm:2D-subalgebras},
we have shown that there are three and eight families
of $G$-inequivalent one- and two-dimensional subalgebras of the algebra~$\mathfrak g$, respectively.
While each family of one-dimensional subalgebras is appropriate for Lie reduction,
only four out of eight two-dimensional subalgebra families have this property.
The rest of the paper has been devoted to the systematic and comprehensive study of
$G$-inequivalent Lie reductions of the multi-layer quasi-geostrophic problem~\eqref{eq:mLaysMod}.
As mentioned above, for the purpose of studying Lie reductions, the vector form~\eqref{eq:mLaysModMatrixForm}
is more convenient than the original system~\eqref{eq:mLaysMod}.

In Section~\ref{sec:ReductionS11},
using the first subalgebra family from Theorem~\ref{thm:1D-subalgebras},
we have reduced the system~\eqref{eq:mLaysModMatrixForm}
to the nonlinear systems of the form~\eqref{eq:RedSystSubalgS11MatrixForm}.
Their maximal Lie invariance algebras coincide with each other for all admitted values of involved parameters
and are entirely induced by the Lie invariance algebra $\mathfrak g$ of the original system~\eqref{eq:mLaysMod},
see Theorem~\ref{thm:mLaysLieInvAlgRedSystem1.1}.
As a particular case of the corresponding invariant solutions,
which is especially convenient for the study and the richest in the sense of possible physical interpretation,
one naturally singles out the stationary solutions
such that on each, $i$th, layer,
the potential vorticity~$q^i$ affinely depends on the stream function~$\psi^i$.
Fixing this affine dependence leads to a coupled linear system of the form~\eqref{eq:ReducitonS11LinearModons}.
Finding a particular solution of this system
and diagonalizing its coupling matrix $\mathsf F-\mathsf B$, where the matrix~$\mathsf B$ is diagonal,
we reduce the system~\eqref{eq:ReducitonS11LinearModons}
to the decoupled homogeneous system~\eqref{eq:ReducitonS11LinearModonsDecoupled}
of the modified Helmholtz, Laplace or Helmholtz equations.
Among the solutions of the original system~\eqref{eq:mLaysMod} that have been constructed in the above way,
we single out those with velocity fields that are defined on the entire $(x,y)$-plane and bounded
and write down their representation~\eqref{eq:Reduction12BoundedSolutions}
in terms of generalized Herglotz wave functions.
This representation has allowed us to recover
baroclinic Rossby waves, coherent baroclinic eddies and hetons as the simplest typical flow patterns.
It is worth to emphasize that the exact solutions of~\eqref{eq:mLaysMod}
obtained under the same fixed affine dependence
can be linearly superposed, thus generating even larger solution families with more complex behavior,
see, e.g., Figure~\ref{fig:Red11GenLinSuperposition}.
Merging solutions associated with different systems of the form~\eqref{eq:ReducitonS11LinearModons}
on subdomains by imposing interface boundary conditions between these subdomains,
we have also extended the construction of dipolar vortices (modons) from~\cite{lari1976a}
to the multi-layer case.
The above physically relevant solutions have been illustrated using real-world geophysical data given in Section~\ref{sec:NumericalExample} for a three-layered ocean model.

The codimension-one Lie reductions of~\eqref{eq:mLaysModMatrixForm}
with respect to subalgebras from the second family presented in Theorem~\ref{thm:1D-subalgebras}
have led to the homogeneous linear systems of partial differential equations
of the form~\eqref{eq:RedSystSubalgS12MatrixHomog}; see Section~\ref{sec:Reduction12}.
We have carried out the group classification of these reduced systems
and described associated hidden symmetries of the original (nonlinear) system~\eqref{eq:mLaysModMatrixForm}.
It turns out that besides the most valuable transformations of linear superposition of solutions,
there are other nontrivial symmetries of this kind.
On top of that, the group classification has also allowed us to distinguish two important cases
for values of parameters appearing in~\eqref{eq:RedSystSubalgS12MatrixHomog}:
the completely decoupled case~\eqref{eq:RedSystSubalgS12Decoupled} and the generically coupled case.
These two cases are further split into the general, scale- and shift-invariant subcases
according to the presence and the kind of submodel Lie symmetries involving~$t$.
In the completely decoupled case, we were able to find wide families of solutions
in terms of elementary functions, integrals of Whittaker functions, Bessel functions
and the general solutions of Klein--Gordon and linearized Benjamin--Bona--Mahony equations.
The integration of the generically coupled case is more involved,
but we have constructed wide families of exact solutions~\eqref{eq:Reduction12CoupledCaseSolutionFromZ2Invertible}
for the shift-invariant subcase.\looseness=-1

The analysis of Lie reductions with respect to the subalgebras
from the third family given in Theorem~\ref{thm:1D-subalgebras}
follows the pattern of those for the second one
since the structure of the obtained reduced systems~\eqref{eq:RedSys13Homogen}
is similar to the structure of the systems~\eqref{eq:RedSystSubalgS12MatrixHomog}.
The difference worth mentioning has arisen only for the completely decoupled systems
of the form~\eqref{eq:RedSys13Homogen}.
By integrating with respect to the first invariant independent variable~$z_1$,
each of them has been reduced to a system of ordinary differential equations with $z_2$ as its independent variable,
whereas $z_1$ plays the role of a parameter,
and has thus been completely integrated in quadratures.

The main benefit of Lie reductions with respect to one-dimensional subalgebras of $\mathfrak g$
is that
most of the associated reduced systems of partial differential equations in two independent variables
are linear or at least reduce to linear systems under certain differential constraints.
In other words, these systems admit the linear superposition of solutions in total
or for large subsets of their solutions.
However, integrating them may be a challenging problem.
Codimension-two Lie reductions of~\eqref{eq:mLaysModMatrixForm}
lead to systems of ordinary differential equations of order at most three with constant coefficients,
whose integration reduces to solving certain eigenvalue problems.
We have shown that each codimension-two Lie reduction of~\eqref{eq:mLaysModMatrixForm}
is a two-step reduction with a codimension-one Lie reduction as the first step.
At the same time, three out of four families of codimension-two Lie reductions
are more convenient for deriving exact solutions of~\eqref{eq:mLaysModMatrixForm}
than their two-step counterparts from Section~\ref{sec:CodimOneLieReds}.

We would like to emphasize that the presented families of exact solutions
can be further extended by acting with point symmetries of the original model~\eqref{eq:mLaysModMatrixForm},
with point symmetries of related submodels and, in case of linear submodels,
with their linear recursion operators or, more specifically,
with Lie-symmetry operators, see, e.g., \cite{kova2025a,kova2024a,popo2025a}.

The results of the present paper open several avenues for further research.
In particular, while we have constructed families of conservation laws and a Hamiltonian structure
of the multi-layer quasi-geostrophic problem~\eqref{eq:mLaysMod} in Section~\ref{sec:HamiltonianStructure},
their comprehensive description with all required proofs remains an open problem.
A related direction for the study is
to extend results of~\cite{moro2023a} on Lax pairs of the two-layer model to an arbitrary number of layers.
One can also consider generalized and potential symmetries and recursion operators of~\eqref{eq:mLaysMod}.
It is also of interest to apply non-Lie methods of finding exact solutions to~\eqref{eq:mLaysMod},
including conditional symmetries~\cite[Chapter~5]{fush1993A},
the method of differential constraints~\cite{olve1994a} (or ``side conditions''~\cite{olve1986b})
and the ansatz method~\cite{fush1995a}.

\section*{Acknowledgments}
This research was undertaken thanks to funding from the Canada Research Chairs program, the NSERC Discovery Grant program
and the AARMS graduate scholarship.
This research was also supported in part by the Ministry of Education, Youth and Sports of the Czech Republic (M\v SMT \v CR) under RVO funding for I\v C47813059.
ROP expresses his gratitude for the hospitality shown by the University of Vienna during his long staying at the university.
The authors express their deepest thanks to the Armed Forces of Ukraine and the civil Ukrainian people
for their bravery and courage in defense of peace and freedom in Europe and in the entire world from russism.\looseness=-1


\begin{thebibliography}{10}\footnotesize\itemsep=-.3ex

\bibitem{CRChandbook1995V2}
Aksenov A.V., Baikov V.A., Chugunov V.A., Gazizov R.K. and Meshkov A.G.,
{\it CRC handbook of Lie group analysis of differential equations. Vol. 2. Applications in engineering and physical sciences},
CRC Press, Boca Raton, FL, 1995.

\bibitem{anco2002a}
Anco S. and Bluman G.,
Direct construction method for conservation laws of partial differential equations. Part I: Examples of conservation law classifications,
{\it Eur.~J. App. Math.} {\bf 13} (2002), 545--566, arXiv:math-ph/0108023.

\bibitem{anco2002b}
Anco S. and Bluman G.,
Direct construction method for conservation laws of partial differential equations. Part II: General treatment,
{\it Eur.~J. App. Math.} {\bf 13} (2002), 567--585, arXiv:math-ph/0108024.
	
\bibitem{andr1998A}
Andreev V.K., Kaptsov O.V., Pukhnachov V.V. and Rodionov A.A.,
{\it Applications of group-theoretical methods in hydrodynamics},
Kluwer Academic Publishers, Dordrecht, 1998.

\bibitem{bihl2015a}
Bihlo A., Dos Santos Cardoso-Bihlo E. and Popovych R.O.,
Algebraic method for finding equivalence groups,
{\it J.~Phys. Conf. Ser.} {\bf 621} (2015) 012001, arXiv:1503.06487.
	
\bibitem{bihl2020a}
Bihlo A., Poltavets N. and Popovych R.O.,
Lie symmetries of two-dimensional shallow water equations with variable bottom topography,
{\it Chaos} {\bf 30} (2020), 073132, arXiv:1911.02097.

\bibitem{bihl2011a}
Bihlo A. and Popovych R.~O.
Lie symmetry analysis and exact solutions of the quasigeostrophic two-layer problem,
{\it J.~Math. Phys.} {\bf 52} (2011), 033103, arXiv:1010.1542.

\bibitem{bihl2012a}
Bihlo A. and Popovych R.O., Point symmetry group of the barotropic vorticity equation,
in {\it Proceedings of 5th Workshop ``Group Analysis of Differential Equations \& Integrable Systems''
(June 6--10, 2010, Protaras, Cyprus)}, University of Cyprus, Nicosia, 2011 pp. 15--27, arXiv:1009.1523.

\bibitem{bihl2020b}
Bihlo A. and Popovych R.O.,
Zeroth-order conservation laws of two-dimensional shallow water equations with variable bottom topography,
{\it Stud. Appl. Math.} {\bf 145} (2020), 291--321, arXiv:1912.11468.

\bibitem{blum2002A}
Bluman G.W. and Anco S.C.,
\textit{Symmetry and integration methods for differential equations},
Springer, New York, 2002.

\bibitem{blum2010A}
Bluman G.W., Cheviakov A.F. and Anco S.C.,
{\it Applications of symmetry methods to partial differential equations},
Springer, New York, 2010. 

\bibitem{boye1976b}
Boyer C.P., Kalnins E.G. and Miller W., Jr.,
Symmetry and separation of variables for the Helmholtz and Laplace equations,
{\it Nagoya Math.~J.} {\bf 60} (1976), 35--80.
	
\bibitem{boyk2024a}
Boyko V.M., Popovych R.O. and Vinnichenko O.O.,
Point- and contact-symmetry pseudogroups of dispersionless Nizhnik equation,
{\it Commun. Nonlinear Sci. Numer. Simul.} {\bf 132} (2024), 107915,
arXiv:2211.09759.

\bibitem{carm2000a}
Carminati J. and Vu K.,
Symbolic computation and differential equations: Lie symmetries,
{\it J.~Symbolic Comput.}  {\bf 29} (2000), 95--116.

\bibitem{chap2024b}
Chapovskyi Ye.Yu, Koval S.D. and Zhur O.,
Subalgebras of Lie algebras. Example of $\mathfrak{sl}_3(\mathbb R)$ revisited,
33~pp,
arXiv:2403.02554

\bibitem{chev2007a}
Cheviakov A.F.,
GeM software package for computation of symmetries and conservation laws of differential equations,
{\it Comput. Phys. Comm.} {\bf 176} (2007), 48--61.

\bibitem{chev2014a}
Cheviakov A.F.,
Conservation properties and potential systems of vorticity-type equations,
{\it J.~Math. Phys.} {\bf 55} (2014), 033508. 

\bibitem{colt2019A}
Colton D. and Kress R.,
{\it Inverse acoustic and electromagnetic scattering theory},
Springer, Cham, 2019.

\bibitem{cott2020a}
Cotter C., Crisan D., Holm D., Pan W. and Shevchenko I.,
Data assimilation for a quasi-geostrophic model with circulation-preserving stochastic transport noise,
{\it J.~Stat. Phys.} {\bf 179} (2020), 1186--1221,
arXiv:1910.03574.

\bibitem{crow2024a}
Crowe M.N. and Johnson E.R.,
Modon solutions in an N-layer quasi-geostrophic model,
{\it J.~Fluid Mech.} {\bf 994} (2024), R1, arXiv:2404.07718. 

\bibitem{card2013a}
Dos Santos Cardoso-Bihlo E. and Popovych R.O.,
Complete point symmetry group of the barotropic vorticity equation on a rotating sphere,
{\it J.~Engrg. Math.} {\bf 82} (2013), 31--38, arXiv:1206.6919.

\bibitem{card2021a}
Dos Santos Cardoso-Bihlo E. and Popovych R.O.,
On the ineffectiveness of constant rotation in the primitive equations and their symmetry analysis,
{\it Commun. Nonlinear Sci. Numer. Simul.} {\bf 101} (2021), 105885,
arXiv:1503.04168.

\bibitem{flie1980a}
Flierl G.R., Larichev V.D., McWilliams J.C. and Reznik G.M.,
The dynamics of baroclinic and barotropic solitary eddies,
{\it Dyn. Atmos. Oceans} {\bf 5} (1980), 1--41.

\bibitem{fush1995a}
Fushchych W.,
Ansatz'95,
{\it J.~Nonlinear Math. Phys.} {\bf 2} (1995), 216--235.

\bibitem{fush1993A}
Fushchych W.I., Shtelen W.M. and Serov N.I.,
{\it Symmetry analysis and exact solutions of equations of nonlinear mathematical physics},
Dordrecht, Kluwer Academic Publishers, 1993.

\bibitem{hart1961a}
Hartman P. and Wilcox C.,
On solutions of the Helmholtz equation in exterior domains,
{\it Math.~Z.} {\bf 75} (1961), 228--255.

\bibitem{hilg2012A}
Hilgert J. and Neeb K.H.,
\textit{Structure and geometry of Lie groups},
Springer, New York, 2012.
	
\bibitem{hydo2000A}
Hydon P.E.,
{\it Symmetry methods for differential equations. A beginner's guide},
Cambridge University Press, Cambridge, 2000. 
	
\bibitem{hydo2000b}
Hydon P.E.,
How to construct the discrete symmetries of partial differential equations,
{\it Eur. J. Appl. Math.} {\bf 11} (2000), 515--527.

\bibitem{ibra1985A}
Ibragimov N.H.
\emph{Transformation groups applied to mathematical physics},
D. Reidel Publishing, Dordrecht, 1985.

\bibitem{ibra2011A}
Ibragimov N.H. and Ibragimov R.N.,
{\it Applications of Lie group analysis in geophysical fluid dynamics},
World Scientific, Singapore, 2011. 

\bibitem{jaco1962A}
Jacobson N.,
{\it Lie algebras},
Dover Publications, New York, 1962.

\bibitem{kame1986A}
Kamenkovich V.M., Koschlyakov V N. and Monin A.S.,
{\it Synoptic eddies in the ocean},
D.~Reidel, Dordrecht, 1986. 

\bibitem{kame1981a}
Kamenkovich V.M., Larichev V.D. and Kharkov B.V.,
A quasigeostrophic baroclinic model for the analysis of synoptic eddies in the open ocean,
{\it Oceanology}, 21 (1981), 673--679.

\bibitem{kelb2013a}
Kelbin O., Cheviakov A.F. and Oberlack M.,
New conservation laws of helically symmetric, plane and rotationally symmetric viscous and inviscid flows,
{\it J.~Fluid Mech.} {\bf 721} (2013), 340--366.

\bibitem{kizn2000a}
Kizner Z. and Berson D.,
Emergence of modons from collapsing vortex structures on the $\beta$-plane,
{\it J.~Mar. Res.} {\bf 58} (2000), 375--403.

\bibitem{kova2023b}
Koval S.D., Bihlo A. and Popovych R.O.,
Extended symmetry analysis of remarkable (1+2)-dimensional Fokker--Planck equation,
{\it European J. Appl. Math.} {\bf 34} (2023), 1067--1098, arXiv:2205.13526.

\bibitem{kova2025a}
Koval S.D., Dos Santos Cardoso-Bihlo E. and Popovych R.O.,
Surprising symmetry properties and exact solutions of Kolmogorov backward equations with power diffusivity,
{\it Stud. Appl. Math.} {\bf 155} (2025), e70105, arXiv:2407.10356.

\bibitem{kova2024a}
Koval S.D. and Popovych R.O.,
Extended symmetry analysis of (1+2)-dimensional fine Kolmogorov backward equation,
{\it Stud. Appl. Math.} {\bf 153} (2024), e12695, 30 pp., arXiv:2402.08822.

\bibitem{lari1976a}
Larichev V.D. and Reznik G.M.,
Two-dimensional Rossby soliton: an exact solution,
{\it Dokl. Akad. Nauk SSSR} {\bf 231} (1976), no.~5, 12--13.

\bibitem{levi1989a}
Levi D., Nucci M.C., Rogers C. and Winternitz P.,
Group theoretical analysis of a rotating shallow liquid in a rigid container,
{\it J.~Phys.~A} {\bf 22} (1989), 4743--4767.

\bibitem{lore1963a}
Lorenz E.N.,
Deterministic nonperiodic flow,
{\it J.~Atmospheric~Sci.} {\bf 20} (1963), 130--141.
	
\bibitem{mill1977A}
Miller W., Jr.,
\textit{Symmetry and separation of variables},
Addison--Wesley, Reading, Mass.--London--Amsterdam, 1977.
		
\bibitem{moro2023a}
Morozov O., The quasigeostrophic two-layer model: Lax representations and conservation laws,
{\it J.~Geom. Phys.} {\bf 192} (2023), 104954.

\bibitem{nauw2004a}
Nauw J.J., Dijkstra H.A. and Simonnet E.,
Regimes of low-frequency variability in a three-layer quasi-geostrophic ocean model,
{\it J.~Marine Res.} {\bf 62} (2004), 685--720.

\bibitem{olve1994a}
Olver P.J.,
Direct reduction and differential constraints,
{\it Proc. Roy. Soc. London Ser. A} {\bf 444} (1994), 509--523. 
	
\bibitem{olve1993A}
Olver P.J.,
{\it Application of Lie groups to differential equations},
Springer, New York, 2000.

\bibitem{olve1986b}
Olver P.J. and Rosenau P.,
The construction of special solutions to partial differential equations,
{\it Phys. Lett.~A} {\bf 114} (1986), 107--112.

\bibitem{opan2017a}
Opanasenko S., Bihlo A. and Popovych R.O.,
Group analysis of general Burgers--Korteweg--de Vries equations,
{\it J.~Math. Phys.} {\bf 58} (2017), 081511, arXiv:1703.06932.

\bibitem{opan2020a}
Opanasenko S., Bihlo A., Popovych R.O. and Sergyeyev A.,
Extended symmetry analysis of isothermal no-slip drift flux model,
{\it Phys.~D} {\bf 402} (2020), 132188, arXiv:1705.09277.

\bibitem{opan2020b}
Opanasenko S., Boyko V. and Popovych R.O.,
Enhanced group classification of nonlinear diffusion-reaction equations with gradient-dependent diffusion,
{\it J.~Math. Anal. Appl.} {\bf 484} (2020), 123739, arXiv:1804.08776.
	
\bibitem{opan2020e}
Opanasenko S. and Popovych R.O.,
Generalized symmetries and conservation laws of (1+1)-dimensional Klein--Gordon equation,
{\it J.~Math. Phys.} {\bf 61} (2020), 101515, arXiv:1810.12434.
	
\bibitem{opan2022a}
Opanasenko S. and Popovych R.O.,
Mapping method of group classification,
{\it J.~Math. Anal. Appl.} {\bf 513} (2022), 126209, arXiv:2109.11490. 
	
\bibitem{ovsi1982A}
Ovsiannikov L.V.,
{\it Group analysis of differential equations},
Academic Press, New York, 1982.

\bibitem{pate1975a}
Patera J., Winternitz P. and Zassenhaus H.,
Continuous subgroups of the fundamental groups of physics. I. General method and the Poincar\'e group,
{\it J.~Math. Phys.} {\bf 16} (1975), 1597--1614.
	
\bibitem{pedo1987A}
Pedlosky J.,
\textit{Geophysical Fluid Dynamics},
Springer, New York, 1987.

\bibitem{poly2002A}
Polyanin A.D.,
{\it Handbook of linear partial differential equations for engineers and scientists},
Chapman \& Hall/CRC, Boca Raton, 2002.

\bibitem{popo2025a}
Popovych D.R., Koval S.D. and Popovych R.O.,
Generalized symmetries of remarkable (1+2)-dimensional Fokker--Planck equation,
{\it European J. Appl. Math.} (2025),
doi:\href{https://doi.org/10.1017/S0956792525100107}{10.1017/S0956792525100107},
arXiv:2409.10348.

\bibitem{popo2002b}
Popovych H.V.,
Lie, partially invariant, and nonclassical submodels of Euler equations,
{\it Proc. Inst. Math. NAS Ukraine} {\bf 43} (2002), 178--183.

\bibitem{popo2003a}
Popovych R.O., Boyko V.M., Nesterenko M.O. and Lutfullin M.W.,
Realizations of real low-dimensional Lie algebras,
arXiv:math-ph/0301029v7 (2005) (extended and revised version of paper {\it J.~Phys.~A} {\bf 36} (2003), 7337--7360).

\bibitem{rein2018a}
Reinaud J.N., Sokolovskiy M.A. and Carton X.,
Hetonic quartets in a two-layer quasi-geostrophic flow: V-states and stability,
{\it Phys. Fluids} {\bf 30} (2018), 056602.
	
\bibitem{roth1952a}
Roth W.E.,
The equations $AX-YB=C$ and $AX-XB=C$ in matrices,
{\it Proc. Amer. Math. Soc.} {\bf 3} (1952), 392--396.

\bibitem{shap1995a}
Shapovalov A.V. and Shirokov I.V.,
Noncommutative integration of linear differential equations,
\textit{Theoret. and Math. Phys.} \textbf{104} (1995), 921--934.

\bibitem{shep1990a}
Shepherd T.G.,
Symmetries, conservation laws, and Hamiltonian structure in geophysical fluid dynamics,
{\it Adv. Geophys.} {\bf 32} (1990), 287--338.

\bibitem{ster1975a}
Stern M.E.,
Minimal properties of planetary eddies,
{\it J. Mar. Res.} {\bf 33} (1975), 1--13.

\bibitem{szat2014a}
Szatmari S. and Bihlo A.,
Symmetry analysis of a system of modified shallow-water equations,
{\it Commun. Nonlinear Sci. Numer. Simul.} {\bf 19} (2014), 530--537,
arXiv:1212.5823.
	
\bibitem{tesc2025A}
Teschl G.,
{\it Partial Differential Equations},
Lecture Notes, 2025,
\url{https://www.mat.univie.ac.at/~gerald/ftp/book-pde/pde.pdf}.
	
\bibitem{vall2006A}
Vallis G.K.,
{\it Atmospheric and Oceanic Fluid Dynamics},
Cambridge University Press, Cambridge, 2006.

\bibitem{vane2020b}
Vaneeva O.O., Bihlo A. and Popovych R.O.,
Generalization of the algebraic method of group classification with application to nonlinear wave and elliptic equations,
{\it Commun. Nonlinear Sci. Numer. Simul.} {\bf91} (2020), 105419, arXiv:2002.08939.

\bibitem{vinn2024a}
Vinnichenko O.O., Boyko V.M. and Popovych R.O.,
Lie reductions and exact solutions of dispersionless Nizhnik equation,
{\it Anal. Math. Phys.} {\bf 14} (2024), 82, arXiv:2308.03744.

\end{thebibliography}
\end{document}